\documentclass[superscriptaddress, amsmath,amssymb,
aps, prl, letterpaper, tightenlines, reprint, notitlepages]{revtex4-2}

\usepackage{graphicx}
\usepackage{floatrow}
\usepackage[english]{babel}

\usepackage{dcolumn}
\usepackage{bm}

\usepackage{amsfonts}
\usepackage{dsfont}
\usepackage{bm}
\usepackage{bbold}
\usepackage{xcolor}
\usepackage{lipsum,babel}
\usepackage[normalem]{ulem}

\usepackage[utf8]{inputenc}
\usepackage{amsmath}
\usepackage{physics}
\usepackage{hyperref}
\usepackage{amssymb}
\usepackage{braket}
\usepackage{mathtools}
\usepackage{color}
\usepackage{amsthm}
\usepackage{algorithm}
\usepackage{algorithmic}

\DeclareMathOperator*{\Per}{\text{Per}}
\DeclareMathOperator*{\Haf}{\text{Haf}}

\newcommand{\cor}[1]{{\color{red}{#1}}}
\newcommand{\gbs}[1]{{\color{blue}{#1}}}

\theoremstyle{plain}
\newtheorem{theorem}{Theorem}
\newtheorem{corollary}{Corollary}

\newtheorem{definition}{Definition}

\theoremstyle{definition}

\begin{document}

\title{Complexity phase transition for continuous-variable cluster state}

\author{Byeongseon Go}
\email{gbs1997@snu.ac.kr}
\affiliation{NextQuantum Innovation Research Center, Department of Physics and Astronomy, Seoul National University, Seoul 08826, Republic of Korea}
\author{Hyunseok Jeong}
\email{h.jeong37@gmail.com}
\affiliation{NextQuantum Innovation Research Center, Department of Physics and Astronomy, Seoul National University, Seoul 08826, Republic of Korea}
\author{Changhun Oh}
\email{changhun0218@gmail.com}
\affiliation{Department of Physics, Korea Advanced Institute of Science and Technology, Daejeon 34141, Republic of Korea}

\begin{abstract}
Continuous-variable (CV) cluster states offer a promising platform for large-scale measurement-based quantum computations (MBQC).
However, finite squeezing inevitably introduces Gaussian noise during MBQC.
While fault-tolerant MBQC schemes exist in principle, they require the scalable incorporation of non-Gaussian resources, such as GKP states, which remain experimentally challenging.
Consequently, a central question at this stage is how finite squeezing fundamentally constrains the intrinsic computational power of CV cluster states themselves.
In this work, we address this question by analyzing the classical complexity of measurement-based linear optics (MBLO) implemented with such states, motivated by its near-term feasibility and recent experimental progress.
We develop an explicit MBLO framework and examine how the squeezing level governs the complexity of the classical simulation of the resulting output states.
Specifically, we identify squeezing-level thresholds that delineate classically tractable and intractable regimes, thereby revealing a squeezing-driven complexity phase transition.
These findings advance our understanding of the squeezing resources necessary for meaningful quantum computation in current experimental regimes.
Furthermore, they underscore the critical need to either scale the squeezing level or integrate error-correction schemes to achieve reliable, large-scale quantum computation with CV cluster states.
\end{abstract}


\maketitle


\emph{Introduction.---}
As a resource for measurement-based quantum computation (MBQC)~\cite{briegel2001persistent, raussendorf2001one}, continuous-variable (CV) cluster states~\cite{menicucci2006universal, gu2009quantum} offer a promising avenue for large-scale quantum computation.
In recent years, substantial experimental progress has been made in generating large-scale CV cluster states~\cite{yoshikawa2008demonstration, pysher2011parallel, yokoyama2013ultra, chen2014experimental, yoshikawa2016invited, cai2017multimode, pfister2019continuous, larsen2019deterministic, asavanant2019generation, larsen2021deterministic, asavanant2021time, du2023generation, wang2024chip, roh2025generation, jia2025continuous, lingua2025continuous}.
However, finite squeezing induces inevitable Gaussian-type noise during MBQC~\cite{menicucci2006universal, gu2009quantum, zhang2006continuous, cable2010bipartite, menicucci2011graphical, alexander2014noise, walshe2019robust, menicucci2014fault, su2018correcting, larsen2020architecture, larsen2021fault}.
In principle, incorporating non-Gaussian resources such as GKP states into CV cluster states can suppress this noise and enable fault-tolerant MBQC~\cite{menicucci2014fault, douce2017continuous, fukui2018high, su2018correcting, larsen2021fault, larsen2020architecture, wu2020quantum, ferrini2015optimization, walshe2021streamlined, du2025complete}. 
However, scaling these non-Gaussian resources commensurate with the system size remains a formidable experimental challenge. Consequently, current large-scale demonstrations have primarily focused on cluster-state generation and (noisy) Gaussian-gate implementations~\cite{yoshikawa2008demonstration, pysher2011parallel, yokoyama2013ultra, chen2014experimental, yoshikawa2016invited, cai2017multimode, pfister2019continuous, larsen2019deterministic, asavanant2019generation, larsen2021deterministic, asavanant2021time, du2023generation, wang2024chip, roh2025generation, jia2025continuous, lingua2025continuous, miwa2009demonstration, wang2010toward, ukai2011demonstration, ukai2011demonstrationCP, su2013gate}.

In this context, \textit{measurement-based linear optics} (MBLO)~\cite{alexander2017measurement} on CV cluster states offers a particularly well-motivated setting.
Linear optics (LO) plays a central role in quantum-advantage demonstrations~\cite{aaronson2011computational, hamilton2017gaussian, oh2025recent, zhong2021phase, madsen2022quantum, deng2023gaussian, liu2025robust}, yet real-space LO interferometers face limited scalability due to depth-dependent noise accumulation (e.g., photon loss)~\cite{oszmaniec2018classical, garcia2019simulating, qi2020regimes, oh2025classical, oh2022classical}.
In contrast, given CV cluster states, MBLO can be realized via \textit{parallel} homodyne measurements~\cite{menicucci2006universal, gu2009quantum}, which can mitigate such types of noise.

Furthermore, MBLO naturally aligns with the current experimental regime, in which Gaussian-gate implementations are feasible (up to finite-squeezing noise and displacements)~\cite{larsen2021deterministic, asavanant2021time, miwa2009demonstration, wang2010toward, ukai2011demonstration, ukai2011demonstrationCP, su2013gate}, whereas deterministic non-Gaussian gate implementations (required for full universality) remain challenging~\cite{sakaguchi2023nonlinear, konno2021nonlinear, marshall2015repeat, walschaers2018tailoring}.
Hence, beyond mere cluster-state generation, the demonstration of quantum advantage through large-scale MBLO constitutes a compelling near-term objective for current CV cluster-state platforms, as reflected in recent MBLO realizations~\cite{verma2025measurement}.

In this work, we investigate the computational complexity of MBLO.
We establish an explicit MBLO framework based on CV cluster states, specifying the cluster-state structure and measurement rules for universal MBLO, and develop a systematic graph-based noise analysis.
Within this framework, we demonstrate that increasing the squeezing level of CV cluster states drives a phase transition in the complexity of classically simulating the resulting output states, transitioning from a classically tractable to a classically intractable regime.
This yields a complexity phase diagram illustrated in Fig.~\ref{fig: schematics}.


\begin{figure}[b]
\includegraphics[width=0.75\linewidth]{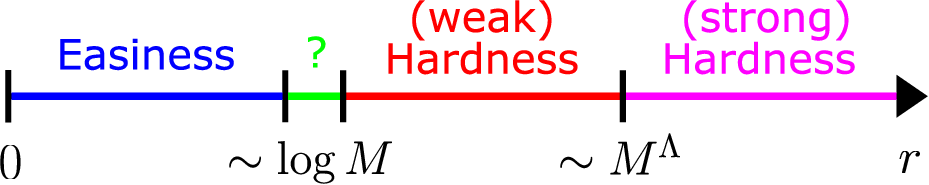}
\caption{Complexity phase diagram for simulating the measurement-based linear optics (MBLO) in Definition~\ref{def: MBLO-based sampling}.
Here, weak (strong) hardness indicates that it requires a stronger (weaker) complexity-theoretic assumption.
}
\label{fig: schematics}
\end{figure}

Overall, our results provide a constructive roadmap for current experiments by identifying a sufficient squeezing level for quantum advantage while also providing a clear experimental certificate of classical simulability.
Although our findings indicate that MBLO at currently accessible squeezing levels remains \textit{susceptible} to classical simulation, and that the asymptotic scaling of the squeezing level appears fundamental, our framework is designed to systematically analyze finite-squeezing noise, thus leaving room for further optimization.

\begin{figure*}[t]
\includegraphics[width=0.98\linewidth]{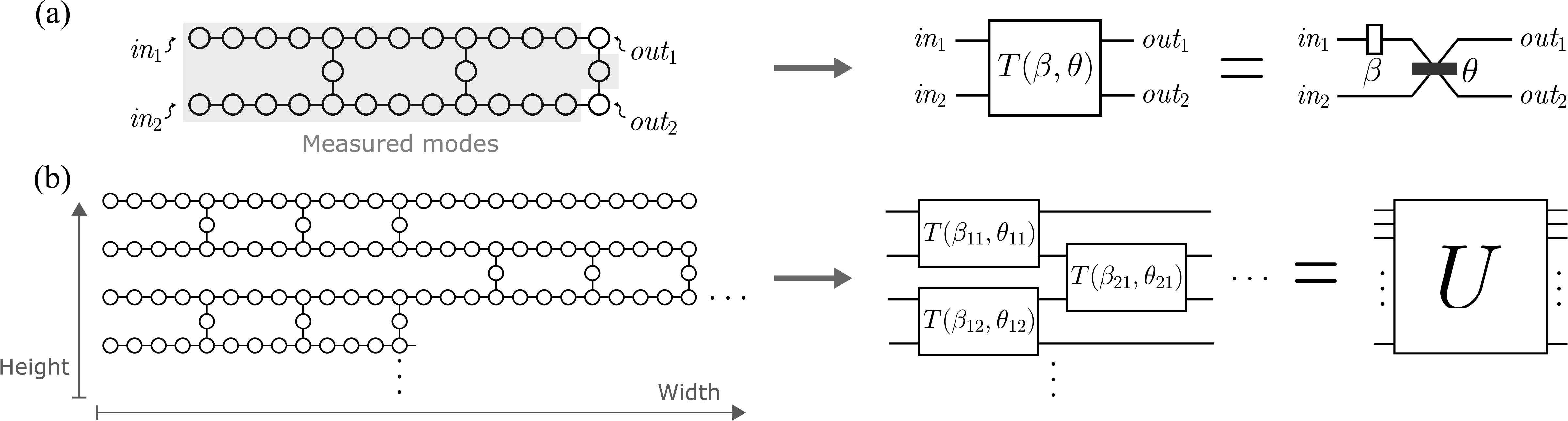}
\caption{
(a) Schematic of the brick graph $G_{T}$. 
The cluster state $\ket{G_{T}}$, with appropriate phase shifts and homodyne measurements on selected modes, can implement an arbitrary (phase-shifted) beam-splitter $T(\beta, \theta)$ up to displacement and finite-squeezing noise in Eq.~\eqref{eq: input-output relation main}. 
(b) Schematic of the two-dimensional brickwork graph $G_{U}$, constructed by cascading brick graph $G_{T}$. 
By~\cite{clements2016optimal}, any LO circuit $U$ can be implemented via a brickwork architecture of depth $\Omega(M)$, with each brick given by a beam splitter $T(\beta, \theta)$. 
Combining this with (a), a cluster state $\ket{G_{U}}$ whose underlying graph $G_{U}$ has width and height $\Theta(M)$ can implement any $M$-mode LO circuit $U$ up to displacement and finite-squeezing noise.
Note that this structure can be easily obtained provided standard two-dimensional cluster states (e.g., grid or hexagonal lattice) via vertex deletion and wire shortening~\cite{gu2009quantum, dahlberg2018transforming, ghosh2023complexity}; see~\cite{supple} for details.
}
\label{fig: graph}
\end{figure*}

\emph{Continuous-variable cluster state.---}
For a given graph $G = (V,E)$ comprising $|V| = n$ vertices, the corresponding $n$-mode CV cluster state is defined as
\begin{align}\label{def: graph state}
    \ket{G} = \prod_{(i,j)\in E} CZ_{ij}\ket{r}^{\otimes n},
\end{align}
where $CZ_{ij} = e^{i\hat{q}_i\hat{q}_j}$ denotes the CZ gate applied on modes $i,j \in [n]$ for each edge $(i,j)\in E$, and $\ket{r}$ is a squeezed-vacuum state along $\hat{p}$-axis with its squeezing level $r$ (i.e., $\langle \hat{p}_{i}^2\rangle = e^{-2r}/2$). 
Conventionally, a CV cluster state refers to $\ket{G}$ with a suitably-chosen graph $G$ that enables MBQC~\cite{briegel2001persistent, raussendorf2001one, menicucci2006universal, gu2009quantum, weedbrook2012gaussian}.
During MBQC, finite squeezing ($r < \infty$) unavoidably induces noise, whereas a perfect implementation is achievable only when $r \to \infty$.
While Eq.~\eqref{def: graph state} assumes ideal CZ gates, such interactions can also be realized using appropriate LO networks and preparing higher squeezing levels of $\ket{r}$~\cite{gu2009quantum, larsen2019deterministic, van2007building, asavanant2019generation, larsen2021deterministic, asavanant2021time, alexander2016one, alexander2018universal, alexander2016flexible, menicucci2011temporal, wang2014weaving}.




\emph{Framework for measurement-based linear optics.---}
Let $M$ denote the size of the target LO circuit to be implemented via MBLO, characterized by a unitary matrix $U \in {\rm U}(M)$; hereafter, we use the terms ``LO circuit" and ``unitary matrix" interchangeably for $U$.

To systematically quantify the finite-squeezing noise during MBLO, we adopt the convention in~\cite{larsen2020architecture, larsen2021deterministic, larsen2021fault, asavanant2021time}, wherein the effect of finite squeezing is captured through an \textit{input-output} relation of the quadrature operators. 
Specifically, after implementing $U$ via MBLO on $\ket{G}$ in Eq.~\eqref{def: graph state} for a suitably chosen graph $G$ (specified below), the $M$-mode output quadratures $(\hat{\bm{q}}_{\rm out}\;\hat{\bm{p}}_{\rm out})^{T}$ are related to the $M$-mode input quadratures $(\hat{\bm{q}}_{\rm in}\;\hat{\bm{p}}_{\rm in})^{T}$ by
\begin{align}\label{eq: input-output relation main}
    \begin{pmatrix}
        \hat{\bm{q}}_{\rm out} \\ \hat{\bm{p}}_{\rm out}
    \end{pmatrix} = \textbf{G} \begin{pmatrix}
        \hat{\bm{q}}_{\rm in} \\ \hat{\bm{p}}_{\rm in}
    \end{pmatrix} + \textbf{N}\hat{\bm{p}} + \textbf{D} \bm{m}  ,
\end{align}
where $\textbf{G}$ is the $2M \times 2M$ symplectic (orthogonal) matrix corresponding to the target LO circuit $U$.
$\hat{\bm{p}}$ is the vector of $\hat{p}$ quadratures of $\ket{r}$ in Eq.~\eqref{def: graph state}, and $\bm{m}$ is the vector of homodyne outcomes, both having dimension $n - M$ (the number of measurements on $\ket{G}$).
The matrices $\bf{D}$ and $\bf{N}$, each of size $2M \times (n-M)$, characterize the displacement and finite-squeezing noise, respectively.
The term $\textbf{D}\bm{m}$ contributes to the first-moment noise, correctable by feed-forward displacement.
In constrast, the term $\textbf{N}\hat{\bm{p}}$ contributes to the second-moment noise, necessitating incorporation of error-correction schemes~\cite{menicucci2014fault, douce2017continuous, fukui2018high, su2018correcting, larsen2021fault, larsen2020architecture, wu2020quantum, ferrini2015optimization, walshe2021streamlined, du2025complete}.




Let $G_{U}$ denote the \textit{brickwork} graph depicted in Fig.~\ref{fig: graph}, the foundation for the MBLO in this work.
This graph is constructed by cascading \textit{brick} graph $G_{T}$ shown in Fig.~\ref{fig: graph}(a).
Specifically, $\ket{G_T}$, when accompanied by appropriate phase shifts and homodyne measurements, can implement an arbitrary beam-splitter operation, up to displacement and finite-squeezing noise.
By leveraging the brickwork LO decomposition of~\cite{clements2016optimal}, a cluster state $\ket{G_{U}}$, whose underlying graph $G_{U}$ in Fig~\ref{fig: graph}(b) has both width and height $\Theta(M)$, can implement universal $U \in {\rm U}(M)$, again up to displacement and finite-squeezing noise.
The precise phase-shift angles required for MBLO on $\ket{G_{U}}$ are provided in~\cite{supple}.

We now further specify our MBLO setting.
As an input, we prepare an $M$-mode \textit{Gaussian} state $\rho_{\rm in}$, which is teleported to the input ports of $\ket{G_U}$ via Bell couplings~\cite{larsen2020architecture, ukai2010universal, alexander2014noise}.
The LO circuit $U$ is then executed on $\rho_{\rm in}$ via MBLO on $\ket{G_U}$ as described above (where $G_{U}$ has width and height $\Omega(M)$ for universality), followed by the feed-forward correction $-\textbf{D}\bm{m}$ given in Eq.~\eqref{eq: input-output relation main}.

Let $\rho_{\rm out}$ be the resulting output state.
Because the MBLO preserves Gaussianity, both $\rho_{\rm in}$ and $\rho_{\rm out}$ are Gaussian states, and thus are fully characterized by their mean vectors and covariance matrices.
Denoting these by $(\bm{\mu}_{\rm in}, V_{\rm in})$ and $(\bm{\mu}, V)$ for $\rho_{\rm in}$ and $\rho_{\rm out}$, respectively, the input-output relation in Eq.~\eqref{eq: input-output relation main} yields
\begin{align}
    \bm{\mu} &= \textbf{G}\bm{\mu}_{\rm in}  , \label{eq: mean of output state main} \\
    V &= \textbf{G}V_{\rm in}\textbf{G}^T + \frac{e^{-2r}}{2}\textbf{NN}^T   .\label{eq: covariance of output state main}
\end{align}


We emphasize that $\rho_{\rm out}$ remains \textit{mixed} even after the feed-forward displacement $-\textbf{D} \bm{m}$ due to the residual noise $\textbf{N} \hat{\bm{p}}$, which vanishes when $r \rightarrow \infty$~\cite{alexander2014noise, menicucci2014fault, larsen2020architecture, larsen2021deterministic, larsen2021fault, walshe2019robust}. 
While one could instead apply a displacement to obtain a pure output state for MBLO, derived by Schur-complement (or by analyzing the Wigner function~\cite{gu2009quantum, alexander2014noise, su2018correcting}), albeit one that still deviates from the target state, we adopt the feed-forward correction $-\textbf{D}\bm{m}$
to characterize the finite-squeezing noise as an \textit{additive} contribution to the covariance as in~\cite{alexander2014noise, menicucci2014fault, larsen2020architecture, larsen2021deterministic, larsen2021fault, walshe2019robust}.


Finally, our goal is to characterize the classical complexity of \textit{weak-simulating} $\rho_{\rm out}$, namely, sampling from projective measurements of $\rho_{\rm out}$ in the local boson-number basis, as formalized below.

\begin{definition}[MBLO-based sampling]\label{def: MBLO-based sampling}
Given as input an $M \times M$ unitary matrix $U$ and a description of an $M$-mode input Gaussian state $\rho_{\rm in}$, the MBLO-based sampling task is to output a sample drawn from the distribution
\begin{align}\label{eq: p(n) main}
    p( \bm{n}) = \Tr\left[\ket{\bm{n}}\!\bra{\bm{n}} \rho_{\rm out} \right]  ,
\end{align} 
where $\rho_{\rm out}$ is the output state produced by MBLO performed on $\ket{G_U}$, with the mean vector and covariance matrix given by Eq.~\eqref{eq: mean of output state main} and Eq.~\eqref{eq: covariance of output state main}, respectively, and $\ket{\bm{n}}$ denotes the boson-number basis corresponding to $\bm{n} = (n_1,\dots,n_{M})$.
\end{definition}

To summarize our results in advance, we show that as the squeezing level $r$ increases, the MBLO-based sampling task in Definition~\ref{def: MBLO-based sampling} undergoes a complexity phase transition, from a classically tractable regime to a classically intractable one, as illustrated in Fig.~\ref{fig: schematics}.

\emph{Easiness regime.---}
Our first result establishes that when the squeezing level $r$ falls below a certain threshold, the MBLO-based sampling task in Definition~\ref{def: MBLO-based sampling} is classically tractable.

\begin{theorem}[Easiness regime]\label{theorem: easiness main}
     There exists a threshold $r_{\rm th} = O\left(\log M\right)$ such that for any squeezing level $r \leq r_{\rm th}$, a polynomial-time classical algorithm exists that can simulate the MBLO-based sampling.
\end{theorem}

\begin{proof}[Proof Sketch of Theorem~\ref{theorem: easiness main}]
(See~\cite{supple} for the full proof) 
Using the phase-space representation given in~\cite{rahimi2016sufficient}, from Eq.~\eqref{eq: p(n) main} we have 
\begin{align}
    p( \bm{n} ) = \pi^{M} \int d\bm{\alpha} Q_{\bm{n}}(\bm{\alpha})P_{\rm out}(\bm{\alpha}), 
\end{align}
where $\bm{\alpha}$ is an $M$-dimensional complex vector representing an $M$-mode phase-space point, $Q_{\bm{n}}(\bm{\alpha})$ is the Husimi Q representation of the number-basis $\ket{\bm{n}}$~\cite{husimi1940some}, and $P_{\rm out}(\bm{\alpha})$ is the Glauber-Sudarshan P representation of the output state $\rho_{\rm out}$~\cite{glauber1963photon, sudarshan1963equivalence}.
Here, for the Gaussian output state $\rho_{\rm out}$, $P_{\rm out}(\bm{\alpha})$ can be written as
\begin{align}\label{eq: p-distribution}
    P_{\rm out}(\bm{\alpha}) = \frac{e^{-({\bm{\alpha}-\bm{\mu}})^{T}\left(V-\frac{1}{2}\mathbb{I}_{2M}\right)^{-1}({\bm{\alpha}-\bm{\mu}})  } }{\pi^{M}\det(V-\frac{1}{2}\mathbb{I}_{2M})}  ,
\end{align}
where $\mathbb{I}_{2M}$ is the $2M \times 2M$ identity matrix, and $\bm{\mu}$ and $V$ are given in Eq.~\eqref{eq: mean of output state main} and Eq.~\eqref{eq: covariance of output state main}, respectively.


By the classical simulability criterion of~\cite{rahimi2016sufficient}, $p(\bm{n})$ can be sampled efficiently whenever $P_{\rm out}(\bm{\alpha})$ is non-negative and efficiently sampleable, which, by Eq.~\eqref{eq: p-distribution}, holds whenever $V - \frac{1}{2}\mathbb{I}_{2M} > 0$.
Moreover, since $V_{\rm in} > 0$ for an arbitrary input state $\rho_{\rm in}$, from Eq.~\eqref{eq: covariance of output state main}, this condition is satisfied whenever
\begin{align}\label{eq: easiness condition main}
    \textbf{N}\textbf{N}^T - e^{2r}\mathbb{I}_{2M} \geq 0  .
\end{align}
Therefore, if the minimum eigenvalue of $\textbf{N}\textbf{N}^T$ is greater than $e^{2r}$, the condition in Eq.~\eqref{eq: easiness condition main} is satisfied, implying that sampling from $p(\bm{n})$ is classically simulable.

Next, our MBLO procedure via the cluster state $\ket{G_{U}}$ implements the target LO circuit $U$ in a \textit{gate-wise} manner (see Fig.~\ref{fig: graph}).
Hence, for LO circuit decomposed by $U = U_{d}U_{d-1}\cdots U_{1}$, where each $U_{i}$ corresponds to the operation applied at the $i$th circuit layer (a \textit{parallel} array of beam-splitter operations) and $d$ is the effective circuit depth of LO circuit, the MBLO procedure implements the overall symplectic transformation $\textbf{G} = \textbf{G}_{d}\textbf{G}_{d-1}\cdots \textbf{G}_{1}$ via the sequential implementation of $\textbf{G}_{i}$, where each $\textbf{G}_{i}$ is the symplectic matrix corresponding to $U_{i}$ (thus orthogonal).
Given that $\textbf{G}$ is constructed through this sequential composition, the noise term $\textbf{NN}^T$ can be decomposed as 
\begin{align}\label{eq: decomposing noise matrix}
    \textbf{NN}^T 
    = \sum_{i = 1}^{d} \bar{\textbf{G}}_i\textbf{N}_i\textbf{N}_i^T\bar{\textbf{G}}_i^{T},
\end{align}
where $\textbf{N}_{i}$ is a noise matrix arises when implementing each $\textbf{G}_{i}$ via MBLO, and $\bar{\textbf{G}}_{k} = \textbf{G}\textbf{G}_{1}^{-1}\textbf{G}_{2}^{-1}\cdots \textbf{G}_{k}^{-1}$.
By a graph-theoretic analysis of the noise matrix in the input-output relation of Eq.~\eqref{eq: input-output relation main}, we show that for arbitrary $\textbf{G}_{i}$, associated with each layer $U_i$ of $U$, the minimum eigenvalue of each $\textbf{N}_{i}\textbf{N}_{i}^T$ is lower-bounded by a non-zero constant.
Since each $\bar{\textbf{G}}_{i}$ is an orthogonal matrix and thus preserves eigenvalues, and since the circuit depth satisfies $d \geq \Omega(M)$ in out settings (for universality~\cite{clements2016optimal}), it follows that the minimum eigenvalue of $\textbf{NN}^T$ is lower-bounded by $\Omega(M)$.
Therefore, there exists a threshold $r_{\rm th} = O(\log M)$ such that for any $r \leq r_{\rm th}$, the simulability condition in Eq.~\eqref{eq: easiness condition main} holds.

\end{proof}

To clarify the value of the simulability threshold $r$ in practice, we numerically analyze the threshold below which MBLO-based sampling becomes classically easy (i.e., $r$ satisfying Eq.~\eqref{eq: easiness condition main}), shown in Fig.~\ref{fig: numerics_easiness}.

While Theorem~\ref{theorem: easiness main} does not cover all possible MBLO implementations, analogous behavior is expected in general (as in~\cite{alexander2017measurement}), since finite-squeezing noise will accumulate in any MBLO framework.
For example, even under the triangular LO decomposition~\cite{reck1994experimental}, implementing interactions between distant modes (e.g., from the first to the last mode) necessitates $\Theta(M)$ circuit depth, which correspondingly accumulates $\Theta(M)$ noise terms $\textbf{N}_{i}\textbf{N}_{i}^T$ in Eq.~\eqref{eq: decomposing noise matrix}, leading to the same conclusion as in Theorem~\ref{theorem: easiness main}.

Also, to compare with recent large-scale MBLO implementations~\cite{verma2025measurement}, their introduced circuit family is \textit{non-universal} as the number of independent parameters scales $\Theta(M)$, whereas many applications~\cite{kwon2022quantum, banchi2018multiphoton, arzani2019random, aaronson2011computational, hamilton2017gaussian} require universality that necessitates $\Theta(M^2)$ parameters~\cite{clements2016optimal}.
Our easiness result thus complements their findings by identifying a barrier when one attempts to \textit{go beyond} such restricted circuit families, because increasing the number of parameters requires increasing circuit depth, which will in turn accumulate finite-squeezing noise proportionally as captured by our analysis.
Suppressing this noise, therefore, requires either scaling the squeezing level accordingly or incorporating error-correction schemes.

\begin{figure}[t]
\includegraphics[width=0.92\linewidth]{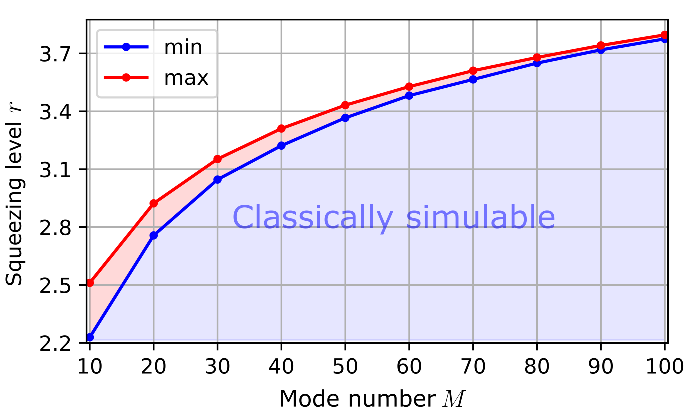}
\caption{Numerically obtained squeezing level $r$ below which MBLO-based sampling becomes classically simulable. 
For each $M \in \{10, 20,\dots,100\}$, we sample $5000$ Haar-random unitaries $U \in \text{U}(M)$.
For each $U$, we compute $r$ such that the minimum eigenvalue of $\textbf{NN}^T$ equals $e^{2r}$.
The blue and red curves indicate the minimum and maximum values of such $r$ over $U$, respectively.
Hence, the blue region implies a simulable regime for all $U$, whereas the red region implies a simulable regime for some portion of $U$. 
}
\label{fig: numerics_easiness}
\end{figure}


\emph{Hardness regime I.---}
Turning to our hardness results, we demonstrate the classical intractability of the MBLO-based sampling task for $r$ beyond a threshold.
This result relies on the widely accepted conjecture that simulating Gaussian boson sampling (GBS) to within an inverse-polynomial total variation distance (TVD) is classically hard~\cite{hamilton2017gaussian, kruse2019detailed, deshpande2022quantum, go2025sufficient}.

\begin{theorem}[Hardness regime I]\label{theorem: weak hardness main}
    Suppose that there exists $\varepsilon = {\rm poly}(M)^{-1}$ such that simulating GBS within TVD error $\varepsilon$ is classically intractable.
    Then, there exists a threshold $r_{\rm th} = \Omega\left(\log M\right)$ such that for any squeezing level $r \geq r_{\rm th}$, no polynomial-time classical algorithm can simulate the MBLO-based sampling.
\end{theorem}

We here sketch the proof of Theorem~\ref{theorem: weak hardness main}.
Let $\rho_{\rm id}$ be the output state $\rho_{\rm out}$ of the MBLO scheme in Definition~\ref{def: MBLO-based sampling} in the infinite-squeezing limit $r \rightarrow \infty$. 
Consider an appropriate input state $\rho_{\rm in}$ such that $\rho_{\rm id}$ coincides with the output state of a standard GBS setup~\cite{hamilton2017gaussian, kruse2019detailed, deshpande2022quantum, go2025sufficient}. 
For the Gaussian states $\rho_{\rm out}$ and $\rho_{\rm id}$, we make use of two facts: 
(1) their TVD $\mathcal{D}_{\rho_{\rm out}, \rho_{\rm id}}$ is upper bounded in terms of their quantum fidelity $F_{\rho_{\rm out}, \rho_{\rm id}}$~\cite{fuchs1999cryptographic}, and
(2) $F_{\rho_{\rm out}, \rho_{\rm id}}$ admits a simple expression in terms of their covariance matrices~\cite{marian2012uhlmann, spedalieri2012limit}.
These allow us to bound $\mathcal{D}_{\rho_{\rm out}, \rho_{\rm id}}$ in terms of the squeezing level $r$ and noise matrix $\bf{N}$ by
\begin{align}
    \mathcal{D}_{\rho_{\rm out}, \rho_{\rm id}} 
    \leq \sqrt{1 - F_{\rho_{\rm out}, \rho_{\rm id}}} 
    \le e^{-r} \sqrt{\frac{\|{\bf N}{\bf N}^T \|_{F} \| V_{\rm in}^{-1} \|_{F}}{8}}   , \label{eq: TVD bound}
\end{align}
for $V_{\rm in}$ being the covariance of $\rho_{\rm in}$.
Moreover, the norm $\|{\bf N}{\bf N}^T \|_{F}$ is bounded by $\text{poly}(M)$, which, by Eq.~\eqref{eq: TVD bound}, implies that sampling from $p(\bm{n})$ enables the simulation of GBS within TVD error $\text{poly}(M)e^{-r}$. 
Therefore, if simulating GBS within a certain inverse-polynomial TVD is classically intractable, it follows that there exists a threshold $r_{\rm th} = \Omega\left(\log M\right)$ such that for all $r \geq r_{\rm th}$, the MBLO-based sampling is classically intractable. 
A detailed argument is provided in~\cite{supple}.



Importantly, under the same conjecture, this hardness argument can readily be extended to the \textit{approximate} simulation within a bounded TVD.
Specifically, by the triangle inequality, sampling from $p(\bm{n})$ in Eq.~\eqref{eq: p(n) main} within TVD $\kappa = \text{poly}(M)^{-1}$ enables simulating ideal GBS within TVD $\kappa + \text{poly}(M)e^{-2r}$, which remains $\text{poly}(M)^{-1}$ when $r \geq \Omega\left(\log M\right)$.
Hence, the following corollary follows directly. 


\begin{corollary}[Hardness of approximate simulation]
    Under the same assumption as in Theorem~\ref{theorem: weak hardness main}, there exists a threshold $r_{\rm th} = \Omega\left(\log M\right)$ such that for any squeezing level $r \geq r_{\rm th}$, no polynomial-time classical algorithm can approximately sample from $p(\bm{n})$ in Eq.~\eqref{eq: p(n) main} within inverse-polynomial TVD.
\end{corollary}

\emph{Hardness regime II.---}
Furthermore, we establish that when the squeezing level $r$ exceeds a certain, larger threshold than that in Theorem~\ref{theorem: weak hardness main}, then simulating $\rho_{\rm out}$ becomes classically intractable unless the polynomial hierarchy (PH) collapses.
This non-collapse of PH is generally considered a weaker and more foundational assumption than that in Theorem~\ref{theorem: weak hardness main};
existing arguments for the simulation hardness of GBS rely on the non-collapse of PH together with additional conjectures that remain open, including the average-case \#P-hardness of approximating hafnians of random matrices~\cite{hamilton2017gaussian, kruse2019detailed, deshpande2022quantum, go2025sufficient} and the anti-concentration of hafnians~\cite{ehrenberg2025transition, ehrenberg2025second}.


\begin{theorem}[Hardness regime II]\label{theorem: strong hardness main}
    For any constant $\Lambda >0$, there exists a threshold $r_{\rm th} = \Omega\left(M^{\Lambda}\right)$ such that for any squeezing level $r \geq r_{\rm th}$, no polynomial-time classical algorithm can simulate the MBLO-based sampling, unless PH collapses to a finite level.
\end{theorem}

While CV cluster states with squeezing levels stated in Theorem~\ref{theorem: strong hardness main} are computationally very powerful, the energy required to generate such states now scales \textit{subexponentially} by~\cite{liu2016power}, highlighting an inherent trade-off between achievable computational power and physical resource requirements in CV cluster states.

\begin{proof}[Proof Sketch of Theorem~\ref{theorem: weak hardness main}]
Consider an $N_0\times N_0$ matrix $W' = \{-1,0,1\}^{N_0 \times N_0}$, for which computing $|\Per(W')|^2$ is \#P-hard in the worst-case.
Given $W'$ and a constant $\Lambda > 0$, let $N = \left\lceil N_0^{1/\lambda} \right\rceil$ for any constant $\lambda \in (0, \Lambda)$ and take $M \geq 2N$. 
Then, one can construct LO circuit $U$ and an input state $\rho_{\rm in}$ such that output probability $q(\bm{n})$ for an $N$-boson outcome $\bm{n}$ of the ideal output state $\rho_{\rm id}$ in Eq.~\eqref{eq: TVD bound} is proportional to $|\Per(W')|^2$ up to a multiplicative factor.  
Moreover, by the additional contribution $\frac{e^{-2r}}{2}\textbf{NN}^T$ in the covariance in Eq.~\eqref{eq: covariance of output state main}, $q(\bm{n})$ and $p(\bm{n})$ in Eq.~\eqref{eq: p(n) main} satisfy (when $q(\bm{n}) \neq 0$)
\begin{align}\label{eq: our goal main 2}
    \left| 1 - \frac{p(\bm{n} )}{q(\bm{n})} \right| \leq (N_0!)^2\text{poly}(M)e^{-2r} .
\end{align}
Hence, when $r$ exceeds a certain threshold $r_{\rm th} = \Omega(M^{\Lambda})$ for $\Lambda > \lambda$, $p(\bm{n})$ and $q(\bm{n})$ are multiplicatively close within inverse-polynomial imprecision.
Then, by Stockmeyer's algorithm~\cite{stockmeyer1985approximation}, one can multiplicatively estimate $\Per(W')^2$ in ${\rm BPP^{NP}}$ given oracle access to the MBLO-based sampler, finally implying that the sampling task is classically hard unless PH collapses to a finite level (${\rm BPP^{NP}}$). 
A detailed proof of Theorem~\ref{theorem: strong hardness main}, including how to deal with the case $q(\bm{n}) \propto |\Per(W')|^2 =0$, is presented in~\cite{supple}.
\end{proof}

\emph{Remarks and future work.---}
While any universal CV cluster state can conceptually implement MBLO, we explicitly utilize the brickwork graph $G_{U}$ (Fig.~\ref{fig: graph}) due to its systematic alignment with graph-based finite-squeezing noise analysis and the brickwork LO decomposition~\cite{clements2016optimal}. 
As shown in~\cite{supple}, this structure can be directly embedded into standard grid or hexagonal lattice CV cluster states, as well as dual- or quad-rail lattices~\cite{menicucci2011temporal, wang2014weaving, alexander2016flexible, alexander2017measurement, asavanant2019generation, larsen2021deterministic, asavanant2021time, alexander2016one, alexander2018universal} that support grid-lattice implementations.

To compare our result with existing MBLO arguments, Ref.~\cite{alexander2017measurement} proposed MBLO schemes based on the quad-rail lattice architecture~\cite{menicucci2011temporal, wang2014weaving, alexander2016flexible} that maps the finite-squeezing noise to photon loss, identifying squeezing levels for the hardness of MBLO-based boson sampling.
In contrast, we identified distinct easiness and hardness regimes as a function of squeezing, which help understand the precise squeezing thresholds that delineate the quantum and classical computational regimes of CV cluster states.

Since our easiness result does not cover all possible MBLO implementations, it remains open whether such simulability holds in more general settings in~\cite{menicucci2011temporal, alexander2016flexible, alexander2017measurement, larsen2019deterministic, asavanant2019generation, larsen2021deterministic, asavanant2021time, verma2025measurement}.
Beyond the specific framework analyzed in this work, it would also be interesting to examine the complexity of broader classes of computational tasks using CV cluster states, thereby further elucidating how the squeezing level affects the intrinsic computational capabilities of CV cluster states.

\emph{Acknowledgements.---}
B.G. and H.J. were supported by the Korean government (Ministry of Science and ICT~(MSIT)),
the NRF grants funded by the Korea government~(MSIT)~(Nos.~RS-2024-00413957,~RS-2024-00438415, and NRF-2023R1A2C1006115),~and the Institute of Information \& Communications Technology Planning \& Evaluation (IITP) grant funded by the Korea government (MSIT) (IITP-2025-RS-2020-II201606 and  IITP-2025-RS-2024-00437191).
C.O. was supported by the NRF Grants (No. RS-2024-00431768 and No. RS-2025-00515456) funded by the Korean government (MSIT) and IITP grants funded by the Korea government (MSIT) (No. RS-2024-00437284, No. IITP-2025-RS-2025-02283189 and No. IITP-2025-RS-2025-02263264) and by Global Partnership Program of Leading Universities in Quantum Science and Technology (RS-2025-08542968) through the National Research Foundation of Korea~(NRF) funded by the Korean government (Ministry of Science and ICT(MSIT)).

\let\oldaddcontentsline\addcontentsline

\renewcommand{\addcontentsline}[3]{}

\bibliographystyle{unsrt}
\bibliography{Reference.bib}

@article{verma2025measurement,
  title={Measurement-Induced Multimode Squeezed Light Interferometers with Scalable Architectures},
  author={Verma, Abhinav and Hastrup, Jacob and Neergaard-Nielsen, Jonas S and Andersen, Ulrik L},
  journal={arXiv preprint arXiv:2503.14449},
  year={2025}
}

@article{konno2021nonlinear,
  title={Nonlinear squeezing for measurement-based non-gaussian operations in time domain},
  author={Konno, Shunya and Sakaguchi, Atsushi and Asavanant, Warit and Ogawa, Hisashi and Kobayashi, Masaya and Marek, Petr and Filip, Radim and Yoshikawa, Jun-ichi and Furusawa, Akira},
  journal={Physical Review Applied},
  volume={15},
  number={2},
  pages={024024},
  year={2021},
  publisher={APS}
}

@article{marshall2015repeat,
  title={Repeat-until-success cubic phase gate for universal continuous-variable quantum computation},
  author={Marshall, Kevin and Pooser, Raphael and Siopsis, George and Weedbrook, Christian},
  journal={Physical Review A},
  volume={91},
  number={3},
  pages={032321},
  year={2015},
  publisher={APS}
}

@article{sakaguchi2023nonlinear,
  title={Nonlinear feedforward enabling quantum computation},
  author={Sakaguchi, Atsushi and Konno, Shunya and Hanamura, Fumiya and Asavanant, Warit and Takase, Kan and Ogawa, Hisashi and Marek, Petr and Filip, Radim and Yoshikawa, Jun-ichi and Huntington, Elanor and others},
  journal={Nature Communications},
  volume={14},
  number={1},
  pages={3817},
  year={2023},
  publisher={Nature Publishing Group UK London}
}

@article{kwon2022quantum,
  title={Quantum metrological power of continuous-variable quantum networks},
  author={Kwon, Hyukgun and Lim, Youngrong and Jiang, Liang and Jeong, Hyunseok and Oh, Changhun},
  journal={Physical Review Letters},
  volume={128},
  number={18},
  pages={180503},
  year={2022},
  publisher={APS}
}

@article{banchi2018multiphoton,
  title={Multiphoton tomography with linear optics and photon counting},
  author={Banchi, Leonardo and Kolthammer, W Steven and Kim, MS},
  journal={Physical Review Letters},
  volume={121},
  number={25},
  pages={250402},
  year={2018},
  publisher={APS}
}

@article{arzani2019random,
  title={Random coding for sharing bosonic quantum secrets},
  author={Arzani, Francesco and Ferrini, Giulia and Grosshans, Fr{\'e}d{\'e}ric and Markham, Damian},
  journal={Physical Review A},
  volume={100},
  number={2},
  pages={022303},
  year={2019},
  publisher={APS}
}

@article{liu2025robust,
  title={Robust quantum computational advantage with programmable 3050-photon Gaussian boson sampling},
  author={Liu, Hua-Liang and Su, Hao and Gong, Si-Qiu and Gu, Yi-Chao and Tang, Hao-Yang and Jia, Meng-Hao and Wei, Qian and Song, Yukun and Wang, Dongzhou and Zheng, Mingyang and others},
  journal={arXiv preprint arXiv:2508.09092},
  year={2025}
}

@article{alexander2014noise,
  title={Noise analysis of single-mode Gaussian operations using continuous-variable cluster states},
  author={Alexander, Rafael N and Armstrong, Seiji C and Ukai, Ryuji and Menicucci, Nicolas C},
  journal={Physical Review A},
  volume={90},
  number={6},
  pages={062324},
  year={2014},
  publisher={APS}
}

@article{larsen2021fault,
  title={Fault-tolerant continuous-variable measurement-based quantum computation architecture},
  author={Larsen, Mikkel V and Chamberland, Christopher and Noh, Kyungjoo and Neergaard-Nielsen, Jonas S and Andersen, Ulrik L},
  journal={Prx Quantum},
  volume={2},
  number={3},
  pages={030325},
  year={2021},
  publisher={APS}
}

@article{walshe2019robust,
  title={Robust fault tolerance for continuous-variable cluster states with excess anti-squeezing},
  author={Walshe, Blayney W and Mensen, Lucas J and Baragiola, Ben Q and Menicucci, Nicolas C},
  journal={arXiv preprint arXiv:1903.02162},
  year={2019}
}

@article{van2007building,
  title={Building Gaussian cluster states by linear optics},
  author={van Loock, Peter and Weedbrook, Christian and Gu, Mile},
  journal={Physical Review A—Atomic, Molecular, and Optical Physics},
  volume={76},
  number={3},
  pages={032321},
  year={2007},
  publisher={APS}
}

@misc{supple,
    title={Supplemental {M}aterial}}

@article{walschaers2018tailoring,
  title={Tailoring non-Gaussian continuous-variable graph states},
  author={Walschaers, Mattia and Sarkar, Supratik and Parigi, Valentina and Treps, Nicolas},
  journal={Physical review letters},
  volume={121},
  number={22},
  pages={220501},
  year={2018},
  publisher={APS}
}

@article{menicucci2011graphical,
  title={Graphical calculus for Gaussian pure states},
  author={Menicucci, Nicolas C and Flammia, Steven T and van Loock, Peter},
  journal={Physical Review A—Atomic, Molecular, and Optical Physics},
  volume={83},
  number={4},
  pages={042335},
  year={2011},
  publisher={APS}
}

@article{ehrenberg2025transition,
  title={Transition of anticoncentration in Gaussian boson sampling},
  author={Ehrenberg, Adam and Iosue, Joseph T and Deshpande, Abhinav and Hangleiter, Dominik and Gorshkov, Alexey V},
  journal={Physical Review Letters},
  volume={134},
  number={14},
  pages={140601},
  year={2025},
  publisher={APS}
}

@article{zhang2006continuous,
  title={Continuous-variable Gaussian analog of cluster states},
  author={Zhang, Jing and Braunstein, Samuel L},
  journal={Physical Review A—Atomic, Molecular, and Optical Physics},
  volume={73},
  number={3},
  pages={032318},
  year={2006},
  publisher={APS}
}

@article{cable2010bipartite,
  title={Bipartite entanglement in continuous variable cluster states},
  author={Cable, Hugo and Browne, Daniel E},
  journal={New Journal of Physics},
  volume={12},
  number={11},
  pages={113046},
  year={2010},
  publisher={IOP Publishing}
}

@article{su2018correcting,
  title={Correcting finite squeezing errors in continuous-variable cluster states},
  author={Su, Daiqin and Weedbrook, Christian and Br{\'a}dler, Kamil},
  journal={arXiv preprint arXiv:1801.03488},
  year={2018}
}

@article{bouland2024average,
  title={On the average-case hardness of BosonSampling},
  author={Bouland, Adam and Datta, Ishaun and Fefferman, Bill and Hernandez, Felipe},
  journal={arXiv preprint arXiv:2411.04566},
  year={2024}
}

@article{weyl1912asymptotische,
  title={Das asymptotische Verteilungsgesetz der Eigenwerte linearer partieller Differentialgleichungen (mit einer Anwendung auf die Theorie der Hohlraumstrahlung)},
  author={Weyl, Hermann},
  journal={Mathematische Annalen},
  volume={71},
  number={4},
  pages={441--479},
  year={1912},
  publisher={Springer}
}

@article{aaronson2011linear,
  title={A linear-optical proof that the permanent is\# P-hard},
  author={Aaronson, Scott},
  journal={Proceedings of the Royal Society A: Mathematical, Physical and Engineering Sciences},
  volume={467},
  number={2136},
  pages={3393--3405},
  year={2011},
  publisher={The Royal Society Publishing}
}

@article{douce2017continuous,
  title={Continuous-variable instantaneous quantum computing is hard to sample},
  author={Douce, Tom and Markham, Damian and Kashefi, Elham and Diamanti, Eleni and Coudreau, Thomas and Milman, Perola and Van Loock, Peter and Ferrini, Giulia},
  journal={Physical review letters},
  volume={118},
  number={7},
  pages={070503},
  year={2017},
  publisher={APS}
}

@article{menicucci2014fault,
  title={Fault-tolerant measurement-based quantum computing with continuous-variable cluster states},
  author={Menicucci, Nicolas C},
  journal={Physical review letters},
  volume={112},
  number={12},
  pages={120504},
  year={2014},
  publisher={APS}
}

@inproceedings{aaronson2011computational,
  title={The computational complexity of linear optics},
  author={Aaronson, Scott and Arkhipov, Alex},
  booktitle={Proceedings of the forty-third annual ACM symposium on Theory of computing},
  pages={333--342},
  year={2011}
}

@article{deshpande2022quantum,
  title={Quantum computational advantage via high-dimensional Gaussian boson sampling},
  author={Deshpande, Abhinav and Mehta, Arthur and Vincent, Trevor and Quesada, Nicol{\'a}s and Hinsche, Marcel and Ioannou, Marios and Madsen, Lars and Lavoie, Jonathan and Qi, Haoyu and Eisert, Jens and others},
  journal={Science advances},
  volume={8},
  number={1},
  pages={eabi7894},
  year={2022},
  publisher={American Association for the Advancement of Science}
}

@article{hamilton2017gaussian,
  title={Gaussian boson sampling},
  author={Hamilton, Craig S and Kruse, Regina and Sansoni, Linda and Barkhofen, Sonja and Silberhorn, Christine and Jex, Igor},
  journal={Physical review letters},
  volume={119},
  number={17},
  pages={170501},
  year={2017},
  publisher={APS}
}

@article{madsen2022quantum,
  title={Quantum computational advantage with a programmable photonic processor},
  author={Madsen, Lars S and Laudenbach, Fabian and Askarani, Mohsen Falamarzi and Rortais, Fabien and Vincent, Trevor and Bulmer, Jacob FF and Miatto, Filippo M and Neuhaus, Leonhard and Helt, Lukas G and Collins, Matthew J and others},
  journal={Nature},
  volume={606},
  number={7912},
  pages={75--81},
  year={2022},
  publisher={Nature Publishing Group UK London}
}

@article{zhong2021phase,
  title={Phase-programmable {G}aussian boson sampling using stimulated squeezed light},
  author={Zhong, Han-Sen and Deng, Yu-Hao and Qin, Jian and Wang, Hui and Chen, Ming-Cheng and Peng, Li-Chao and Luo, Yi-Han and Wu, Dian and Gong, Si-Qiu and Su, Hao and others},
  journal={Physical review letters},
  volume={127},
  number={18},
  pages={180502},
  year={2021},
  publisher={APS}
}

@article{deng2023gaussian,
  title = {Gaussian Boson Sampling with Pseudo-Photon-Number-Resolving Detectors and Quantum Computational Advantage},
  author={Deng, Yu-Hao and Gu, Yi-Chao and Liu, Hua-Liang and Gong, Si-Qiu and Su, Hao and Zhang, Zhi-Jiong and Tang, Hao-Yang and Jia, Meng-Hao and Xu, Jia-Min and Chen, Ming-Cheng and others},
  journal={Physical review letters},
  volume = {131},
  number={15},
  pages = {150601},
  year = {2023},
  publisher = {APS}
}

@article{stockmeyer1985approximation,
  title={On approximation algorithms for\# P},
  author={Stockmeyer, Larry},
  journal={SIAM Journal on Computing},
  volume={14},
  number={4},
  pages={849--861},
  year={1985},
  publisher={SIAM}
}

@article{raussendorf2001one,
  title={A one-way quantum computer},
  author={Raussendorf, Robert and Briegel, Hans J},
  journal={Physical review letters},
  volume={86},
  number={22},
  pages={5188},
  year={2001},
  publisher={APS}
}

@article{oszmaniec2018classical,
  title={Classical simulation of photonic linear optics with lost particles},
  author={Oszmaniec, Micha{\l} and Brod, Daniel J},
  journal={New Journal of Physics},
  volume={20},
  number={9},
  pages={092002},
  year={2018},
  publisher={IOP Publishing}
}

@article{garcia2019simulating,
  title={Simulating boson sampling in lossy architectures},
  author={Garc{\'\i}a-Patr{\'o}n, Ra{\'u}l and Renema, Jelmer J and Shchesnovich, Valery},
  journal={Quantum},
  volume={3},
  pages={169},
  year={2019},
  publisher={Verein zur F{\"o}rderung des Open Access Publizierens in den Quantenwissenschaften}
}

@article{qi2020regimes,
  title={Regimes of classical simulability for noisy Gaussian boson sampling},
  author={Qi, Haoyu and Brod, Daniel J and Quesada, Nicol{\'a}s and Garc{\'\i}a-Patr{\'o}n, Ra{\'u}l},
  journal={Physical review letters},
  volume={124},
  number={10},
  pages={100502},
  year={2020},
  publisher={APS}
}

@article{toda1991pp,
  title={PP is as hard as the polynomial-time hierarchy},
  author={Toda, Seinosuke},
  journal={SIAM Journal on Computing},
  volume={20},
  number={5},
  pages={865--877},
  year={1991},
  publisher={SIAM}
}

@article{oh2022classical,
  title={Classical simulation of boson sampling based on graph structure},
  author={Oh, Changhun and Lim, Youngrong and Fefferman, Bill and Jiang, Liang},
  journal={Physical Review Letters},
  volume={128},
  number={19},
  pages={190501},
  year={2022},
  publisher={APS}
}

@article{gu2009quantum,
  title={Quantum computing with continuous-variable clusters},
  author={Gu, Mile and Weedbrook, Christian and Menicucci, Nicolas C and Ralph, Timothy C and van Loock, Peter},
  journal={Physical Review A—Atomic, Molecular, and Optical Physics},
  volume={79},
  number={6},
  pages={062318},
  year={2009},
  publisher={APS}
}

@article{pfister2019continuous,
  title={Continuous-variable quantum computing in the quantum optical frequency comb},
  author={Pfister, Olivier},
  journal={Journal of Physics B: Atomic, Molecular and Optical Physics},
  volume={53},
  number={1},
  pages={012001},
  year={2019},
  publisher={IOP Publishing}
}

@article{oh2025recent,
  title={Recent theoretical and experimental progress on boson sampling},
  author={Oh, Changhun},
  journal={Current Optics and Photonics},
  volume={9},
  number={1},
  pages={1--18},
  year={2025},
  publisher={Optical Society of Korea}
}

@article{alexander2016one,
  title={One-way quantum computing with arbitrarily large time-frequency continuous-variable cluster states from a single optical parametric oscillator},
  author={Alexander, Rafael N and Wang, Pei and Sridhar, Niranjan and Chen, Moran and Pfister, Olivier and Menicucci, Nicolas C},
  journal={Physical Review A},
  volume={94},
  number={3},
  pages={032327},
  year={2016},
  publisher={APS}
}

@article{oh2025classical,
  title={Classical simulability of constant-depth linear-optical circuits with noise},
  author={Oh, Changhun},
  journal={npj Quantum Information},
  volume={11},
  number={1},
  pages={126},
  year={2025},
  publisher={Nature Publishing Group UK London}
}

@article{wang2014weaving,
  title={Weaving quantum optical frequency combs into continuous-variable hypercubic cluster states},
  author={Wang, Pei and Chen, Moran and Menicucci, Nicolas C and Pfister, Olivier},
  journal={Physical Review A},
  volume={90},
  number={3},
  pages={032325},
  year={2014},
  publisher={APS}
}

@article{liu2016power,
  title={Power of one qumode for quantum computation},
  author={Liu, Nana and Thompson, Jayne and Weedbrook, Christian and Lloyd, Seth and Vedral, Vlatko and Gu, Mile and Modi, Kavan},
  journal={Physical Review A},
  volume={93},
  number={5},
  pages={052304},
  year={2016},
  publisher={APS}
}

@article{wang2010toward,
  title={Toward demonstrating controlled-X operation based on continuous-variable four-partite cluster states and quantum teleporters},
  author={Wang, Yu and Su, Xiaolong and Shen, Heng and Tan, Aihong and Xie, Changde and Peng, Kunchi},
  journal={Physical Review A—Atomic, Molecular, and Optical Physics},
  volume={81},
  number={2},
  pages={022311},
  year={2010},
  publisher={APS}
}

@article{menicucci2011temporal,
  title={Temporal-mode continuous-variable cluster states using linear optics},
  author={Menicucci, Nicolas C},
  journal={Physical Review A—Atomic, Molecular, and Optical Physics},
  volume={83},
  number={6},
  pages={062314},
  year={2011},
  publisher={APS}
}

@article{alexander2016flexible,
  title={Flexible quantum circuits using scalable continuous-variable cluster states},
  author={Alexander, Rafael N and Menicucci, Nicolas C},
  journal={Physical Review A},
  volume={93},
  number={6},
  pages={062326},
  year={2016},
  publisher={APS}
}

@article{dahlberg2018transforming,
  title={Transforming graph states using single-qubit operations},
  author={Dahlberg, Axel and Wehner, Stephanie},
  journal={Philosophical Transactions of the Royal Society A: Mathematical, Physical and Engineering Sciences},
  volume={376},
  number={2123},
  pages={20170325},
  year={2018},
  publisher={The Royal Society Publishing}
}

@article{alexander2018universal,
  title={Universal quantum computation with temporal-mode bilayer square lattices},
  author={Alexander, Rafael N and Yokoyama, Shota and Furusawa, Akira and Menicucci, Nicolas C},
  journal={Physical Review A},
  volume={97},
  number={3},
  pages={032302},
  year={2018},
  publisher={APS}
}

@article{larsen2020architecture,
  title={Architecture and noise analysis of continuous-variable quantum gates using two-dimensional cluster states},
  author={Larsen, Mikkel V and Neergaard-Nielsen, Jonas S and Andersen, Ulrik L},
  journal={Physical Review A},
  volume={102},
  number={4},
  pages={042608},
  year={2020},
  publisher={APS}
}

@article{yoshikawa2008demonstration,
  title={Demonstration of a quantum nondemolition sum gate},
  author={Yoshikawa, Jun-ichi and Miwa, Yoshichika and Huck, Alexander and Andersen, Ulrik L and van Loock, Peter and Furusawa, Akira},
  journal={Physical Review Letters},
  volume={101},
  number={25},
  pages={250501},
  year={2008},
  publisher={APS}
}

@article{su2013gate,
  title={Gate sequence for continuous variable one-way quantum computation},
  author={Su, Xiaolong and Hao, Shuhong and Deng, Xiaowei and Ma, Lingyu and Wang, Meihong and Jia, Xiaojun and Xie, Changde and Peng, Kunchi},
  journal={Nature communications},
  volume={4},
  number={1},
  pages={2828},
  year={2013},
  publisher={Nature Publishing Group UK London}
}

@article{miwa2009demonstration,
  title={Demonstration of a universal one-way quantum quadratic phase gate},
  author={Miwa, Yoshichika and Yoshikawa, Jun-ichi and van Loock, Peter and Furusawa, Akira},
  journal={Physical Review A—Atomic, Molecular, and Optical Physics},
  volume={80},
  number={5},
  pages={050303},
  year={2009},
  publisher={APS}
}

@article{du2025complete,
  title={A complete continuous-variable quantum computation architecture based on the 2D spatiotemporal cluster state},
  author={Du, Peilin and Zhang, Jing and Zhang, Tiancai and Yang, Rongguo and Gao, Jiangrui},
  journal={Scientific Reports},
  volume={15},
  number={1},
  pages={18199},
  year={2025},
  publisher={Nature Publishing Group UK London}
}

@article{walshe2021streamlined,
  title={Streamlined quantum computing with macronode cluster states},
  author={Walshe, Blayney W and Alexander, Rafael N and Menicucci, Nicolas C and Baragiola, Ben Q},
  journal={Physical Review A},
  volume={104},
  number={6},
  pages={062427},
  year={2021},
  publisher={APS}
}

@article{ferrini2015optimization,
  title={Optimization of networks for measurement-based quantum computation},
  author={Ferrini, Giulia and Roslund, Jonathan and Arzani, Francesco and Cai, Yin and Fabre, Claude and Treps, Nicolas},
  journal={Physical Review A},
  volume={91},
  number={3},
  pages={032314},
  year={2015},
  publisher={APS}
}

@article{glauber1963photon,
  title={Photon correlations},
  author={Glauber, Roy J},
  journal={Physical Review Letters},
  volume={10},
  number={3},
  pages={84},
  year={1963},
  publisher={APS}
}

@article{husimi1940some,
  title={Some formal properties of the density matrix},
  author={Husimi, K{\^o}di},
  journal={Proceedings of the Physico-Mathematical Society of Japan. 3rd Series},
  volume={22},
  number={4},
  pages={264--314},
  year={1940},
  publisher={The Physical Society of Japan, The Mathematical Society of Japan}
}

@article{sudarshan1963equivalence,
  title={Equivalence of semiclassical and quantum mechanical descriptions of statistical light beams},
  author={Sudarshan,  E. C. G.},
  journal={Physical Review Letters},
  volume={10},
  number={7},
  pages={277},
  year={1963},
  publisher={APS}
}

@article{clements2016optimal,
  title={Optimal design for universal multiport interferometers},
  author={Clements, William R and Humphreys, Peter C and Metcalf, Benjamin J and Kolthammer, W Steven and Walmsley, Ian A},
  journal={Optica},
  volume={3},
  number={12},
  pages={1460--1465},
  year={2016},
  publisher={Optica Publishing Group}
}

@article{weedbrook2012gaussian,
  title={Gaussian quantum information},
  author={Weedbrook, Christian and Pirandola, Stefano and Garc{\'\i}a-Patr{\'o}n, Ra{\'u}l and Cerf, Nicolas J and Ralph, Timothy C and Shapiro, Jeffrey H and Lloyd, Seth},
  journal={Reviews of Modern Physics},
  volume={84},
  number={2},
  pages={621--669},
  year={2012},
  publisher={APS}
}

@article{ghosh2023complexity,
  title={Complexity phase transitions generated by entanglement},
  author={Ghosh, Soumik and Deshpande, Abhinav and Hangleiter, Dominik and Gorshkov, Alexey V and Fefferman, Bill},
  journal={Physical Review Letters},
  volume={131},
  number={3},
  pages={030601},
  year={2023},
  publisher={APS}
}

@article{marian2012uhlmann,
  title={Uhlmann fidelity between two-mode Gaussian states},
  author={Marian, Paulina and Marian, Tudor A},
  journal={Physical Review A—Atomic, Molecular, and Optical Physics},
  volume={86},
  number={2},
  pages={022340},
  year={2012},
  publisher={APS}
}

@article{menicucci2006universal,
  title={Universal quantum computation with continuous-variable cluster states},
  author={Menicucci, Nicolas C and Van Loock, Peter and Gu, Mile and Weedbrook, Christian and Ralph, Timothy C and Nielsen, Michael A},
  journal={Physical review letters},
  volume={97},
  number={11},
  pages={110501},
  year={2006},
  publisher={APS}
}

@article{mezher2018efficient,
  title={Efficient quantum pseudorandomness with simple graph states},
  author={Mezher, Rawad and Ghalbouni, Joe and Dgheim, Joseph and Markham, Damian},
  journal={Physical Review A},
  volume={97},
  number={2},
  pages={022333},
  year={2018},
  publisher={APS}
}

@article{efimov2021hafnian,
  title={Hafnian of two-parameter matrices},
  author={Efimov, Dmitry},
  journal={arXiv preprint arXiv:2101.09722},
  year={2021}
}

@article{rahimi2016sufficient,
  title={Sufficient conditions for efficient classical simulation of quantum optics},
  author={Rahimi-Keshari, Saleh and Ralph, Timothy C and Caves, Carlton M},
  journal={Physical Review X},
  volume={6},
  number={2},
  pages={021039},
  year={2016},
  publisher={APS}
}

@inproceedings{engelfriet1996concatenation,
  title={Concatenation of graphs},
  author={Engelfriet, Joost and Vereijken, Jan Joris},
  booktitle={Graph Grammars and Their Application to Computer Science: 5th International Workshop Williamsburg, VA, USA, November 13--18, 1994 Selected Papers 5},
  pages={368--382},
  year={1996},
  organization={Springer}
}

@article{kruse2019detailed,
  title={Detailed study of Gaussian boson sampling},
  author={Kruse, Regina and Hamilton, Craig S and Sansoni, Linda and Barkhofen, Sonja and Silberhorn, Christine and Jex, Igor},
  journal={Physical Review A},
  volume={100},
  number={3},
  pages={032326},
  year={2019},
  publisher={APS}
}

@article{ukai2011demonstration,
  title={Demonstration of unconditional one-way quantum computations for continuous variables},
  author={Ukai, Ryuji and Iwata, Noriaki and Shimokawa, Yuji and Armstrong, Seiji C and Politi, Alberto and Yoshikawa, Jun-ichi and Van Loock, Peter and Furusawa, Akira},
  journal={Physical review letters},
  volume={106},
  number={24},
  pages={240504},
  year={2011},
  publisher={APS}
}

@article{ukai2011demonstrationCP,
  title={Demonstration of a Controlled-Phase Gate for Continuous-Variable One-Way Quantum Computation},
  author={Ukai, Ryuji and Yokoyama, Shota and Yoshikawa, Jun-ichi and van Loock, Peter and Furusawa, Akira},
  journal={Physical review letters},
  volume={107},
  number={25},
  pages={250501},
  year={2011},
  publisher={APS}
}

@article{reck1994experimental,
  title={Experimental realization of any discrete unitary operator},
  author={Reck, Michael and Zeilinger, Anton and Bernstein, Herbert J and Bertani, Philip},
  journal={Physical review letters},
  volume={73},
  number={1},
  pages={58},
  year={1994},
  publisher={APS}
}

@article{ukai2010universal,
  title={Universal linear Bogoliubov transformations through one-way quantum computation},
  author={Ukai, Ryuji and Yoshikawa, Jun-ichi and Iwata, Noriaki and van Loock, Peter and Furusawa, Akira},
  journal={Physical Review A—Atomic, Molecular, and Optical Physics},
  volume={81},
  number={3},
  pages={032315},
  year={2010},
  publisher={APS}
}

@article{pysher2011parallel,
  title={Parallel generation of quadripartite cluster entanglement in the optical frequency comb},
  author={Pysher, Matthew and Miwa, Yoshichika and Shahrokhshahi, Reihaneh and Bloomer, Russell and Pfister, Olivier},
  journal={Physical review letters},
  volume={107},
  number={3},
  pages={030505},
  year={2011},
  publisher={APS}
}

@article{yokoyama2013ultra,
  title={Ultra-large-scale continuous-variable cluster states multiplexed in the time domain},
  author={Yokoyama, Shota and Ukai, Ryuji and Armstrong, Seiji C and Sornphiphatphong, Chanond and Kaji, Toshiyuki and Suzuki, Shigenari and Yoshikawa, Jun-ichi and Yonezawa, Hidehiro and Menicucci, Nicolas C and Furusawa, Akira},
  journal={Nature Photonics},
  volume={7},
  number={12},
  pages={982--986},
  year={2013},
  publisher={Nature Publishing Group UK London}
}

@article{chen2014experimental,
  title={Experimental realization of multipartite entanglement of 60 modes of a quantum optical frequency comb},
  author={Chen, Moran and Menicucci, Nicolas C and Pfister, Olivier},
  journal={Physical review letters},
  volume={112},
  number={12},
  pages={120505},
  year={2014},
  publisher={APS}
}

@article{yoshikawa2016invited,
  title={Invited article: Generation of one-million-mode continuous-variable cluster state by unlimited time-domain multiplexing},
  author={Yoshikawa, Jun-ichi and Yokoyama, Shota and Kaji, Toshiyuki and Sornphiphatphong, Chanond and Shiozawa, Yu and Makino, Kenzo and Furusawa, Akira},
  journal={APL photonics},
  volume={1},
  number={6},
  year={2016},
  publisher={AIP Publishing}
}

@article{larsen2019deterministic,
  title={Deterministic generation of a two-dimensional cluster state},
  author={Larsen, Mikkel V and Guo, Xueshi and Breum, Casper R and Neergaard-Nielsen, Jonas S and Andersen, Ulrik L},
  journal={Science},
  volume={366},
  number={6463},
  pages={369--372},
  year={2019},
  publisher={American Association for the Advancement of Science}
}

@article{asavanant2019generation,
  title={Generation of time-domain-multiplexed two-dimensional cluster state},
  author={Asavanant, Warit and Shiozawa, Yu and Yokoyama, Shota and Charoensombutamon, Baramee and Emura, Hiroki and Alexander, Rafael N and Takeda, Shuntaro and Yoshikawa, Jun-ichi and Menicucci, Nicolas C and Yonezawa, Hidehiro and others},
  journal={Science},
  volume={366},
  number={6463},
  pages={373--376},
  year={2019},
  publisher={American Association for the Advancement of Science}
}

@article{larsen2021deterministic,
  title={Deterministic multi-mode gates on a scalable photonic quantum computing platform},
  author={Larsen, Mikkel V and Guo, Xueshi and Breum, Casper R and Neergaard-Nielsen, Jonas S and Andersen, Ulrik L},
  journal={Nature Physics},
  volume={17},
  number={9},
  pages={1018--1023},
  year={2021},
  publisher={Nature Publishing Group UK London}
}

@article{asavanant2021time,
  title={Time-domain-multiplexed measurement-based quantum operations with 25-MHz clock frequency},
  author={Asavanant, Warit and Charoensombutamon, Baramee and Yokoyama, Shota and Ebihara, Takeru and Nakamura, Tomohiro and Alexander, Rafael N and Endo, Mamoru and Yoshikawa, Jun-ichi and Menicucci, Nicolas C and Yonezawa, Hidehiro and others},
  journal={Physical Review Applied},
  volume={16},
  number={3},
  pages={034005},
  year={2021},
  publisher={APS}
}

@article{du2023generation,
  title={Generation of large-scale continuous-variable cluster states multiplexed both in time and frequency domains},
  author={Du, Peilin and Wang, Yu and Liu, Kui and Yang, Rongguo and Zhang, Jing},
  journal={Optics Express},
  volume={31},
  number={5},
  pages={7535--7544},
  year={2023},
  publisher={Optica Publishing Group}
}

@article{wang2024chip,
  title={Chip-scale generation of 60-mode continuous-variable cluster states},
  author={Wang, Ze and Li, Kangkang and Wang, Yue and Zhou, Xin and Cheng, Yinke and Jing, Boxuan and Sun, Fengxiao and Li, Jincheng and Li, Zhilin and Gong, Qihuang and others},
  journal={arXiv e-prints},
  pages={arXiv--2406},
  year={2024}
}

@article{roh2025generation,
  title={Generation of three-dimensional cluster entangled state},
  author={Roh, Chan and Gwak, Geunhee and Yoon, Young-Do and Ra, Young-Sik},
  journal={Nature Photonics},
  pages={1--7},
  year={2025},
  publisher={Nature Publishing Group UK London}
}

@article{jia2025continuous,
  title={Continuous-variable multipartite entanglement in an integrated microcomb},
  author={Jia, Xinyu and Zhai, Chonghao and Zhu, Xuezhi and You, Chang and Cao, Yunyun and Zhang, Xuguang and Zheng, Yun and Fu, Zhaorong and Mao, Jun and Dai, Tianxiang and others},
  journal={Nature},
  pages={1--8},
  year={2025},
  publisher={Nature Publishing Group UK London}
}

@article{lingua2025continuous,
  title={Continuous-variable square-ladder cluster states in a microwave frequency comb},
  author={Lingua, Fabio and Rivera Hern{\'a}ndez, Juan Carlos and Cortinovis, Michele and Haviland, David B},
  journal={Physical Review Letters},
  volume={134},
  number={18},
  pages={183602},
  year={2025},
  publisher={APS}
}

@article{cai2017multimode,
  title={Multimode entanglement in reconfigurable graph states using optical frequency combs},
  author={Cai, Yin and Roslund, Jonathan and Ferrini, Giulia and Arzani, Francesco and Xu, X and Fabre, Claude and Treps, Nicolas},
  journal={Nature communications},
  volume={8},
  number={1},
  pages={15645},
  year={2017},
  publisher={Nature Publishing Group UK London}
}

@article{fukui2018high,
  title={High-threshold fault-tolerant quantum computation with analog quantum error correction},
  author={Fukui, Kosuke and Tomita, Akihisa and Okamoto, Atsushi and Fujii, Keisuke},
  journal={Physical review X},
  volume={8},
  number={2},
  pages={021054},
  year={2018},
  publisher={APS}
}

@article{wu2020quantum,
  title={Quantum computing with multidimensional continuous-variable cluster states in a scalable photonic platform},
  author={Wu, Bo-Han and Alexander, Rafael N and Liu, Shuai and Zhang, Zheshen},
  journal={Physical Review Research},
  volume={2},
  number={2},
  pages={023138},
  year={2020},
  publisher={APS}
}

@article{spedalieri2012limit,
  title={A limit formula for the quantum fidelity},
  author={Spedalieri, Gaetana and Weedbrook, Christian and Pirandola, Stefano},
  journal={Journal of Physics A: Mathematical and Theoretical},
  volume={46},
  number={2},
  pages={025304},
  year={2012},
  publisher={IOP Publishing}
}

@article{go2025sufficient,
  title={Sufficient conditions for hardness of lossy Gaussian boson sampling},
  author={Go, Byeongseon and Oh, Changhun and Jeong, Hyunseok},
  journal={arXiv preprint arXiv:2511.07853},
  year={2025}
}

@article{alexander2017measurement,
  title={Measurement-based linear optics},
  author={Alexander, Rafael N and Gabay, Natasha C and Rohde, Peter P and Menicucci, Nicolas C},
  journal={Physical Review Letters},
  volume={118},
  number={11},
  pages={110503},
  year={2017},
  publisher={APS}
}

@article{briegel2001persistent,
  title={Persistent entanglement in arrays of interacting particles},
  author={Briegel, Hans J and Raussendorf, Robert},
  journal={Physical Review Letters},
  volume={86},
  number={5},
  pages={910},
  year={2001},
  publisher={APS}
}

@article{fuchs1999cryptographic,
  title={Cryptographic distinguishability measures for quantum-mechanical states},
  author={Fuchs, Christopher A and Van De Graaf, Jeroen},
  journal={IEEE Transactions on Information Theory},
  volume={45},
  number={4},
  pages={1216--1227},
  year={1999},
  publisher={IEEE}
}

@article{ehrenberg2025second,
  title={Second moment of Hafnians in Gaussian boson sampling},
  author={Ehrenberg, Adam and Iosue, Joseph T and Deshpande, Abhinav and Hangleiter, Dominik and Gorshkov, Alexey V},
  journal={Physical Review A},
  volume={111},
  number={4},
  pages={042412},
  year={2025},
  publisher={APS}
}

\let\addcontentsline\oldaddcontentsline

\end{document}


\definecolor{red}{RGB}{255,0,0}
\preprint{APS/123-QED}

\onecolumngrid
\clearpage
\widetext	

\begin{center}
\textbf{\large Supplemental Material for ``Complexity phase transition for continuous-variable cluster state''}
\end{center}

\setcounter{page}{1}
\renewcommand{\thesection}{S\arabic{section}}
\renewcommand{\theequation}{S\arabic{equation}}
\renewcommand{\thefigure}{S\arabic{figure}}
\renewcommand{\thetable}{S\arabic{table}}

\addtocontents{toc}{\protect\setcounter{tocdepth}{0}}

\tableofcontents


\clearpage

\section{Conventions for graph}\label{appendix: section: graph convention}

We review the conventions for graph operations that are introduced in~\cite{engelfriet1996concatenation}.
For a graph $G = (V,E)$ with vertices $V$ and edges $E$, one can define the set of \textit{begin} nodes $\text{begin}(G)$ and \textit{end} nodes $\text{end}(G)$ as the subset of $V$ with implicitly designated sequence. These subsequently correspond to the input and output modes (with designated sequence) of the cluster state.

Using these conventions for the begin and end nodes of a graph, we can define the graph operations of concatenation and sum.
For given graphs $G$ and $G'$ with $|\text{begin}(G')| = |\text{end}(G)| = k$, we define the \textit{concatenation} $G \circ G'$ as the graph obtained by first taking the disjoint union of $G$ and $G'$, and then identifying the $i$-th end node of $G$ with the $i$-th begin node of $G'$ for every
$i \in [k]$, such that $\text{begin}(G \circ G') =  \text{begin}(G)$ and $\text{end}(G \circ G') = \text{end}(G')$.
Also, for given graphs $G$ and $G'$, we define the \textit{sum} $G \oplus G'$ as their disjoint union, such that $\text{begin}(G \oplus G') =  \text{begin}(G)\cdot \text{begin}(G')$ and $\text{end}(G \oplus G') = \text{end}(G)\cdot\text{end}(G')$, where now $\cdot$ denotes the concatenation of sequences.

With these definitions established, we construct the fundamental graphs utilized throughout this work. 
First, we define $G_{H}^{(l)}$ as a (horizontal) one-dimensional chain graph with $l$ nodes, where $\text{begin}(G_{H}^{(l)})$ and $\text{end}(G_{H}^{(l)})$ corresponds to each of the end nodes of the one-dimensional chain graph, respectively (such that $|\text{begin}(G_{H}^{(l)})| = |\text{end}(G_{H}^{(l)})| = 1$). 
Similarly, we define $G_{V}^{(l)}$ as a (vertical) one-dimensional chain graph with $l \geq 2$ nodes, where now both $\text{begin}(G_{V}^{(l)})$ and $\text{end}(G_{V}^{(l)})$ are identified with the two end vertices of the chain (such that, $\text{begin}(G_{H}^{(l)}) = \text{end}(G_{H}^{(l)})$ with $|\text{begin}(G_{H}^{(l)})| = |\text{end}(G_{H}^{(l)})| = 2$). 
Note that $G_{H}^{(l)}$ and $G_{V}^{(l)}$ correspond to the identical graph (i.e., one-dimensional chain with $l$ number of vertices), but their begin and end nodes are differently assigned.


To proceed, we introduce graphs that will be used throughout this work.
We define the specific ``$\sqsupset$-shaped" graph $G_{\sqsupset}$ structured as 
\begin{align}\label{eq: def of G_subseteq}
    G_{\sqsupset} = (G_{H}^{(5)} \oplus G_{H}^{(5)})\circ G_{V}^{(3)}, 
\end{align} 
composed of total $11$ vertices with $|\text{begin}(G_{\sqsupset})| = |\text{end}(G_{\sqsupset})| = 2$ (see Fig.~\ref{fig: graph G_subseteq} for the illustration).
Using $G_{\sqsupset}$ we define the \textit{brick} graph $G_{T}$ as three-concatenations of $G_{\sqsupset}$, such that 
\begin{align}\label{eq: def of G_T}
    G_{T} = G_{\sqsupset} \circ G_{\sqsupset} \circ G_{\sqsupset},
\end{align}
composed of total $29$ vertices with $|\text{begin}(G_{T})| = |\text{end}(G_{T})| = 2$ (see Fig.~\ref{fig: brickwork graph G_T} for the illustration).
Using brick graph $G_{T}$ we define the \textit{unit-depth} graph $G_{D}^{(M)}$ as a parallel composition of $G_{T}$, given by 
\begin{align}\label{eq: def of G_D^M}
    G_{D}^{(M)} = ( \underbrace{G_{T} \oplus G_{T} \oplus \cdots \oplus G_{T}}_{\frac{M}{2} \; \text{times}} )\circ (G_{H}^{(13)} \oplus \underbrace{G_{T}\oplus G_{T}\oplus \cdots \oplus G_{T}}_{\frac{M}{2} - 1 \; \text{times}} \oplus \, G_{H}^{(13)}) ,
\end{align}
where we set $M$ as even for simplicity.
Here, $|\text{begin}(G_{D}^{(M)})| = |\text{end}(G_{D}^{(M)})| = M$, such that the total number of vertices of $G_{D}^{(M)}$ is $\Theta(M)$. 
Finally, we define the \textit{$k$-depth} graph $G_{D}^{(M,k)}$ as the concatenation of $k$ copies of $G_{D}^{(M)}$ such that 
\begin{align}\label{eq: def of G_D^Mk}
    G_{D}^{(M,k)} = \underbrace{G_{D}^{(M)} \circ G_{D}^{(M)}\circ \cdots \circ G_{D}^{(M)}}_{k\;\text{times}},
\end{align}
which contains $\Omega(Mk)$ vertices.
As it turns out, the cluster state $\ket{G_{D}^{(M)}}$ corresponding to the $k$-depth graph $G_{D}^{(M)}$ for sufficiently large $k$ serves as a resource for a universal $M$-mode LO circuit implementation via a measurement-based computation scheme.

\section{Introduction to measurement-based quantum computation with continuous-variable cluster state}\label{appendix: section: introduction to cluster state}

This section reviews fundamental teleportation protocols on continuous-variable (CV) cluster states, which serve as the foundational building blocks for our measurement-based computation framework.

\subsection{Elementary teleportation scheme using infinitely-squeezed ancilla}

We begin by reviewing the elementary measurement-based teleportation scheme in the ideal setting introduced in~\cite{menicucci2006universal, gu2009quantum}, consisting of a zero-momentum eigenstate $\ket{0}_{p}$, a $CZ$ gate, and homodyne measurement in the $\hat{p}$ basis, as depicted below.

\begin{figure}[h]
\includegraphics[width=0.27\linewidth]{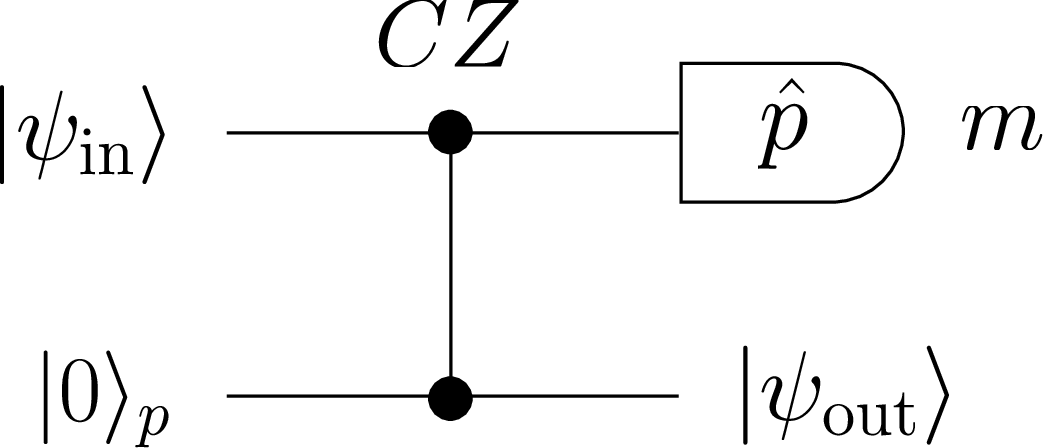}
\caption{Elementary teleportation scheme for infinitely-squeezed ancilla}
\end{figure}
Specifically, given an input state $\ket{\psi_{\rm{in}}}$ combined with an ancillary zero-momentum eigenbasis $\ket{0}_{p}$ as $\ket{\phi}\otimes\ket{0}_{p} = \int ds f(s) \ket{s}_{q}\ket{0}_{p}$ for $f(s) = \bra{s}_{q}\ket{\phi}$, the state after $CZ = e^{i\hat{q}_1\hat{q}_2}$ operation is 
\begin{align}
    CZ (\ket{\phi}\otimes\ket{0}_{p}) &= \int ds f(s) e^{is\hat{q}_2}\ket{s}_{q}\ket{0}_{p} = \int ds f(s)\ket{s}_{q}\ket{s}_{p}.
\end{align}
Conditioning on a $\hat{p}$-basis measurement of the first mode with outcome $m$, the remaining state is projected onto (up to an overall normalization factor)
\begin{align}
    \ket{\psi_{\rm{out}}} = \int ds f(s)\bra{m}_{p}\ket{s}_{q}\ket{s}_{p} = \int ds f(s)e^{-ism}\ket{s}_{p} 
    = e^{-im\hat{p}}\int ds f(s)\ket{s}_{p} 
    = X(m)F\ket{\psi_{\rm{in}}}, \label{eq: elementary teleportation}
\end{align}
where $X(m) = e^{-im\hat{p}}$ denotes a displacement along $\hat{q}$-axis by $m$, and $F = e^{i\pi(\hat{q}^2 + \hat{p}^2)/4}$ the Fourier gate.
Hence, in the ideal case, the input state in the first mode is perfectly teleported to the second mode, up to a phase rotation by  $\frac{\pi}{2}$ and displacement by $m$.

\subsection{Elementary teleportation scheme using finitely-squeezed ancilla}

Because the ideal zero-momentum eigenstate $\ket{0}_{p}$ requires infinite energy and is physically unrealizable, we follow~\cite{menicucci2006universal, gu2009quantum} by employing a finitely-squeezed vacuum state $\ket{r}$ with squeezing level $r$, whose variance along $\hat{p}$-axis is given $\frac{e^{-2r}}{2}$, as a physically realizable approximation to $\ket{0}_{p}$, as illustrated below.

\begin{figure}[h]
\centering
\includegraphics[width=0.27\linewidth]{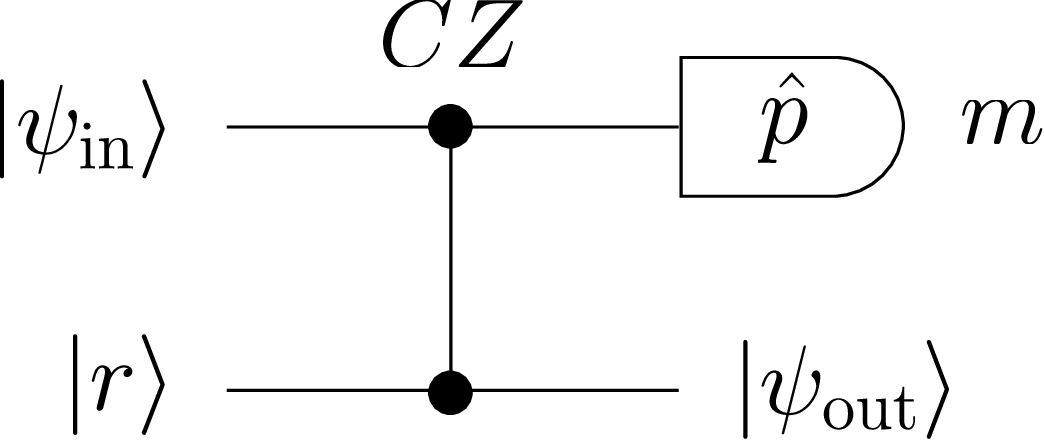}
\caption{Elementary teleportation scheme for finitely-squeezed ancilla}
\label{fig: elementary finite}
\end{figure}

However, teleportation based on $\ket{r}$ is intrinsically imperfect and introduces noise~\cite{menicucci2006universal, gu2009quantum, alexander2014noise, menicucci2014fault}.
To quantify this noise induced by finite squeezing, we follow the quadrature-transformation description of teleportation with a finitely squeezed ancilla established in~\cite{larsen2020architecture, larsen2021deterministic, larsen2021fault}.

Let $(\hat{q}_1,\hat{q}_2,\hat{p}_1,\hat{p}_2)^T$ and $(\hat{q}'_1,\hat{q}'_2,\hat{p}'_1,\hat{p}'_2)^T$ denote the quadrature operators before and after the $CZ = e^{i\hat{q}_1\hat{q}_2}$ gate in Fig.~\ref{fig: elementary finite}, respectively. 
Here, the first mode corresponds to the input state, while the second mode corresponds to the ancilla state, which becomes the output mode after teleportation.
In the Heisenberg picture, they are related by
\begin{align}
    \begin{pmatrix}
        \hat{q}'_1 \\ \hat{q}'_2 \\ \hat{p}'_1 \\ \hat{p}'_2 
    \end{pmatrix} = \begin{pmatrix}
        1 & 0 & 0 & 0 \\ 0 & 1 & 0 & 0 \\ 0 & 1 & 1 & 0 \\ 1 & 0 & 0 & 1
    \end{pmatrix} \begin{pmatrix}
        \hat{q}_1 \\ \hat{q}_2 \\ \hat{p}_1 \\ \hat{p}_2 
    \end{pmatrix}.
\end{align}
Accordingly,  
\begin{align}
    \begin{pmatrix}
        \hat{q}'_2 \\ \hat{p}'_2
    \end{pmatrix} = \begin{pmatrix}
        \hat{q}_2 \\ \hat{q}_1 + \hat{p}_2
    \end{pmatrix} = \begin{pmatrix}
        \hat{p}'_1 - \hat{p}_1 \\ \hat{q}_1 + \hat{p}_2
    \end{pmatrix} = \begin{pmatrix}
        0 & -1 \\ 1 & 0
    \end{pmatrix} \begin{pmatrix}
        \hat{q}_1 \\ \hat{p}_1
    \end{pmatrix} + \begin{pmatrix}
        0 \\ 1
    \end{pmatrix}\hat{p}_2 + \begin{pmatrix}
        1 \\ 0
    \end{pmatrix} \hat{p}'_1   .  \label{eq: input output relation for elementary teleport}
\end{align}
Hence, conditioned on obtaining outcome $m$ from a $\hat{p}$-basis measurement on the first mode, the corresponding input-output relation is
\begin{align}
    \begin{pmatrix}
        \hat{q}'_2 \\ \hat{p}'_2
    \end{pmatrix} = \begin{pmatrix}
        0 & -1 \\ 1 & 0
    \end{pmatrix} \begin{pmatrix}
        \hat{q}_1 \\ \hat{p}_1
    \end{pmatrix} + \begin{pmatrix}
        0 \\ 1
    \end{pmatrix}\hat{p}_2 + \begin{pmatrix}
        1 \\ 0
    \end{pmatrix} m = \textbf{G}\begin{pmatrix}
        \hat{q}_1 \\ \hat{p}_1
    \end{pmatrix} + \textbf{N}\hat{p}_2 + \textbf{D} m   .  \label{eq: input output relation for elementary teleportation}
\end{align}
Here, the matrix $\bf{G}$ in the first term represents the resulting gate operation after teleportation, corresponding to the Fourier gate $F$ in Eq.~\eqref{eq: elementary teleportation}. 
The matrix $\bf{D}$ in the last term represents the displacement that occurred during the teleportation scheme, corresponding to the displacement $X(m)$ in Eq.~\eqref{eq: elementary teleportation}. 
The matrix $\bf{N}$ in the second term represents the noise added to the quadrature due to finite squeezing.
Specifically, $\hat{p}_2$ denotes the $\hat{p}$ quadrature of the squeezed-vacuum ancilla $\ket{r}$ that has zero mean and variance $\langle \hat{p}_2^2\rangle = \frac{e^{-2r}}{2}$, which, after multiplication by $\bf{N}$, induces Gaussian-type noise. 
Indeed, this noise vanishes in the infinite squeezing limit $r\rightarrow \infty$, in which case the input-output relation in Eq.~\eqref{eq: input output relation for elementary teleportation} converges to the ideal teleportation scheme in Eq.~\eqref{eq: elementary teleportation} (i.e., ${\bf{D}} m $ and $\bf{G}$ in Eq.~\eqref{eq: input output relation for elementary teleportation} correspond to the displacement $X(m)$ and the Fourier gate $F$ in Eq.~\eqref{eq: elementary teleportation}, respectively).

Now, in the same setting, consider inserting an additional phase shift before the $\hat{p}$-basis measurement, as illustrated below. 


\begin{figure}[h]
\centering
\includegraphics[width=0.67\linewidth]{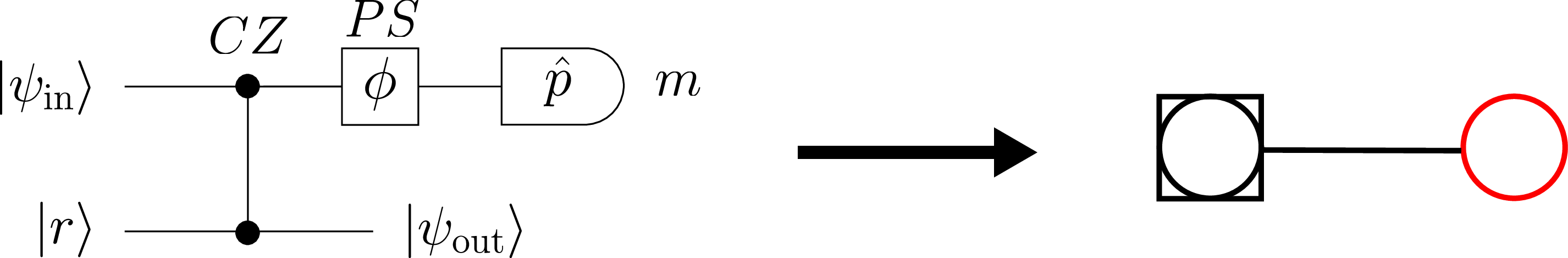}
\caption{Teleportation scheme with phase shift and corresponding graphical notation $G_{H}^{(2)}$. In the graphical notation, the edge of the graph represents the applied $CZ$ gate, and the vertices of the graph represent the modes. Also, the square symbol is to represent the input of the graph corresponding to the input mode (i.e., begin of the graph $G_{H}^{(2)}$), and the red vertex is to represent the output of the graph corresponding to the output mode (i.e., end of the graph $G_{H}^{(2)}$). }
\label{fig: cluster_2mode_PS}
\end{figure}

Let $\phi$ denote the corresponding phase-shift angle.
The quadrature operators before and after the sequential application of the $CZ$ gate and the phase shift are then related by
\begin{align}
    \begin{pmatrix}
        \hat{q}'_1 \\ \hat{q}'_2 \\ \hat{p}'_1 \\ \hat{p}'_2 
    \end{pmatrix} = \begin{pmatrix}
        \cos\phi & -\sin\phi & -\sin\phi & 0 \\ 0 & 1 & 0 & 0 \\ \sin\phi & \cos\phi & \cos\phi & 0 \\ 1 & 0 & 0 & 1
    \end{pmatrix} \begin{pmatrix}
        \hat{q}_1 \\ \hat{q}_2 \\ \hat{p}_1 \\ \hat{p}_2 
    \end{pmatrix}.
\end{align}
Upon measuring the first mode in the $\hat{p}$-basis with outcome $m$, and following the same analysis as in Eq.~\eqref{eq: input output relation for elementary teleport}, we arrive at the input-output relation
\begin{align}
    \begin{pmatrix}
        \hat{q}'_2 \\ \hat{p}'_2
    \end{pmatrix} = \begin{pmatrix}
        -\tan\phi & -1 \\ 1 & 0
    \end{pmatrix} \begin{pmatrix}
        \hat{q}_1 \\ \hat{p}_1
    \end{pmatrix} + \begin{pmatrix}
        0 \\ 1
    \end{pmatrix}\hat{p}_2 + \frac{1}{\cos\phi}\begin{pmatrix}
        1 \\ 0
    \end{pmatrix} m = \textbf{G}\begin{pmatrix}
        \hat{q}_1 \\ \hat{p}_1
    \end{pmatrix} + \textbf{N} \hat{p}_2 + \textbf{D}  m    . \label{eq: input output relation for 2 mode cluster}
\end{align}
Although a single teleportation step realizes a restricted class of operations ($\bf{G}$ in Eq.~\eqref{eq: input output relation for 2 mode cluster}), sequential applications of this primitive suffice to synthesize arbitrary single-mode Gaussian operations~\cite{ukai2010universal}.

It is worth emphasizing that $\tan\phi$ (and likewise $1/\cos\phi$) diverges as $\phi\rightarrow\pm\frac {\pi}{2}$, which can cause both $\bf{G}$ and $\bf{N}$ to diverge in principle. 
However, this limit corresponds to a $\hat{q}$-basis measurement, which removes the measured mode that carries the input state~\cite{larsen2020architecture}. 
Since our main focus is on teleportation-based gate (or circuit) implementation on the input state, this limiting case is not relevant to our results. 



We can also use graphical notation for this measurement-based teleportation scheme, illustrated in Fig.~\ref{fig: cluster_2mode_PS}, as previously used in~\cite{mezher2018efficient}, but with slightly different conventions. 
Specifically, the edges represent the $CZ$ gate, and the vertices represent the modes.
Also, the square symbol is to represent the input of the graph corresponding to the input mode (i.e., $\text{begin}(G)$ of the graph $G$), while the red vertex is to represent the output of the graph corresponding to the output mode (i.e., $\text{end}(G)$ of the graph $G$).

\subsection{Input-output relation for one-dimensional cluster states}\label{appendix: section: example of input output relations}

We now extend the preceding arguments to circuits with \textit{sequential} application of $CZ$ gates, following the framework of~\cite{ukai2010universal} while generalizing their analysis to the finite-squeezing regime.
Note that this setting is naturally described by a one-dimensional CV cluster state $\ket{G}$ associated with a chain graph $G$, where the input state is identified with one or two squeezed-vacuum states used in the cluster-state preparation, depending on whether $G$ corresponds to a vertical or horizontal graph.
In the following, we present several examples of one-dimensional CV cluster states with different choices of input and output modes, and their corresponding input-output relations.

We first consider a three-mode \textit{horizontal} chain with two sequential $CZ$ gate implementations as depicted in Fig.~\ref{fig: horizontal 3 mode cluster state}.
This corresponds to a cluster state $\ket{G_{H}^{(3)}}$ for a graph $G_{H}^{(3)}$ when the input is given by a single-mode squeezed vacuum. 

\begin{figure}[h]
\centering
\includegraphics[width=0.7\linewidth]{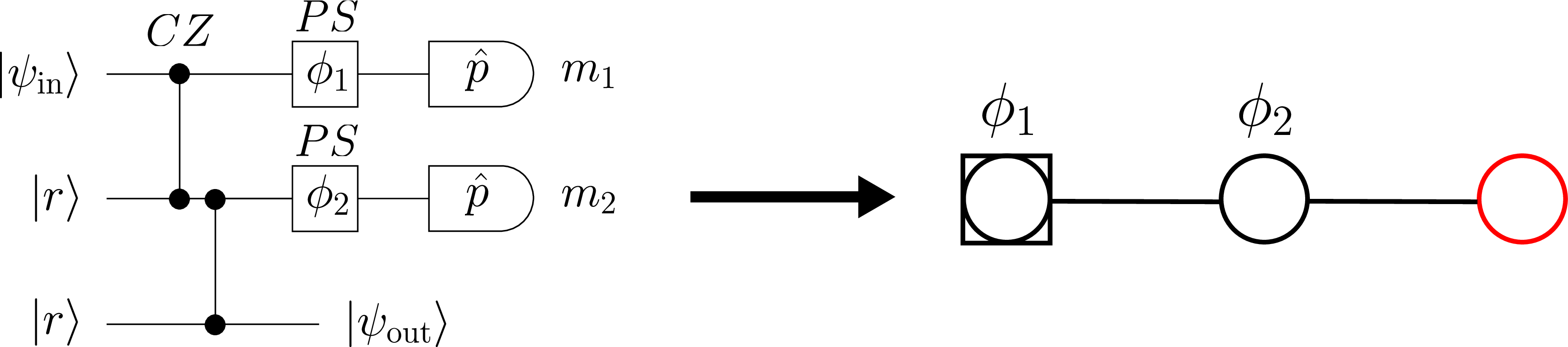}
\caption{Teleportation scheme for three-mode one-dimensional CV cluster state corresponding to the graph $G_{H}^{(3)}$}
\label{fig: horizontal 3 mode cluster state}
\end{figure}

Similarly to before, we perform phase shifts and homodyne measurements for all modes except the output mode.
The input-output relation can then be derived by applying the input-output relation in Eq.~\eqref{eq: input output relation for 2 mode cluster} twice, yielding
\begin{align}
    \begin{pmatrix}
        \hat{q}_{\rm out} \\ \hat{p}_{\rm out}
    \end{pmatrix} &= \begin{pmatrix}
        \tan\phi_2\tan\phi_1 -1 & \tan\phi_2 \\ -\tan\phi_1 & -1
    \end{pmatrix} \begin{pmatrix}
        \hat{q}_{\rm in} \\ \hat{p}_{\rm in}
    \end{pmatrix} + \begin{pmatrix}
        -1 & 0 \\ 0 & 1
    \end{pmatrix}\begin{pmatrix}
        \hat{p}_2 \\ \hat{p}_3
    \end{pmatrix} + \begin{pmatrix}
        -\tan\phi_2 & 1 \\ 1 &  0
    \end{pmatrix} \begin{pmatrix}
        \sec\phi_1 & 0 \\ 0 & \sec\phi_2
    \end{pmatrix} \begin{pmatrix}
        m_1 \\ m_2
    \end{pmatrix} \\ &= \textbf{G}\begin{pmatrix}
        \hat{q}_{\rm in} \\ \hat{p}_{\rm in}
    \end{pmatrix} + \textbf{N}\begin{pmatrix}
        \hat{p}_2 \\ \hat{p}_3
    \end{pmatrix}+ \textbf{D}\begin{pmatrix}
        m_1 \\ m_2
    \end{pmatrix}.
\end{align}
Compared with Eq.~\eqref{eq: input output relation for 2 mode cluster}, the two-measurement teleportation scheme above can realize a broader class of operations $\bf{G}$ than the single-measurement scheme in Fig.~\ref{fig: cluster_2mode_PS} (i.e., its ``expressiveness" increases).
Indeed, Ref.~\cite{ukai2010universal} shows that two repetitions of this procedure, requiring four measurements in total, suffice to realize an arbitrary single-mode Gaussian operation.

We next consider a three-mode \textit{vertical} chain, for which the input and output modes of the measurement-based computation are the same (see Fig.~\ref{fig: vertical 3 mode cluster state}).

\begin{figure}[h]
\centering
\includegraphics[width=0.48\linewidth]{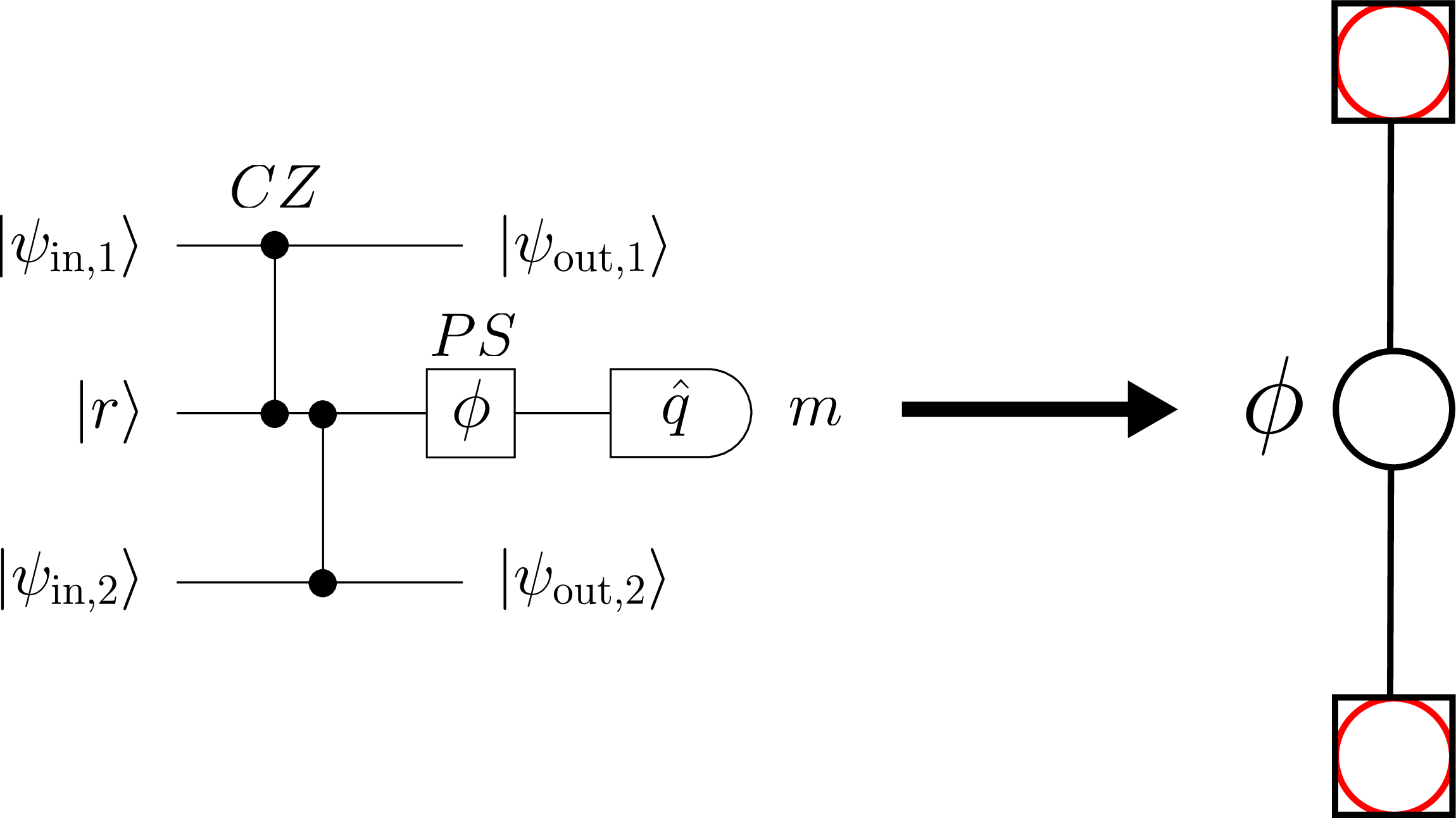}
\caption{Teleportation scheme for three-mode one-dimensional cluster state corresponding to the graph $G_{V}^{(3)}$}
\label{fig: vertical 3 mode cluster state}
\end{figure}

Note that this corresponds to a cluster state $\ket{G_{V}^{(3)}}$ for a graph $G_{V}^{(3)}$ whose input is given by two single-mode squeezed vacuums. 
Then, after applying a phase shift by $\phi$ on the second mode, the quadrature operators for the whole system transform as 
\begin{align}
    \begin{pmatrix}
        \hat{q}'_1 \\ \hat{q}'_2  \\ \hat{q}'_3 \\ \hat{p}'_1 \\ \hat{p}'_2 \\ \hat{p}'_3 
    \end{pmatrix} = \begin{pmatrix}
        1 & 0 & 0 & 0 & 0 & 0 \\
        -\sin\phi & \cos\phi & -\sin\phi & 0 & -\sin\phi & 0 \\ 
        0 & 0 & 1 & 0 & 0 & 0\\ 
        0 & 1 & 0 & 1 & 0 & 0 \\ 
        \cos\phi & \sin\phi & \cos\phi & 0 &  \cos\phi & 0 \\
        0 & 1 & 0 & 0 & 0 & 1
    \end{pmatrix} \begin{pmatrix}
        \hat{q}_1 \\ \hat{q}_2 \\ \hat{q}_3 \\ \hat{p}_1 \\ \hat{p}_2  \\\hat{p}_3  
    \end{pmatrix} 
    \quad \Rightarrow \quad \begin{pmatrix}
        \hat{q}'_1  \\ \hat{q}'_3 \\ \hat{p}'_1 \\ \hat{p}'_3 
    \end{pmatrix} =  \begin{pmatrix}
        \hat{q}_1 \\ \hat{q}_3 \\ \hat{q}_2 + \hat{p}_1  \\ \hat{q}_2 +\hat{p}_3  
    \end{pmatrix}  ,
\end{align}
where $\hat{q}_2$ satisfies the relation $\hat{q}'_2 = -\sin\phi\hat{q}_1 + \cos\phi\hat{q}_2 - \sin\phi\hat{q}_3 -\sin\phi\hat{p}_2$, which gives $\hat{q}_2 = \sec\phi\hat{q}'_2 + \tan\phi\hat{q}_1 + \tan\phi\hat{q}_3 + \tan\phi\hat{q}_2$.
Hence, when measuring the second mode in the $\hat{q}$-basis and obtaining outcome $m$, the corresponding input-output relation is
\begin{align}\label{eq: input output relation for vertical cluster state}
    \begin{pmatrix}
        \hat{q}_{\rm out,1}  \\ \hat{q}_{\rm out,2} \\ \hat{p}_{\rm out,1} \\ \hat{p}_{\rm out,2}
    \end{pmatrix} = \begin{pmatrix}
        1 & 0 & 0 & 0 \\ 0 & 1 & 0 & 0 \\ \tan\phi & \tan\phi & 1 & 0 \\ \tan\phi & \tan\phi & 0 & 1
    \end{pmatrix} \begin{pmatrix}
        \hat{q}_{\rm in,1}  \\ \hat{q}_{\rm in,2} \\ \hat{p}_{\rm in,1} \\ \hat{p}_{\rm in,2}
    \end{pmatrix} + \begin{pmatrix}
        0 \\ 0 \\ \tan\phi \\ \tan\phi 
    \end{pmatrix} \hat{p}_{2} + \begin{pmatrix}
        0 \\ 0 \\ \sec\phi \\ \sec\phi
    \end{pmatrix} m = \textbf{G}  \begin{pmatrix}
        \hat{q}_{\rm in,1}  \\ \hat{q}_{\rm in,2} \\ \hat{p}_{\rm in,1} \\ \hat{p}_{\rm in,2}
    \end{pmatrix} + \textbf{N}   \hat{p}_{2}    + \textbf{D}  m ,
\end{align}
which implements an entangling operation between the two input modes.
Once again, $\bf{G}$ and $\bf{N}$ here can diverge in principle as $\phi\rightarrow\pm\frac {\pi}{2}$.
In the present setting, however, this limit corresponds to a $\hat{p}$-basis measurement, which removes the mode carrying the input state (e.g., mode 1 and 3 in Fig.~\ref{fig: vertical 3 mode cluster state})~\cite{larsen2020architecture}.
Since our interest in this teleportation scheme is to mediate an entangling operation on the input state, this limiting case is not relevant to our results.

To avoid confusion, we stress the following phase-shift convention for measurement-based gate implementations: modes that carry the state to be teleported are, by default, measured in the $\hat{p}$ basis, whereas ancillary modes used solely to mediate entangling operations are, by default, measured in the $\hat{q}$ basis. 
Measurements in the complementary bases are excluded, as they would destroy the teleported state.


\section{Graph-based analysis for input-output relation}\label{appendix: section: conventions for input output relation}

As demonstrated in Sec.~\ref{appendix: section: example of input output relations} and~\cite{larsen2020architecture, larsen2021deterministic, larsen2021fault}, the input-output relation for a measurement-based computation scheme on a CV cluster state is entirely determined by the underlying graph structure (with designated input and output nodes) and the applied phase shifts. 
Hence, the graph conventions introduced in Sec.~\ref{appendix: section: graph convention}, combined with phase-shift angles, can be naturally extended to cluster state analysis for measurement-based computation.
This framework structurally requires the number of begin nodes to equal the number of end nodes, with each end node connected to at least one begin node to ensure a well-defined input-output relation.
For the present analysis, we consider the input state for measurement-based computation (corresponding to the input quadrature) to be ancilla states of the cluster state (i.e., squeezed vacuum states); we later describe in Sec.~\ref{appendix: section: coupling} how we can arbitrarily set the input state for measurement-based computation given a cluster state and the desired input state.

More specifically, consider a graph $G = (V,E)$ with pre-assigned ${\rm{begin}}(G)$ and ${\rm{end}}(G)$ satisfying $|{\rm{begin}}(G)| = |{\rm{end}}(G)|$, where each element of ${\rm{end}}(G)$ is connected with at least one of ${\rm{begin}}(G)$. 
Given such a graph $G$, let $\ket{G}$ be a cluster state with its input and output modes given by ${\rm{begin}}(G)$ and ${\rm{end}}(G)$, respectively.
Also, let $\bm{\phi}$ be a set of phase-shift angles with a designated sequence, with its size given by $|V| - |{\rm{end}}(G)|$.
Then, after phase shift operations according to $\bm{\phi}$ and homodyne measurements on all the modes except the output modes of $\ket{G}$ (i.e., $V \setminus {\rm{end}}(G)$), the output quadrature operators $(\hat{\bm{q}}_{\rm out}, \hat{\bm{p}}_{\rm out})^{T}$ for the output modes (i.e., ${\rm{end}}(G)$) follow the input-output relation 
\begin{align}\label{eq: input output relation 1}
    \begin{pmatrix}
        \hat{\bm{q}}_{\rm out} \\ \hat{\bm{p}}_{\rm out}
    \end{pmatrix} = {\bf{G}} \begin{pmatrix}
        \hat{\bm{q}}_{\rm in} \\ \hat{\bm{p}}_{\rm in}
    \end{pmatrix} + {\bf{N}}   \hat{\bm{p}} + {\bf{D}}   \bm{m} ,
\end{align}
where $(\hat{\bm{q}}_{\rm in}, \hat{\bm{p}}_{\rm in})^{T}$ denotes the quadrature operators for the input modes (i.e., ${\rm{begin}}(G)$), $\hat{\bm{p}}$ denotes the momentum quadrature operators for all the modes except input modes (i.e., $V \setminus {\rm{begin}}(G)$), and $\bm{m}$ denotes the measured outcomes for all the modes except the output modes (i.e., $V \setminus {\rm{end}}(G)$). 
Therefore, given a graph $G$ and a set of phase-shift angles $\bm{\phi}$, based on the measurement-based computation scheme for the cluster state $\ket{G}$ as above, we can determine a symplectic transformation matrix $\bf{G}$, a noise matrix $\bf{N}$, and a displacement matrix $\bf{D}$.
Accordingly, we formulate these matrices in terms of a graph $G$ and phase shift angle $\bm{\phi}$ as follows.

\begin{definition}[Matrices for input-output relation]\label{def: matrices for input output relation}
Given a graph $G$ with $|{\rm{begin}}(G)| = |{\rm{end}}(G)|$, and phase-shift angles $\bm{\phi}$ (with a designated sequence) satisfying $|\bm{\phi}| = |V| - |{\rm{end}}(G)|$, let $\ket{G}$ be the corresponding CV cluster state, with input and output modes of $\ket{G}$ specified by ${\rm{begin}}(G)$ and ${\rm{end}}(G)$, respectively.
By applying phase shifts according to $\bm{\phi}$ and performing homodyne measurements on all the modes except the output modes of $\ket{G}$, the input-output relation in Eq.~\eqref{eq: input output relation 1} is established, where the matrices ${\bf G}$, ${\bf N}$, and ${\bf D}$ in Eq.~\eqref{eq: input output relation 1} are fully determined by the graph $G$ and the choice of general-dyne angles $\bm{\phi}$.
Accordingly, we define ${\bf G}(G,\bm{\phi})$ as the symplectic transformation matrix obtained through this procedure (i.e., $\bf{G}$ in Eq.~\eqref{eq: input output relation 1}).
We further define ${\bf N}(G,\bm{\phi})$ as the noise matrix (i.e., $\bf{N}$ in Eq.~\eqref{eq: input output relation 1}) and ${\bf D}(G,\bm{\phi})$ as the displacement matrix (i.e., $\bf{D}$ in Eq.~\eqref{eq: input output relation 1}) induced during this procedure.
\end{definition}


We now study how the matrices for the input-output relation in Definition~\ref{def: matrices for input output relation} transform under graph operations (concatenation and sum) introduced in Sec.~\ref{appendix: section: graph convention}. 
Let $G_1 = (V_1, E_1)$ and $G_2 = (V_2, E_2)$ be graphs equipped with designated begin and end nodes.
Also, let $\bm{\phi}_1$ be a set of phase-shift angles for $G_1$ assigned for each of $V_1$ except $\text{end}(G_1)$ with designated sequence, and similar for $\bm{\phi}_{2}$ and $G_2$.
For concatenation of graph $G_1$ and $G_2$ with $|\text{begin}(G_2)| = |\text{end}(G_1)|$, together with phase-shift angles $\bm{\phi}_1$ and $\bm{\phi}_2$ for each of $G_1$ and $G_2$, we have the following properties
\begin{align}
    {\bf G}(G_{1} \circ G_{2},\bm{\phi}_1 \cdot \bm{\phi}_2) &= {\bf G}(G_2,\bm{\phi}_2){\bf G}(G_1,\bm{\phi}_1), \label{eq: transformation matrix after graph concatenation} \\
    {\bf N}(G_{1} \circ G_{2},\bm{\phi}_1 \cdot \bm{\phi}_2) &= \left[ {\bf G}(G_2,\bm{\phi}_2) {\bf N}(G_1,\bm{\phi}_1) \;|\; {\bf N}(G_2,\bm{\phi}_2)  \right],  \label{eq: noise matrix after graph concatenation} \\
    {\bf D}(G_{1} \circ G_{2},\bm{\phi}_1 \cdot \bm{\phi}_2) &= \left[ {\bf G}(G_2,\bm{\phi}_2) {\bf D}(G_1,\bm{\phi}_1) \;|\; {\bf D}(G_2,\bm{\phi}_2)  \right], \label{eq: displacement matrix after graph concatenation} 
\end{align}
where $\cdot$ denotes the concatenation of sequences, and $[ A \;|\; B]$ denotes the concatenation of the matrices $A$ and $B$ along column axis. 
Also, for the sum of graphs $G_1$ and $G_2$, we have the following properties
\begin{align}
    {\bf G}(G_{1} \oplus G_{2},\bm{\phi}_1 \cdot \bm{\phi}_2) &= P({\bf G}(G_1,\bm{\phi}_1) \oplus {\bf G}(G_2,\bm{\phi}_2) )P^{T} ,\label{eq: transformation matrix after graph sum}\\
    {\bf N}(G_{1} \oplus G_{2},\bm{\phi}_1 \cdot \bm{\phi}_2) &=  P({\bf N}(G_1,\bm{\phi}_1) \oplus {\bf N}(G_2,\bm{\phi}_2) ) ,\label{eq: noise matrix after graph sum} \\
    {\bf D}(G_{1} \oplus G_{2},\bm{\phi}_1 \cdot \bm{\phi}_2) &= P({\bf D}(G_1,\bm{\phi}_1) \oplus {\bf D}(G_2,\bm{\phi}_2) )  , \label{eq: displacement matrix after graph sum}
\end{align}
where $P$ represents the permutation that reorders the quadrature operators from $(\hat{\bm{q}}_{1} \; \hat{\bm{p}}_{1} \; \hat{\bm{q}}_{2} \; \hat{\bm{p}}_{2})^{T}$ to $(\hat{\bm{q}}_{1} \;  \hat{\bm{q}}_{2} \; \hat{\bm{p}}_{1} \; \hat{\bm{p}}_{2})^{T}$ (i.e., $xpxp$-ordering to $xxpp$-ordering).
More explicitly, $P$ is an $2m\times 2m$ permutation matrix for $m = m_1 + m_2$ with $m_1 = |{\rm{begin}}(G_1)| = |{\rm{end}}(G_1)|$ and $m_2 = |{\rm{begin}}(G_2)| = |{\rm{end}}(G_2)|$, defined by
\begin{align}\label{eq: def of permutation}
    P = \begin{pmatrix}
        \mathbb{I}_{m_1} & 0_{m_1} & 0_{m_2} & 0_{m_2} \\ 0_{m_1} & 0_{m_1} & \mathbb{I}_{m_2} & 0_{m_2} \\ 0_{m_1} & \mathbb{I}_{m_1} & 0_{m_2} & 0_{m_2} \\ 0_{m_1} & 0_{m_1} & 0_{m_2} & \mathbb{I}_{m_2}
    \end{pmatrix} ,
\end{align}
where $\mathbb{I}_{m_1}$ and $0_{m_1}$ are $m_1\times m_1$ identity and all-zero matrix, respectively (similarly for $\mathbb{I}_{m_2}$ and $0_{m_2}$).

\section{Universal cluster state for linear optical circuit implementation}\label{appendix: section: implementing universal}

We present a graph structure of a CV cluster state and the corresponding set of phase-shift angles that enable measurement-based linear optics (MBLO), i.e., the implementation of an arbitrary linear optical (LO) circuit via measurement-based computation on the cluster state. 
We begin by identifying a cluster-state graph structure and a corresponding set of phase-shift angles that implement an arbitrary brick beamsplitter, a fundamental building block for universal LO circuits.
Building on this construction, we then propose a full graph structure and phase-shift angles for a universal LO circuit via MBLO.

\subsection{Cluster state for brick beam splitters}

We first introduce the graph structure of a CV cluster state and the corresponding phase-shift angles that can implement building blocks for universal LO circuits. 
Based on the foundational works on constructing LO circuits~\cite{reck1994experimental, clements2016optimal}, an arbitrary LO circuit can be decomposed into parameterized \textit{brick} beam splitters, as shown in Fig.~\ref{fig: LOdecomposition} below.

\begin{figure}[h]
\centering
\includegraphics[width=0.95\linewidth]{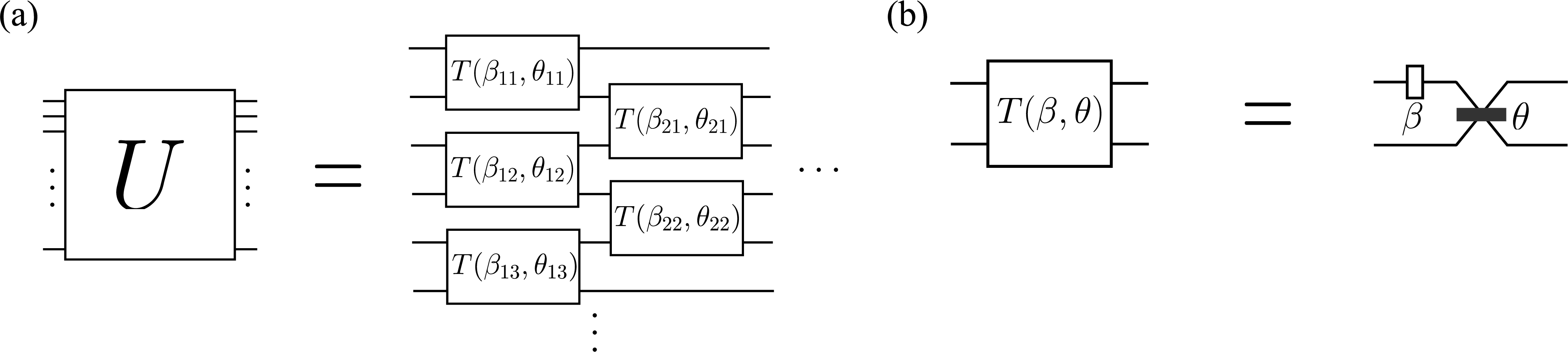}
\caption{(a) Brickwork decomposition of LO circuit by~\cite{clements2016optimal}. 
An arbitrary LO circuit $U$ can be decomposed into a brickwork LO circuit composed of parameterized brick beam splitters.
(b) Brick beam splitter defined in Eq.~\eqref{eq: brickwork beamsplitter T}, representing a beam-splitter operation with transmissivity $\cos\theta$, preceded by a phase shift $\beta$ on the first mode. 
}
\label{fig: LOdecomposition}
\end{figure}

Here, each brick beam splitter composing the LO circuit (see Fig~\ref{fig: LOdecomposition}(b)) can be written as 
\begin{align}\label{eq: brickwork beamsplitter T}
    T = \begin{pmatrix}
        e^{i\beta}\cos\theta & -\sin\theta \\ e^{i\beta}\sin\theta & \cos\theta
    \end{pmatrix} ,
\end{align}
with $\theta \in [0, \frac{\pi}{2}]$ and $\beta \in [0, 2\pi]$.
In the following, we propose a cluster state and corresponding phase-shift angles to implement an arbitrary brick beam splitter $T$.

First, consider the cluster state $\ket{G_{\sqsupset}}$ corresponding to the graph $G_{\sqsupset} = (G_{H}^{(5)} \oplus G_{H}^{(5)})\circ G_{V}^{(3)}$ defined in Eq.~\eqref{eq: def of G_subseteq}, illustrated in Fig.~\ref{fig: graph G_subseteq} below.

\begin{figure}[h]
\centering
\includegraphics[width=0.35\linewidth]{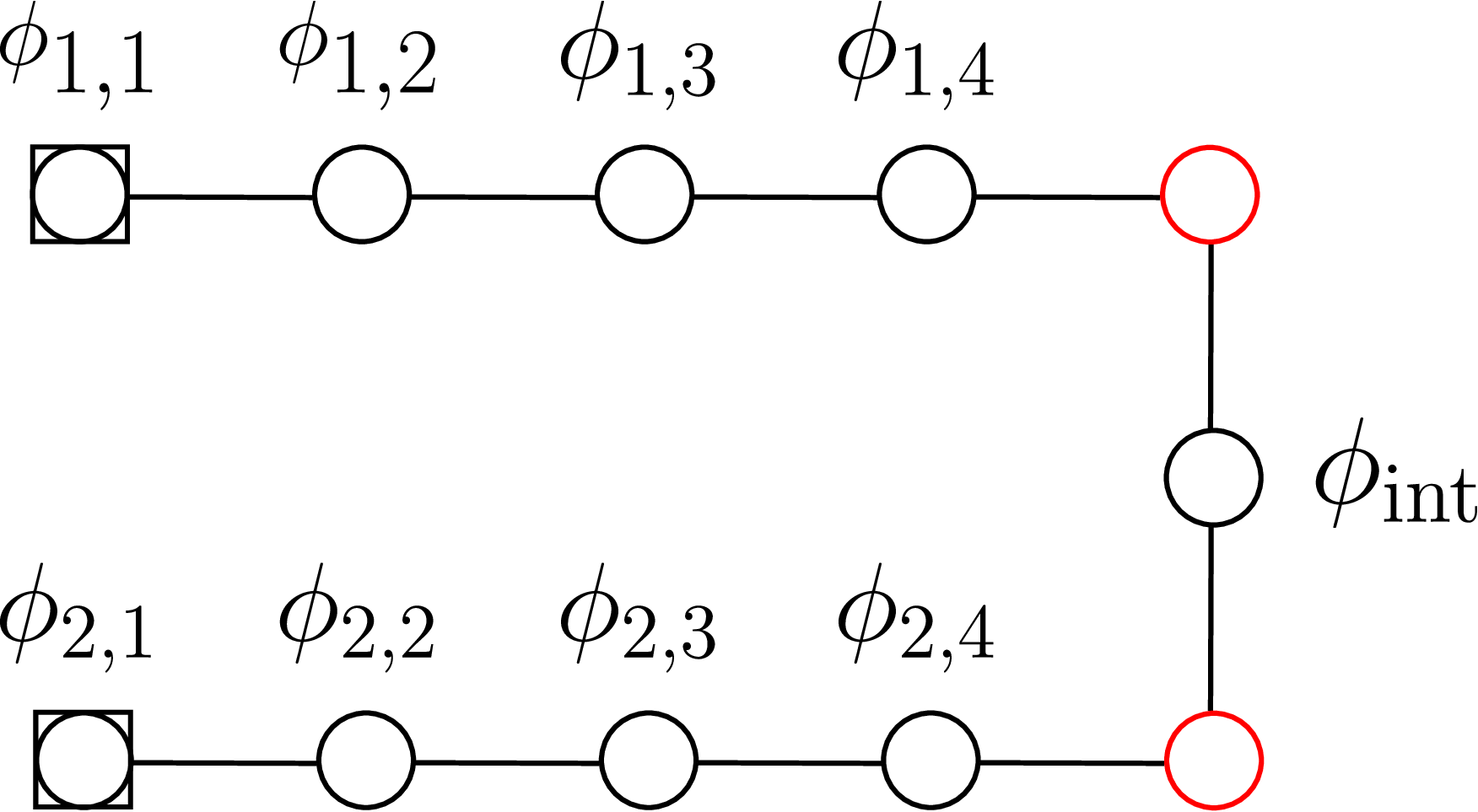}
\caption{Schematic for a cluster state corresponding to the graph $G_{\sqsupset}$ defined in Eq.~\eqref{eq: def of G_subseteq}, with assigned input/output and phase-shift angles.}
\label{fig: graph G_subseteq}
\end{figure}

Also, we choose the phase-shift angles in Fig.~\ref{fig: graph G_subseteq} as follows:
\begin{align}\label{eq: phase-shift angles a for brickwork}
    \begin{pmatrix}
        \tan\phi_{1,1} \\ \tan\phi_{1,2} \\ \tan\phi_{1,3} \\ \tan\phi_{1,4}
    \end{pmatrix} = \begin{pmatrix}
        -\frac{1}{2} \\ - \frac{1 - (-\cos\theta + \sin\theta )}{1 - \frac{1}{2}(-\cos\theta + \sin\theta)} \\ -1 + \frac{1}{2}(-\cos\theta + \sin\theta) \\ -\frac{\frac{1}{2}}{1 - \frac{1}{2}(-\cos\theta + \sin\theta ) }
    \end{pmatrix},\quad  
    \begin{pmatrix}
        \tan\phi_{2,1} \\ \tan\phi_{2,2} \\ \tan\phi_{2,3} \\ \tan\phi_{2,4}
    \end{pmatrix} = \begin{pmatrix}
        -\frac{1}{2} \\ - \frac{1 - (\cos\theta + \sin\theta )}{1 - \frac{1}{2}(\cos\theta + \sin\theta)} \\ -1 + \frac{1}{2}(\cos\theta + \sin\theta) \\ -\frac{\frac{1}{2}}{1 - \frac{1}{2}(\cos\theta + \sin\theta ) }
    \end{pmatrix},\quad 
    \tan\phi_{\text{int}} = - \sin\theta ,
\end{align}
Then, for the graph $G_{\sqsupset}= (G_{H}^{(5)} \oplus G_{H}^{(5)})\circ G_{V}^{(3)}$ and the corresponding phase-shift angles $\bm{\phi}_{a} = \{\phi_{1,1},\dots,\phi_{1,4},\phi_{2,1},\dots, \phi_{2,4}, \phi_{\text{int}}\}$ in Eq.~\eqref{eq: phase-shift angles a for brickwork}, using the transformation rule depicted in Eq.~\eqref{eq: transformation matrix after graph concatenation} and Eq.~\eqref{eq: transformation matrix after graph sum}, we obtain
\begin{align}
    {\bf G}(G_{\sqsupset}, \bm{\phi}_{a}) = \begin{pmatrix}
        0 & 0 & 1 & 0 \\ 0 & 0 & 0 & 1 \\ -1 & 0 & -\cos\theta & -\sin\theta \\ 0 & -1 & -\sin\theta & \cos\theta
    \end{pmatrix} \quad \Rightarrow \quad 
    {\bf G}( G_{\sqsupset} \circ G_{\sqsupset} \circ G_{\sqsupset}, \bm{\phi}_{a}\cdot \bm{\phi}_{a} \cdot \bm{\phi}_{a}) = \begin{pmatrix}
       \cos\theta & \sin\theta & 0 & 0 \\ \sin\theta & -\cos\theta & 0 & 0 \\ 0 & 0 & \cos\theta & \sin\theta \\ 0 & 0 & \sin\theta & -\cos\theta
    \end{pmatrix} .
\end{align}
Indeed, this gives an arbitrary beam-splitter operation parameterized by $ \theta$ when repeated three times, with an additional $ \pi$ phase shift on the second mode.
Hence, by preparing cluster state $\ket{G_{T}}$ with its graph $G_{T} = G_{\sqsupset} \circ G_{\sqsupset} \circ G_{\sqsupset}$ defined in Eq.~\eqref{eq: def of G_T} (see Fig.~\ref{fig: brickwork graph G_T} for illustration of $G_{T}$) and performing homodyne measurements after applying phase shifts according to $\bm{\phi}_{a}\cdot \bm{\phi}_{a} \cdot \bm{\phi}_{a}$, one can realize an arbitrary beam splitter operation.

Based on this understanding, one can further implement the brick beam splitter $T$ in Eq.~\eqref{eq: brickwork beamsplitter T} by implementing an additional phase shift $\beta$ before the beam splitter operation depicted above. 
More specifically, to implement a phase shift angle $\beta$ on the first mode (and undo the unwanted $\pi$-phase shift described above), we consider the following phase-shift angles in Fig~\ref{fig: graph G_subseteq}:

\begin{align}
    &\begin{pmatrix}
        \tan\phi_{1,1} \\ \tan\phi_{1,2} \\ \tan\phi_{1,3} \\ \tan\phi_{1,4}
    \end{pmatrix} = \begin{pmatrix}
        \frac{-\cos\beta + \sqrt{2}\sin\left( \theta - \frac{\pi}{4} \right) \sin\beta}{| -\cos\beta + \sqrt{2}\sin\left( \theta - \frac{\pi}{4} \right) \sin\beta |} \frac{| -\cos\beta + \sqrt{2}\sin\left( \theta - \frac{\pi}{4} \right) \sin\beta | - \sqrt{1 + 2\sin^2\left( \theta - \frac{\pi}{4} \right)}}{\sin\beta + \sqrt{2}\sin\left( \theta - \frac{\pi}{4} \right) \cos\beta} 
        \\
        \frac{-\cos\beta + \sqrt{2}\sin\left( \theta - \frac{\pi}{4} \right) \sin\beta}{| -\cos\beta + \sqrt{2}\sin\left( \theta - \frac{\pi}{4} \right) \sin\beta |}  \frac{1 - (\sin\beta + \sqrt{2}\sin\left( \theta - \frac{\pi}{4} \right) \cos\beta )   }{\sqrt{1 + 2\sin^2\left( \theta - \frac{\pi}{4} \right)}}  
        \\  
        \frac{-\cos\beta + \sqrt{2}\sin\left( \theta - \frac{\pi}{4} \right) \sin\beta}{| -\cos\beta + \sqrt{2}\sin\left( \theta - \frac{\pi}{4} \right) \sin\beta | } \sqrt{1 + 2\sin^2\left( \theta - \frac{\pi}{4} \right)} 
        \\
        - \frac{\cos\beta}{\sin\beta + \sqrt{2}\sin\left( \theta - \frac{\pi}{4} \right) \cos\beta}
        -
        \frac{-\cos\beta + \sqrt{2}\sin\left( \theta - \frac{\pi}{4} \right) \sin\beta}{| -\cos\beta + \sqrt{2}\sin\left( \theta - \frac{\pi}{4} \right) \sin\beta |}  \frac{1 - (\sin\beta + \sqrt{2}\sin\left( \theta - \frac{\pi}{4} \right) \cos\beta )   }{(\sin\beta + \sqrt{2}\sin\left( \theta - \frac{\pi}{4} \right) \cos\beta) \sqrt{1 + 2\sin^2\left( \theta - \frac{\pi}{4} \right)}}  
    \end{pmatrix},  \nonumber \\ 
    &\begin{pmatrix}
        \tan\phi_{2,1} \\ \tan\phi_{2,2} \\ \tan\phi_{2,3} \\ \tan\phi_{2,4}
    \end{pmatrix} =\begin{pmatrix}
        \frac{1}{2} \\ \frac{1 + (\cos\theta + \sin\theta )}{1 + \frac{1}{2}(\cos\theta + \sin\theta)} \\ 1 + \frac{1}{2}(\cos\theta + \sin\theta) \\ \frac{\frac{1}{2}}{1 + \frac{1}{2}(\cos\theta + \sin\theta ) }
    \end{pmatrix}, \quad 
    \tan\phi_{\text{int}} =  - \sin\theta   , \label{eq: phase-shift angles b for brickwork}
\end{align}
where we set $x/|x| = 1$ when $x = 0$. 
Similarly, for phase-shift angles $\bm{\phi}_{b} = \{ \phi_{1,1}, \dots, \phi_{1,4}, \phi_{2,1}, \dots, \phi_{2,4}, \phi_{\text{int}}\}$ given as Eq.~\eqref{eq: phase-shift angles b for brickwork}, using the transformation rule depicted in Eq.~\eqref{eq: transformation matrix after graph concatenation} and Eq.~\eqref{eq: transformation matrix after graph sum}, we have
\begin{align}
    {\bf G}(G_{\sqsupset}, \bm{\phi}_{b}) 
    &= \begin{pmatrix} \sin\beta & 0 & \cos\beta & 0 \\ 0 & 0 & 0 & -1 \\ -\cos\beta -\sin\beta\cos\theta & 0 & \sin\beta - \cos\beta\cos\theta & \sin\theta \\ -\sin\beta\sin\theta & 1 & -\cos\beta\sin\theta & -\cos\theta \end{pmatrix}  \\
    &= \begin{pmatrix}
        0 & 0 & 1 & 0 \\ 0 & 0 & 0 & 1 \\ -1 & 0 & -\cos\theta & -\sin\theta \\ 0 & -1 & -\sin\theta & \cos\theta
    \end{pmatrix} \cdot \begin{pmatrix}
        \cos\beta & 0 & -\sin\beta & 0 \\ 0 & -1 & 0 & 0 \\ \sin\beta & 0 & \cos\beta & 0 \\ 0 & 0 & 0 & -1
    \end{pmatrix},
\end{align}
which is indeed a product of ${\bf G}(G_{\sqsupset}, \bm{\phi}_{a})$ and a symplectic transformation corresponding to phase shift operation ($\beta$ on the first mode, $\pi$ on the second mode). 
Then, by graph concatenation $G_{T} = G_{\sqsupset} \circ G_{\sqsupset} \circ G_{\sqsupset}$ and corresponding phase-shift angles $\bm{\phi}_{T} = \bm{\phi}_{b}\cdot \bm{\phi}_{a} \cdot \bm{\phi}_{a}$, we finally obtain the desired operation
\begin{align}\label{eq: measurementbased brickwork beam splitter}
    {\bf G}(G_{T},\bm{\phi}_{T} ) = \begin{pmatrix}
        \cos\beta\cos\theta & -\sin\theta & -\sin\beta\cos\theta & 0 \\ \cos\beta\sin\theta & \cos\theta & -\sin\beta\sin\theta & 0 \\ \sin\beta\cos\theta & 0 & \cos\beta\cos\theta & -\sin\theta \\ \sin\beta\sin\theta & 0 & \cos\beta\sin\theta & \cos\theta
    \end{pmatrix} = \begin{pmatrix}
        T_{R} & -T_{I} \\ T_{I} & T_{R}
    \end{pmatrix} \quad \text{for} \quad T = \begin{pmatrix}
        e^{i\beta}\cos\theta & -\sin\theta \\ e^{i\beta}\sin\theta & \cos\theta
    \end{pmatrix}   ,
\end{align}
where $T_{R}$ and $T_{I}$ are the real part and imaginary part of the brick beam splitter $T$. 
Hence, cluster state $\ket{G_{T}}$ corresponding to the graph $G_{T}$ with phase-shift angles $\bm{\phi}_{T} = \bm{\phi}_{b}\cdot \bm{\phi}_{a} \cdot \bm{\phi}_{a}$ can implement an arbitrary brick beam splitter $T$, represented in Fig.~\ref{fig: brickwork graph G_T} below.


\begin{figure}[h]
\centering
\includegraphics[width=0.9\linewidth]{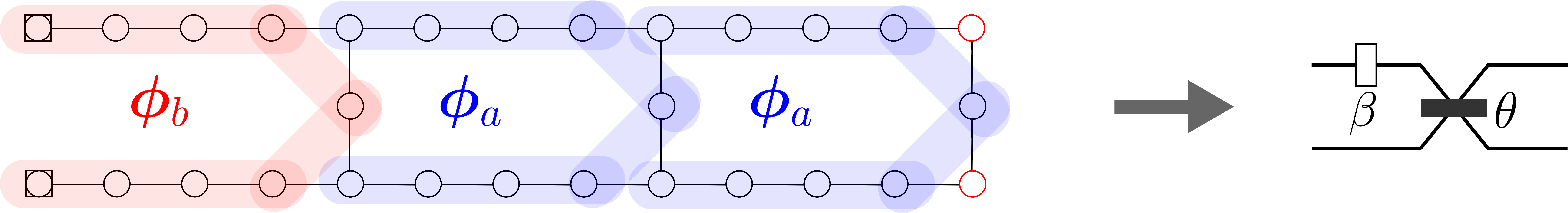}
\caption{Schematic for a cluster state corresponding to the graph $G_{T}$ defined in Eq.~\eqref{eq: def of G_T}, with assigned input/output and phase-shift angles. 
As given in Eq.~\eqref{eq: measurementbased brickwork beam splitter}, $\ket{G_{T}}$ with phase-shift angles $\bm{\phi}_{T} = \bm{\phi}_{b}\cdot \bm{\phi}_{a} \cdot \bm{\phi}_{a}$ for $\bm{\phi}_{a}$ and $\bm{\phi}_{b}$ specified in Eq.~\eqref{eq: phase-shift angles a for brickwork} and Eq.~\eqref{eq: phase-shift angles b for brickwork} can implement an arbitrary brick beam splitter $T$.
}
\label{fig: brickwork graph G_T}
\end{figure}

For analytical convenience, we define the phase-shift angles $\bm{\phi}_{T}$ such that the resulting measurement-based operation on $\ket{G_{T}}$ implements a brick beam splitter.

\begin{definition}\label{def: brickwork cluster state and angle}
    We define $\bm{\phi}_{T} = \bm{\phi}_{b}\cdot \bm{\phi}_{a} \cdot \bm{\phi}_{a}$ as the set of brick phase-shift angles that, combined with the brick graph $G_{T}$ in Eq.~\eqref{eq: def of G_T}, gives the operation ${\bf G}(G_{T},\bm{\phi}_{T} ) = \begin{pmatrix}
        T_{R} & -T_{I} \\ T_{I} & T_{R}
    \end{pmatrix}$ for a brick beam splitter $T = \begin{pmatrix}
        e^{i\beta}\cos\theta & -\sin\theta \\ e^{i\beta}\sin\theta & \cos\theta
    \end{pmatrix}$.
    Here, the phase-shift angles sets $\bm{\phi}_{a} = \{\phi_{1,1},\dots,\phi_{1,4},\phi_{2,1},\dots, \phi_{2,4}, \phi_{\text{int}}\}$ and $\bm{\phi}_{b} = \{\phi_{1,1},\dots,\phi_{1,4},\phi_{2,1},\dots, \phi_{2,4}, \phi_{\text{int}}\}$ are specified in Eq.~\eqref{eq: phase-shift angles a for brickwork} and Eq.~\eqref{eq: phase-shift angles b for brickwork}, respectively.
\end{definition}

\subsection{Cluster state for universal linear optical circuits}\label{section: cluster state for universal LO circuit}

We now describe the graph structure of the cluster state and the corresponding phase-shift angles that can realize universal LO circuits.
Among many decomposition methods for LO circuits, we select the decomposition scheme described in~\cite{clements2016optimal}.
To sketch the LO decomposition scheme in~\cite{clements2016optimal} briefly (see Fig.~\ref{fig: LOdecomposition} for illustration), given an arbitrary $M \times M$ unitary matrix $U$ characterizing an $M$-mode LO circuit (let $M$ be even for convenience), one can construct a set of angles $\{\beta_{i,j}, \theta_{i,j}\}$ for $i \in [M-1] $ and $j \in [\frac{M}{2}]$ corresponding to $M(M-1)/2$-number of brick beam splitters 
\begin{align}
    T_{i,j} = \begin{pmatrix}
        e^{i\beta_{i,j}}\cos\theta_{i,j} & -\sin\theta_{i,j} \\ e^{i\beta_{i,j}}\sin\theta_{i,j} & \cos\theta_{i,j}
    \end{pmatrix}   ,
\end{align}
together with phase-shift angles $\{\beta_{m}\}$ for $m \in [M]$, such that the matrix $U$ can be decomposed by
\begin{align}\label{eq: decomposing unitary circuit}
    U = U_{\rm ps}U_{\frac{M}{2}}U_{\frac{M}{2}-1}\cdots U_{2}U_{1} \quad \text{where} 
    \quad U_{j} = \left( \mathbb{I}_{1} \oplus \left( \bigoplus_{i = \frac{M}{2} + 1}^{M -1} T_{i,j} \right) \oplus \mathbb{I}_1 \right) \left( \bigoplus_{i = 1}^{\frac{M}{2}}T_{i,j} \right) \quad \text{and} \quad U_{\rm ps} = \bigoplus_{m = 1}^{M } e^{i\beta_{m}} .
\end{align}

Now recall that, we define in Sec.~\ref{appendix: section: graph convention} the $k$-depth graph as
\begin{align}\label{eq: recalling k-depth graph}
    G_{D}^{(M,k)} = \underbrace{G_{D}^{(M)} \circ G_{D}^{(M)}\circ \cdots \circ G_{D}^{(M)}}_{k\;\text{times}}  \quad \text{where} \quad G_{D}^{(M)} = ( \underbrace{G_{T} \oplus G_{T} \oplus \cdots \oplus G_{T}}_{\frac{M}{2} \; \text{times}} )\circ (G_{H}^{(13)} \oplus \underbrace{G_{T}\oplus G_{T}\oplus \cdots \oplus G_{T}}_{\frac{M}{2} - 1 \; \text{times}} \oplus \, G_{H}^{(13)}) .
    \quad
\end{align}
Also recall that for each $G_{T}$ composing the graph $G_{D}^{(M,k)}$, the cluster state $\ket{G_{T}}$ with the set of phase-shift angles $\bm{\phi}_{T}$ can implement an arbitrary brick beam splitter $T$.
Comparing Eq.~\eqref{eq: decomposing unitary circuit} and Eq.~\eqref{eq: recalling k-depth graph} reveals that the cluster state $\ket{G_{D}^{(M,k)}}$ for $k \geq \frac{M}{2}$, accompanied with appropriately chosen phase-shift angles, can synthesize the unitary sequence $U_{\frac{M}{2}}U_{\frac{M}{2}-1}\cdots U_{2}U_{1}$ in Eq.~\eqref{eq: decomposing unitary circuit}. 
Moreover, since an arbitrary single-mode phase shift $e^{i\beta_{m}}$ can be implemented using a one-dimensional cluster state $\ket{G_{H}^{(l)}}$ provided $l \geq 5$, as shown in~\cite{ukai2010universal}, the cluster state $\ket{G_{D}^{(M)}}$ with appropriately chosen phase-shift angles can implement the final phase-shift operation $U_{\rm ps}$ in Eq.~\eqref{eq: decomposing unitary circuit} (since $G_{D}^{(M)}$ contains the parallel array of one-dimensional chain graph).
Taken together, these imply that the cluster state $\ket{G_{D}^{(M,k)}}$ with $k \geq \frac{M}{2} + 1$, and with appropriately chosen phase-shift angles, is sufficient to implement an arbitrary LO circuit $U$.


More concretely, consider the brick phase-shift angles $\bm{\phi}_{T_{i,j}}$ defined in Definition~\ref{def: brickwork cluster state and angle} that satisfy
\begin{align}
    {\bf G}(G_{T}, \bm{\phi}_{T_{i,j}}) = \begin{pmatrix}
     (T_{i,j})_{R} & -(T_{i,j})_{I} \\ (T_{i,j})_{I} & (T_{i,j})_{R}
\end{pmatrix}    ,
\end{align}
for the real $(T_{i,j})_{R}$ and imaginary $(T_{i,j})_{R}$ part of $T_{i,j}$. 
From these building blocks, we construct a set of phase-shift angles $\bm{\phi}_{T_{j}}$ such that
\begin{align}
    {\bf G}(G_{D}^{(M)}, \bm{\phi}_{T_{j}}) = \begin{pmatrix}
        (U_{j})_{R} & -(U_j)_{I} \\ (U_j)_{I} & (U_j)_{R} 
    \end{pmatrix} ,
\end{align}
where $(U_{j})_{R}$ and $(U_j)_{I}$ are the real part and imaginary part of the unitary matrix $U_j$ in Eq.~\eqref{eq: decomposing unitary circuit} composing the $j$th layer of $U$.
Specifically, $\bm{\phi}_{T_{j}}$ can be obtained by concatenating the brick phase-shift angles $\bm{\phi}_{T_{i,j}}$ as
\begin{align}
    \bm{\phi}_{T_{j}} = \bm{\phi}_{T_{1,j}}\cdot\bm{\phi}_{T_{2,j}}\cdot \cdots \cdot \bm{\phi}_{T_{\frac{M}{2},j}} \cdot \bm{\phi}_0 \cdot \bm{\phi}_{T_{\frac{M}{2}+1, j}} \cdot \cdots \cdot \bm{\phi}_{T_{M-1, j}}\cdot \bm{\phi}_0 ,
\end{align}
where $\bm{\phi}_0$ is an all-zero phase-shift angles $\bm{\phi}_0 = \{0,\dots,0\}$ which, combined with the horizontal chain graph $G_{H}^{(13)}$, gives the identity matrix ${\bf G}(G_{H}^{(13)}, \bm{\phi}_0) = \mathbb{I}_2$.
Also, following the construction in~\cite{ukai2010universal}, one can explicitly construct a set of phase-shift angles $\bm{\phi}_{\rm ps}$ satisfying
\begin{align}
    {\bf G}(G_{D}^{(M)}, \bm{\phi}_{\rm ps}) = \begin{pmatrix}
        (U_{\rm ps})_{R} & -(U_{\rm ps})_{I} \\ (U_{\rm ps})_{I} & (U_{\rm ps})_{R} 
    \end{pmatrix} ,
\end{align}
by assigning the phase-shift angles specified in~\cite{ukai2010universal} at the beginning modes of $G_{D}^{(M)}$ (to implement each $e^{i\beta_{m}}$) and setting all remaining phase-shift angles to zero.
Hence, one can construct universal phase-shift angles $\bm{\phi}_{U}$ by concatenating $\bm{\phi}_{T_{j}}$ for all $j \in [k]$ with $k \geq \frac{M}{2}$ and $\bm{\phi}_{\rm ps}$ such that
\begin{align}\label{eq: universal graph and angles}
    \bm{\phi}_{U} = \bm{\phi}_{T_{1}}\cdot \cdots\cdot\bm{\phi}_{T_{k}} \cdot  \bm{\phi}_{\rm ps} ,
\end{align}
which enables one to implement the target $U$ such that
\begin{align}
    {\bf G}(G_{D}^{(M,k)}, \bm{\phi}_{U}) = \begin{pmatrix}
    U_{R} & -U_{I} \\ U_{I} & U_{R}
\end{pmatrix} ,
\end{align}
for the real part $U_{R}$ and imaginary part $U_{I}$ of the unitary matrix $U$.
Therefore, the cluster state $\ket{G_{D}^{(M,k)}}$, corresponding to $k$-depth graph $G_{D}^{(M,k)}$ with $k \geq \frac{M}{2} + 1$, serves as a universal resource for implementing LO circuits, which we used throughout our main results.

\begin{theorem}\label{lemma: condition of graph for universality}
    Cluster state $\ket{G_{D}^{(M,k)}}$ corresponding to the $k$-depth graph $G_{D}^{(M,k)}$ with $k \geq \frac{M}{2} + 1$ serves as a resource for universal LO circuit implementation. 
    In other words, given an arbitrary $M\times M$ unitary circuit matrix $U$, one can efficiently construct a set of phase-shift angles $\bm{\phi}_U$ such that ${\bf G}(G_{D}^{(M,k)}, \bm{\phi}_U) = \begin{pmatrix}
    U_{R} & -U_{I} \\ U_{I} & U_{R}
\end{pmatrix}$ for real part $U_{R}$ and imaginary part $U_{I}$ of the unitary matrix $U$.
\end{theorem}

Figure~\ref{fig: brickwork graph} further illustrates this MBLO procedure, showing that the state $\ket{G_{D}^{(M,k)}}$ with phase-shift angles $ \bm{\phi}_{U} $ in Eq.~\eqref{eq: universal graph and angles} can implement the brickwork decomposition of the LO circuit depicted in Fig.~\ref{fig: LOdecomposition}(a).

\begin{figure}[h]
\centering
\includegraphics[width=0.95\linewidth]{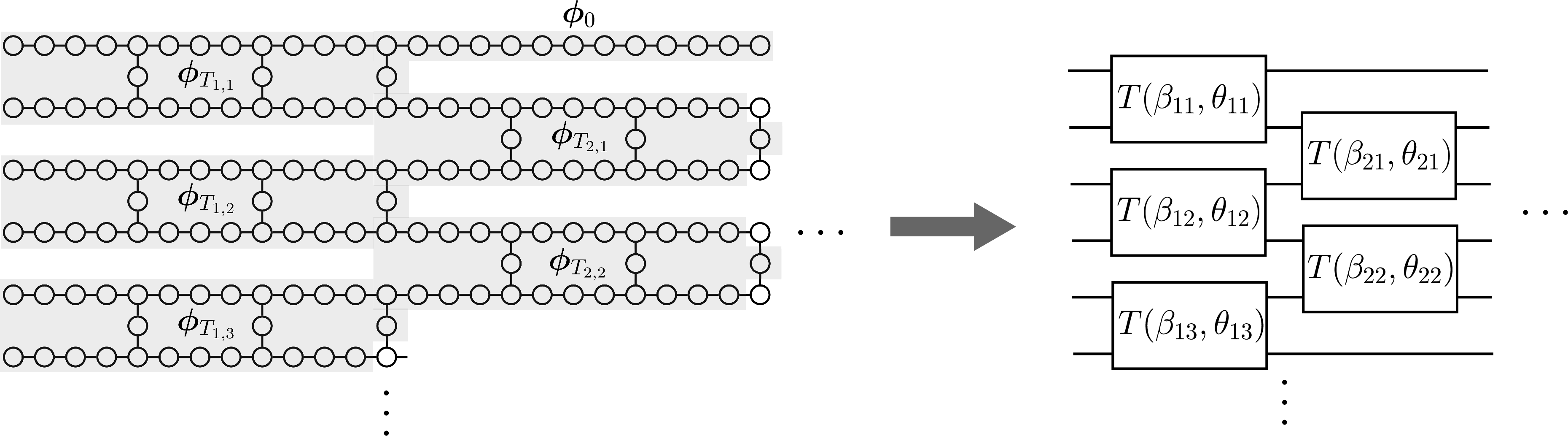}
\caption{Schematic of the cluster state $\ket{G_{D}^{(M,k)}}$ for the graph $G_{D}^{(M,k)}$ defined in Eq.~\eqref{eq: def of G_D^Mk}.
As discussed, with phase-shift angles $ \bm{\phi}_{U} $ given in Eq.~\eqref{eq: universal graph and angles}, this state realizes the brickwork decomposition of the LO circuit shown in Fig.~\ref{fig: LOdecomposition}(a).
}
\label{fig: brickwork graph}
\end{figure}

For convenience in later analysis, we define the \textit{universal} brickwork graph and the phase-shift angles.

\begin{definition}[Universal brickwork graph and universal phase-shift angles]\label{def: brickwork graph and angles}
    We define the universal brickwork graph $G_U$ as $G_{D}^{(M,k)}$ with $k \geq \frac{M}{2} + 1$.
    Also, given a unitary circuit matrix $U$ we define $\bm{\phi}_{U}$ as a set of universal phase-shift angles given in Eq.~\eqref{eq: universal graph and angles} for implementing $U$ by cluster state $\ket{G_U}$, such that ${\bf G}(G_{U}, \bm{\phi}_U) = \begin{pmatrix}
    U_{R} & -U_{I} \\ U_{I} & U_{R}
\end{pmatrix}$ for real part $U_{R}$ and imaginary part $U_{I}$ of the unitary matrix $U$.
\end{definition}

\section{Input state coupling via Bell measurement}\label{appendix: section: coupling}

Our preceding analysis implicitly assumed that the MBLO input states were the intrinsic ancilla modes (i.e., a product of squeezed-vacuum states) used to construct the CV cluster state.
In this section, we review the standard Bell-measurement construction in~\cite{larsen2020architecture, ukai2010universal, alexander2014noise} and explain how an arbitrary input state can be incorporated into the input mode of a cluster state through a Bell measurement that couples the state to the cluster-state input mode.

\begin{figure}[h]
\centering
\includegraphics[width=0.75\linewidth]{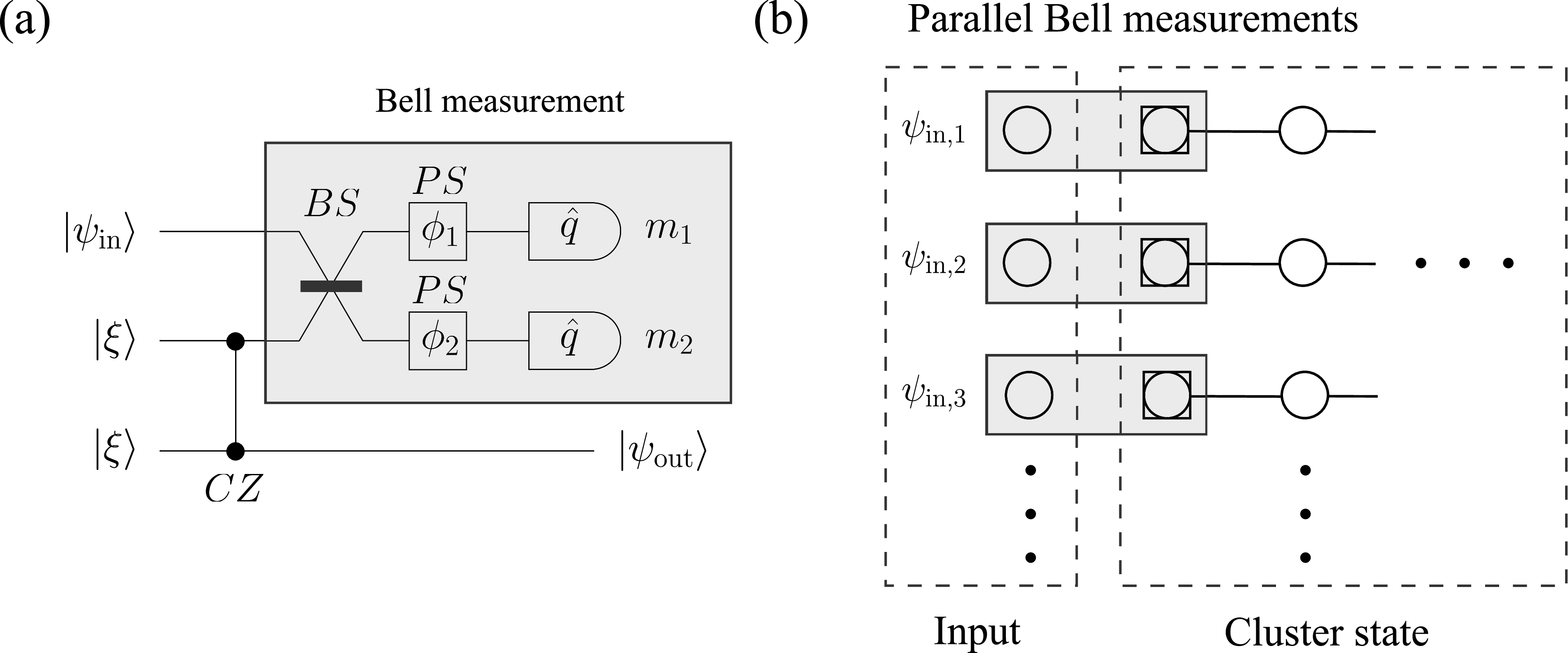}
\caption{Schematics for generalized teleportation scheme via Bell measurements, which teleports an arbitrary quantum state to an input port of a cluster state.}
\label{fig: bell measurement}
\end{figure}

We first schematize how to teleport an arbitrary single-mode input state to a cluster state in Fig.~\ref{fig: bell measurement} (a). 
The scheme consists of an input state, an entangled ancillary state (i.e., cluster state), and an entangled measurement (i.e., Bell measurement) performed jointly on the input state and one mode of the entangled ancillary state (see~\cite{larsen2020architecture, ukai2010universal, alexander2014noise} for more details). 
As illustrated in Fig.~\ref{fig: bell measurement} (a), the Bell measurement is implemented by first applying a balanced beam splitter on the input state and one mode of the entangled ancillary state, followed by a phase shift and homodyne measurement on $\hat{q}$ basis for each mode. 
Then, as described in~\cite{larsen2020architecture}, after measuring $m_1$ and $m_2$ through the Bell measurement, the output quadrature operators are related to the input quadrature operators by
\begin{align}
    \begin{pmatrix}
        \hat{q}_{\rm out} \\ \hat{p}_{\rm out} 
    \end{pmatrix}= \frac{1}{\sin \phi_{-}}\begin{pmatrix}
        \cos\phi_{+} + \cos\phi_{-} & \sin\phi_{+} \\
        -\sin\phi_{+} & \cos\phi_{+} - \cos\phi_{-}
    \end{pmatrix}\begin{pmatrix}
        \hat{q}_{\rm in} \\ \hat{p}_{\rm in} 
    \end{pmatrix} + \begin{pmatrix}
        -1 & 0 \\ 
        0 & 1
    \end{pmatrix}\begin{pmatrix}
        \hat{p}_2 \\ \hat{p}_3 
    \end{pmatrix} + \frac{\sqrt{2}}{\sin \phi_{-}}\begin{pmatrix}
       -\cos\phi_{2} & -\cos\phi_{1} \\
        \sin\phi_{2} & \sin\phi_{1} 
    \end{pmatrix}\begin{pmatrix}
        m_{1} \\ m_{2} 
    \end{pmatrix} ,
\end{align}
where $\phi_{\pm} = \phi_{1} \pm \phi_{2}$.
Since we aim to teleport the input state without additional single-mode operations, we can set the phase-shift angles to $\phi_{1} = -\phi_{2} = \frac{\pi}{4}$, which yields the identity operation. 
Hence, an arbitrary single-mode input state can be successfully teleported to the cluster state using this process, while introducing additional noise and displacement.

Based on this understanding, one can teleport an arbitrary $M$-mode input state to a cluster state by parallel Bell measurements, as depicted in Fig.~\ref{fig: bell measurement} (b).
By performing a Bell measurement on each mode, with phase-shift angles given as above for each Bell measurement, one obtains the following input-output relation:
\begin{align}\label{eq: input output for M bell measurements}
    \begin{pmatrix}
        \hat{\bm{q}}_{\rm out} \\ \hat{\bm{p}}_{\rm out}
    \end{pmatrix} =  \begin{pmatrix}
        \hat{\bm{q}}_{\rm in} \\ \hat{\bm{p}}_{\rm in}
    \end{pmatrix} + (-\mathbb{I}_{M}\oplus \mathbb{I}_{M})  \hat{\bm{p}} + \begin{pmatrix}
        -\mathbb{I}_{M} & -\mathbb{I}_{M} \\ -\mathbb{I}_{M} & \mathbb{I}_{M}
    \end{pmatrix}  \bm{m}.
\end{align}
Therefore, by the parallel Bell measurements illustrated in Fig.~\ref{fig: bell measurement} (b), an arbitrary $M$-mode input state can be successfully teleported to the cluster state, with additional noise and displacement denoted in the right-hand side of Eq.~\eqref{eq: input output for M bell measurements}.

It should be emphasized that, in order to teleport an input state to the target modes (in our case, input modes of the cluster state $\ket{G_{U}}$ for universal circuit implementation), the Bell measurement requires an ancillary state that is entangled with those target modes.
Specifically, as depicted in Fig.~\ref{fig: bell measurement} (a), teleporting a prepared input state at the first mode to the third mode requires an ancillary state on the second mode. 
This implies that, even with the preparation of both the input state and the cluster state $\ket{G_{U}}$, it is still insufficient for implementing a universal LO circuit on the prepared input state via Bell measurements as in Fig.~\ref{fig: bell measurement} (b).
Rather, one should prepare the \textit{appended} cluster state $\ket{G_0}$ corresponding to a concatenated graph $G_0 = (G_{H}^{(2)} \oplus \cdots \oplus G_{H}^{(2)}) \circ G_{U}$ to perform Bell measurements as in Fig.~\ref{fig: bell measurement} (b), which finally allows teleportation of an arbitrary input state to the input modes of $\ket{G_{U}}$.


\section{Our settings for measurement-based linear optics}\label{section: our settings supple}

We formally define the main computational problem regarding MBLO, which is the primary focus of our complexity analysis.
To describe our MBLO procedure, given an $M$-mode unitary circuit matrix $U$ and an $M$-mode input state $\rho_{\rm in}$, we first prepare the CV cluster state $\ket{G_0}$ corresponding to the $G_0 = (G_{H}^{(2)} \oplus \cdots \oplus G_{H}^{(2)})\circ G_{U}$, which is a concatenation of $(G_{H}^{(2)} \oplus \cdots \oplus G_{H}^{(2)})$ (required for Bell measurements) to the universal brickwork graph $G_{U}$ defined in Definition~\ref{def: brickwork graph and angles}.
Next, we apply Bell measurements between $\rho_{\rm in}$ and $\ket{G_0}$ as depicted in Sec.~\ref{appendix: section: coupling}, thus obtaining the cluster state $\ket{G_{U}}$ (with additional noise and displacement given in Eq.~\eqref{eq: input output for M bell measurements}), whose input is now updated to $\rho_{\rm in}$.
We then perform homodyne measurements on the obtained $\ket{G_U}$ with their angles given by the universal phase-shift angles $\bm{\phi}_U$ in Definition~\ref{def: brickwork graph and angles}, on all modes except the designated output modes.

By doing so, the output quadrature operators $(\hat{\bm{q}}_{\rm out}, \hat{\bm{p}}_{\rm out})^{T}$ satisfy the following input-output relation.
\begin{align}\label{eq: inputoutputrelation}
    \begin{pmatrix}
        \hat{\bm{q}}_{\rm out} \\ \hat{\bm{p}}_{\rm out}
    \end{pmatrix} = \textbf{G}_{U} \begin{pmatrix}
        \hat{\bm{q}}_{\rm in} \\ \hat{\bm{p}}_{\rm in}
    \end{pmatrix} + \textbf{N}_{U}   \hat{\bm{p}} + \textbf{D}_{U}  \bm{m}   ,
\end{align}
where $(\hat{\bm{q}}_{\rm in}, \hat{\bm{p}}_{\rm in})^{T}$ denotes the quadrature operators of $\rho_{\rm in}$, $\hat{\bm{p}}$ denotes the momentum quadrature operators of $\Omega(M^2)$ number of the squeezed vacuum ancilla states of the cluster state, and $\bm{m}$ denotes the $\Omega(M^2)$ number of measured outcomes on all the modes except the output modes. 
Also, $\textbf{G}_{U}$ is a target symplectic transformation, and $\textbf{N}_{U}$ and $\textbf{D}_{U}$ are induced noise and displacement matrices in our setup. 
Specifically, they are given by
\begin{align}\label{eq: matrices for our setup}
    &\textbf{G}_{U} = \begin{pmatrix}
    U_{R} & -U_{I} \\ U_{I} & U_{R}
    \end{pmatrix}, 
    \quad 
    \textbf{N}_{U} = \left[ \textbf{G}_{U} \cdot  (-\mathbb{I}_{M}\oplus \mathbb{I}_{M}) \;|\;{\bf N}(G_{U}, \bm{\phi}_U)\right] ,
    \quad \textbf{D}_{U} = \left[ \textbf{G}_{U} \cdot \begin{pmatrix}
        -\mathbb{I}_{M} & -\mathbb{I}_{M} \\ -\mathbb{I}_{M} & \mathbb{I}_{M}
    \end{pmatrix}  \; \bigg| \; {\bf D}(G_{U}, \bm{\phi}_U)\right] ,
\end{align}
where ${\bf N}(G_{U}, \bm{\phi}_U)$ and ${\bf D}(G_{U}, \bm{\phi}_U)$ denote the noise and displacement matrix induced when implementing $U$ via the cluster state $\ket{G_{U}}$ as defined in Definition~\ref{def: matrices for input output relation}, and $[ A \;|\; B]$ denotes the concatenation of the matrices $A$ and $B$ along column axis.

Here, we apply feed-forward displacement $-\textbf{D}_U \bm{m}$ at the output state to undo the unwanted displacement.
Note that we can compute in advance the displacement matrix $\textbf{D}_U$ before obtaining the homodyne measurement outcome $\bm{m}$.
Specifically, when $U$ is specified such that $\bm{\phi}_U$ is also specified, we can compute the matrix ${\bf D}(G_{U}, \bm{\phi}_U)$ by matrix transformation rule in Eq.~\eqref{eq: displacement matrix after graph concatenation} and Eq.~\eqref{eq: displacement matrix after graph sum}.
Then, after obtaining the homodyne outcome $\bm{m}$, we compute $\textbf{D}_U \bm{m}$ with a pre-computed $\textbf{D}_U$, and operate feed-forward displacement to the output state by $-\textbf{D}_U \bm{m}$.

Let $\rho_{\rm out}$ be the final output state after the whole MBLO process. 
To further specify the problem, we consider a Gaussian input state $\rho_{\rm in}$; then the final output state $\rho_{\rm out}$ is also Gaussian because the entire process conserves Gaussianity. 
Hence, input and output states can be fully described by their mean vectors and covariance matrices.
By the input-output relation in Eq.~\eqref{eq: inputoutputrelation}, after feed-forward displacement by $-\textbf{D}_U \bm{m}$, they are related by
\begin{align}
    \bm{\mu} &= \textbf{G}_{U}\bm{\mu}_{\rm in}   ,  \label{eq: mean of output state} \\
    V &= \textbf{G}_{U}V_{\rm in}\textbf{G}_{U}^T + \frac{e^{-2r}}{2}\textbf{N}_{U}\textbf{N}_{U}^T    , \label{eq: covariance of output state}
\end{align}
where we denote $(\bm{\mu}_{\rm in}, V_{\rm in})$ and $(\bm{\mu}, V)$ as the mean vector and the covariance matrix of $\rho_{\rm in}$ and $\rho_{\rm out}$, respectively.

As stated in the main text, for the output state $\rho_{\rm out}$ with its mean vector and covariance matrix specified above, the computational task of interest is the \textit{weak-simulation} of $\rho_{\rm out}$ in the local boson-number basis, formalized as follows.

\begin{definition}[MBLO-based sampling]\label{def: MBLO-based sampling supple}
We define the MBLO-based sampling task as follows. 
On input an $M \times M$ unitary matrix $U$ and a polynomial-size description of an $M$-mode input Gaussian state $\rho_{\rm in}$ with its mean vector and covariance matrix $(\bm{\mu}_{\rm in}, V_{\rm in})$, the task is to output a sample drawn from the probability distribution
\begin{align}\label{eq: probability distribution for photon number outcome}
    p( \bm{n}) = \Tr\left[\ket{\bm{n}}\bra{\bm{n}} \rho_{\rm out} \right]  ,
\end{align} 
where $\rho_{\rm out}$ is the Gaussian state with its mean vector and covariance matrix given by Eq.~\eqref{eq: mean of output state} and Eq.~\eqref{eq: covariance of output state}, and $\ket{\bm{n}}$ denotes the local boson-number basis represented by $\bm{n} \in \mathbb{N}^{M}$.
\end{definition}

Here, a polynomial-size description means that the mean vector and the covariance matrix of $\rho_{\rm in}$ admit a description using polynomial-size bits and can be computed in polynomial time, which holds unless $\rho_{\rm in}$ has super-exponential energy.

\section{Easiness regime}\label{appendix: section: easiness}

We here restate our easiness result (Theorem~1 in the main text) and provide a detailed proof.

\begin{theorem}[Easiness regime]\label{theorem: easiness}
     Given an arbitrary $U$ and $\rho_{\rm in}$, there exists a squeezing level threshold $r_{\rm th} = O\left(\log M\right)$ such that for any $r \leq r_{\rm th}$, the MBLO-based sampling task in Definition~\ref{def: MBLO-based sampling supple} can be simulated in classical polynomial time.
\end{theorem}

\begin{proof}[Proof of Theorem~\ref{theorem: easiness}]

Using the phase space representation given in~\cite{rahimi2016sufficient}, the output probability can be expressed as
\begin{align}
    p( \bm{n} ) = \Tr\left[\ket{\bm{n}}\bra{\bm{n}} \rho_{\rm out} \right] 
    = \pi^{M} \int d\bm{\alpha} Q_{\bm{n}}(\bm{\alpha})P_{\rm out}(\bm{\alpha}), 
\end{align}
where $\bm{\alpha}$ denotes a $2M$-dimensional vector representing a phase space point, $Q_{\bm{n}}(\bm{\alpha})$ denotes the Husimi Q representation of the Fock state $\ket{\bm{n}}\bra{\bm{n}}$~\cite{husimi1940some}, and $P_{\rm out}(\bm{\alpha})$ denotes a Glauber-Sudarshan P representation of the output state $\rho_{\rm out}$~\cite{glauber1963photon, sudarshan1963equivalence}.

By the classical simulability criteria presented in~\cite{rahimi2016sufficient}, efficient sampling from $p(\bm{n})$ is possible if both $Q_{\bm{n}}$ and $P_{\rm out}(\bm{\alpha})$ are non-negative and efficiently sampleable probability distributions. 
Since $Q_{\bm{n}}$ follows a Poisson distribution and therefore always non-negative and efficiently sampleable, the remaining problem is to determine whether $P_{\rm out}(\bm{\alpha})$ is likewise non-negative and efficiently sampleable.
Moreover, because $\rho_{\rm out}$ is a Gaussian state, we can represent $P_{\rm out}(\bm{\alpha})$ as 
\begin{align}
    P_{\rm out}(\bm{\alpha}) = \frac{1}{\pi^{M}|V-\frac{1}{2}\mathbb{I}_{2M}|}\exp[-({\bm{\alpha}-\bm{\mu}})^{T}\left(V-\frac{1}{2}\mathbb{I}_{2M}\right)^{-1}({\bm{\alpha}-\bm{\mu}})  ],
\end{align}
where $|\cdot|$ denotes the determinant. 
Therefore, $P_{\rm out}(\bm{\alpha})$ becomes non-negative and efficiently computable whenever $V - \frac{1}{2}\mathbb{I}_{2M} > 0$. 
Note that, a sufficient condition for the positive definiteness of $V - \frac{1}{2}\mathbb{I}_{2M}$ is
\begin{align}\label{eq: psd}
    \textbf{N}_{U}\textbf{N}_{U}^T - e^{2r}\mathbb{I}_{2M} \geq 0 ,
\end{align}
since this implies
\begin{align}
    V - \frac{1}{2}\mathbb{I}_{2M} = \textbf{G}_{U}V_{\rm in}\textbf{G}_{U}^T + \frac{e^{-2r}}{2}\textbf{N}_{U}\textbf{N}_{U}^T - \frac{1}{2}\mathbb{I}_{2M} > 0 ,
\end{align}
where we used the fact that $V_{\rm in}$ is positive definite for any input state $\rho_{\rm in}$.


To show the positive semi-definiteness in Eq.~\eqref{eq: psd}, we investigate the minimum eigenvalue of $\textbf{N}_{U}\textbf{N}_{U}^T$.
Note that as depicted in Eq.~\eqref{eq: matrices for our setup}, $\textbf{N}_{U}\textbf{N}_{U}^T$ can be expressed as 
\begin{align}
    \textbf{N}_{U}\textbf{N}_{U}^T = \mathbb{I}_{2M} + {\bf N}(G_{U}, \bm{\phi}_U) {\bf N}(G_{U}, \bm{\phi}_U)^{T}, 
\end{align}
where $G_{U}$ and $\bm{\phi}_U$ are universal brickwork graph and universal phase shfit angles for implementing $U$ defined in Definition~\ref{def: brickwork cluster state and angle}, and ${\bf N}(G_{U}, \bm{\phi}_U)$ is a noise matrix induced when implementing $U$ via the cluster state $\ket{G_{U}}$ defined in Definition~\ref{def: matrices for input output relation}.
Hence, by Weyl's inequality~\cite{weyl1912asymptotische}, the minimum eigenvalue of $\textbf{N}_{U}\textbf{N}_{U}^T$ is at least as large as the minimum eigenvalue of ${\bf N}(G_{U}, \bm{\phi}_U) {\bf N}(G_{U}, \bm{\phi}_U)^{T}$ plus $1$.

To proceed, we now introduce a crucial lemma, which gives a lower bound of the minimum eigenvalue of the symmetric product ${\bf N}(G_{U}, \bm{\phi}_U) {\bf N}(G_{U}, \bm{\phi}_U)^{T}$.

\begin{lemma}\label{lemma: minimum eigenvalue for universal brickwork graph state}
    Let ${\bf N}(G_{U}, \bm{\phi}_U)$ be a noise matrix for a universal brickwork graph $G_{U}$ and the corresponding set of universal phase-shift angles $\bm{\phi}_U$ defined in Definition~\ref{def: brickwork graph and angles}. 
    Then, the minimum eigenvalue $\lambda_{min}$ of ${\bf N}(G_{U}, \bm{\phi}_U) {\bf N}(G_{U}, \bm{\phi}_U)^{T}$ satisfies 
    \begin{align}
        \lambda_{min} \geq (3-2\sqrt{2})M  + 1  .
    \end{align}
\end{lemma}
\begin{proof}
    See Sec.~\ref{appendix: section: lemma for minimum eigenvalue}.
\end{proof}
Therefore, by Lemma~\ref{lemma: minimum eigenvalue for universal brickwork graph state}, we finally obtain the squeezing level bound $r_{\rm th}$ for classical simulability:
\begin{align}
    e^{2r} \leq (3-2\sqrt{2})M + 2   \quad \Rightarrow \quad r \leq r_{\rm th} =  \frac{1}{2}\log((3-2\sqrt{2})M + 2)  = O ( \log M ),
\end{align}
concluding the proof.

\end{proof}

\section{Hardness regime I}

This section restates our first hardness result (Theorem~2 in the main text) and provides its detailed proof.

To this end, we first introduce a general setting of Gaussian boson sampling (GBS) and formally define the (approximate) classical simulation of GBS, i.e., sampling within bounded total variation distance (TVD), following the standard convention in the literature~\cite{hamilton2017gaussian, kruse2019detailed, deshpande2022quantum, go2025sufficient}.

\begin{definition}[Approximate GBS]\label{def: GBS}
    Fix an error parameter $\kappa < 1$. 
    Consider an input state $\sigma_{\rm in} = \ket{r_0}\bra{r_0}^{K}\otimes\ket{0}\bra{0}^{M-K}$ for a single-mode squeezed vacuum state $\ket{r_0}$ with bounded squeezing level $r_0 = O(1)$ and $K \leq M$.
    For an $M$-mode LO circuit $U$, let $\sigma_U$ denote the output state obtained by evolving $\sigma_{\rm in}$ through $U$.
    The approximate GBS refers to the task that, on input an LO circuit $U$, samples an outcome $\bm{n} \in \mathbb{N}^{M}$ according to any probability distribution $\bar{q}(\bm{n})$ satisfying
    \begin{align}
        D_{\rm TV}(\bar{q}, q) \le \kappa  ,
    \end{align}
    where $D_{\rm TV}$ denotes TVD and 
    \begin{align}
        q(\bm{n}) = \Tr\!\left[\ket{\bm{n}}\!\bra{\bm{n}}\,\sigma_U\right]
    \end{align}
    denotes the ideal output distribution of GBS.
\end{definition}

Evidenced by the hardness results for GBS, we here conjecture that this approximate GBS in Definition~\ref{def: GBS} is classically intractable for suitably chosen $K$ and $\kappa$~\cite{hamilton2017gaussian, kruse2019detailed, deshpande2022quantum, go2025sufficient}.

\begin{conjecture}[Classical hardness of GBS]\label{conj: hardness of GBS}
    There exists $K \leq M$ and an error parameter $\kappa = O({\rm poly}(M)^{-1})$ such that no polynomial-time classical algorithm can simulate the approximate GBS in Definition~\ref{def: GBS}.
\end{conjecture}

Based on these arguments, we restate our first hardness result.

\begin{theorem}[Hardness regime I]\label{theorem: weak hardness}
    Suppose that Conjecture~\ref{conj: hardness of GBS} holds.
    Then, there exists a threshold $r_{\rm th} = \Omega\left(\log M\right)$ such that for any squeezing level $r \geq r_{\rm th}$, no polynomial-time classical algorithm can simulate the MBLO-based sampling in Definition~\ref{def: MBLO-based sampling supple}.
\end{theorem}

To prove Theorem~\ref{theorem: weak hardness}, let $\rho_{\rm id}$ be the output state $\rho_{\rm out}$ of the MBLO scheme in Definition~\ref{def: MBLO-based sampling supple} in the infinite-squeezing limit $r \rightarrow \infty$. 
Set the input state of MBLO be $\rho_{\rm in} = \ket{r_0}\bra{r_0}^{K}\otimes\ket{0}\bra{0}^{M-K}$ for a single-mode squeezed vacuum state $\ket{r_0}$ with bounded squeezing level $r_0 = O(1)$, such that $\rho_{\rm id}$ coincides with the GBS output state $\sigma_{\rm in}$ in Definition~\ref{def: GBS}.
Namely, $\rho_{\rm id}$ has zero mean and covariance matrix $V_{\rm id} = \textbf{G}_{U}V_{\rm in}\textbf{G}_{U}^T$, with ${\bf G}$ being the symplectic (orthogonal) matrix corresponding to $U$ and $V_{\rm in}$ being the covariance matrix of $\rho_{\rm id}$.


We now introduce a crucial lemma, which quantifies the required squeezing level to have a bounded TVD between $\rho_{\rm id}$ and $\rho_{\rm out}$.

\begin{lemma}\label{lemma: tvd by finite squeezing}
    Let $\rho_{\rm id}$ be a Gaussian state with covariance matrix $V_{\rm id} = {\bf G}_{U}V_{\rm in} {\bf G}_{U}^{T} $ for ${\bf G}$ being the symplectic matrix corresponding to $U$ and $V_{\rm in}$ being the covariance matrix of the product squeezed state $\rho_{\rm in}$.
    Let $\rho_{\rm out}$ be the corresponding state obtained by the MBLO procedure described in Sec.~\ref{section: our settings supple}.
    Then, for any $0 < \beta < 1$, when
    \begin{align}\label{eq: lemma: tvd by finite squeezing}
        r \ge \frac{1}{2} \log( \frac{ \left\Vert  {\bf N}_{U}{\bf N}_{U}^T \right\Vert_{F}\| V_{\rm in}^{-1} \|_{F} }{ 8 \beta^2 })   ,
    \end{align}
    the TVD between $\rho_{\rm id}$ and $\rho_{\rm out}$ is upper bounded by $\beta$.
\end{lemma}

We defer a proof of Lemma~\ref{lemma: tvd by finite squeezing} to the end of this section. 
Lemma~\ref{lemma: tvd by finite squeezing} implies that when the squeezing level satisfies Eq.~\eqref{eq: lemma: tvd by finite squeezing}, the distribution of MBLO-based sampling in Definition~\ref{def: MBLO-based sampling supple} deviates from the distribution of GBS by at most $\beta$ in TVD.

Note that the right-hand side of Eq.~\eqref{eq: lemma: tvd by finite squeezing} involves the norm of ${\bf N}_{U}{\bf N}_{U}^T$. 
To obtain a lower bound on $r$, we bound this term from above, as stated follows.

\begin{lemma}\label{lemma: norm of noise matrix}
    For an arbitrary $U$, the Frobenius norm of the noise matrix product ${\bf N}_{U}{\bf N}_{U}^T$ satisfies 
    \begin{align}
    \left\Vert  {\bf N}_{U}{\bf N}_{U}^T \right\Vert_{F}  =   O(Mk)  ,
    \end{align}
    where $k \geq M/2 + 1$ as given in Definition~\ref{def: brickwork graph and angles}.
\end{lemma}
\begin{proof}
    See Sec.~\ref{appendix: section: norm of noise matrix}.
\end{proof}

Now suppose Conjecture~\ref{conj: hardness of GBS} holds; that is, there exists $\kappa = \text{poly}(M)^{-1}$ such that approximate GBS within TVD error $\kappa$ is hard to classically simulate.
Here, by our definition of $\rho_{\rm in}$, $\| V_{\rm in}^{-1} \|_{F}$ is bounded to at most $O(\sqrt{M})$.
Combining these facts with Lemma~\ref{lemma: norm of noise matrix}, it follows that there exists a threshold $r_{\rm th} = \Omega\left(\log M\right)$ such that for all $r \geq r_{\rm th}$, the MBLO-based sampling is classically intractable, thus concluding the proof of Theorem~\ref{theorem: weak hardness}.

\begin{proof}[Proof of Lemma~\ref{lemma: tvd by finite squeezing}]

By~\cite{marian2012uhlmann, spedalieri2012limit}, the quantum fidelity between $\rho_{\rm out}$ and $\rho_{\rm id}$, both having zero means with $\rho_{\rm id}$ being pure, can be written as
\begin{align}\label{eq: fidelity between fin and id}
    F_{\rho_{\rm out}, \rho_{\rm id}} = \frac{1}{\sqrt{|V + V_{\rm id}|}}.
\end{align}
Note that $\rho_{\rm id}$ has zero mean and covariance matrix $V_{\rm id} = \textbf{G}_{U}V_{\rm in}\textbf{G}_{U}^T$, where the covariance matrix of the input squeezed vacuum states $V_{\rm in}$ can be further decomposed as $V_{\rm in} = S(\mathbb{I}_{2M}/2)S^{T}$.
Here, $\mathbb{I}_{2M}/2$ is the covariance matrix of the vacuum state and $S$ is a (diagonal) symplectic matrix representing the squeezing operation used to prepare the input state.
Hence, the determinant in the denominator above can be bounded by 
\begin{align}
    |V + V_{\rm id}| 
    &= \left|2V_{\rm id} + \frac{e^{-2r}}{2} {\bf N}_{U}{\bf N}_{U}^T  \right|
    = \left| \mathbb{I}_{2M} + \frac{e^{-2r}}{2}S^{-1} \textbf{G}_{U}^{T} {\bf N}_{U}{\bf N}_{U}^T \textbf{G}_{U} S^{-T} \right| \\
    &\le \left(1 + \frac{1}{2M}\Tr[\frac{e^{-2r}}{2}S^{-1} \textbf{G}_{U}^{T} {\bf N}_{U}{\bf N}_{U}^T \textbf{G}_{U} S^{-T}] \right)^{2M} \label{eq: AM-GM}\\
    &\le \left(1 + \frac{e^{-2r}}{4M}\sqrt{\Tr[\left( {\bf N}_{U}{\bf N}_{U}^T \right)^2]}\sqrt{\Tr[\left( \textbf{G}_{U} S^{-T}S^{-1} \textbf{G}_{U}^{T} \right)^2]} \right)^{2M} \label{eq: cauchy schwarz}\\
    &= \left(1 + \frac{e^{-2r}}{8M}\|  {\bf N}_{U}{\bf N}_{U}^T \|_{F} \| V_{\rm in}^{-1} \|_{F} \right)^{2M}    ,
\end{align}
where we used the AM–GM inequality and the positive definiteness of the matrix $S^{-1} \textbf{G}_{U}^{T} \textbf{N}\textbf{N}^{T} \textbf{G}_{U} S^{-T}$ in Eq.~\eqref{eq: AM-GM}, and used the Cauchy-Schwarz inequality in Eq.~\eqref{eq: cauchy schwarz}.  
To simplify the form, we here assume that $r$ is large enough to satisfy
\begin{align}
    \frac{e^{-2r}}{8} \|{\bf N}_{U}{\bf N}_{U}^T \|_{F}  \| V_{\rm in}^{-1} \|_{F} < 1   .
\end{align}
Then, as $(1+x/M)^{M} \le (1-x)^{-1}$ for $0\le x < 1$, we have 
\begin{align}
    \left| V + V_{\rm id} \right| \le \left(1 - \frac{e^{-2r}}{8}\|  {\bf N}_{U}{\bf N}_{U}^T \|_{F} \| V_{\rm in}^{-1} \|_{F} \right)^{-2}   ,
\end{align}
and by substituting the above inequality in Eq.~\eqref{eq: fidelity between fin and id}, we obtain the bound
\begin{align}
    F_{\rho_{\rm out}, \rho_{\rm id}} \ge 1 - \frac{e^{-2r}}{8}\|  {\bf N}_{U}{\bf N}_{U}^T \|_{F} \| V_{\rm in}^{-1} \|_{F} \quad\Rightarrow\quad 1 - F_{\rho_{\rm out}, \rho_{\rm id}} \leq \frac{e^{-2r}}{8}\| {\bf N}_{U}{\bf N}_{U}^T \|_{F} \| V_{\rm in}^{-1} \|_{F} .
\end{align}
Let ${\bf D}_{\rho_{\rm out}, \rho_{\rm id}} $ denote the TVD between the boson-number distributions of $\rho_{\rm out}$ and $\rho_{\rm id}$. 
Then, ${\bf D}_{\rho_{\rm out}, \rho_{\rm id}} $ can be upper bounded in terms of the quantum infidelity between $\rho_{\rm out}$ and $\rho_{\rm id}$~\cite{fuchs1999cryptographic} by
\begin{align}
    {\bf D}_{\rho_{\rm out}, \rho_{\rm id}} \le \sqrt{1 - F_{\rho_{\rm out}, \rho_{\rm id}}} \le \frac{e^{-r}}{2\sqrt{2}}\sqrt{\|{\bf N}_{U}{\bf N}_{U}^T \|_{F} \| V_{\rm in}^{-1} \|_{F} }   .
\end{align}
Therefore, for the error parameter $\beta$ with $0 < \beta < 1$, by setting the ancillary squeezing parameter $r$ satisfying
\begin{align}
    r \ge \frac{1}{2} \log(\frac{\|{\bf N}_{U}{\bf N}_{U}^T \|_{F} \| V_{\rm id}^{-1} \|_{F} }{ 8 \beta^2 })   ,
\end{align}
then the TVD is upper bounded by $\beta$.

\end{proof}

\section{Hardness regime II}

In this section, we restate our second hardness result (Theorem~3 in the main text) and provide a detailed proof. 
More concretely, our second hardness result can be stated as follows. 

\begin{theorem}[Hardness regime II]\label{theorem: strong hardness}
Let $\mathcal{M}$ be an oracle for the MBLO-based sampling defined in Definition~\ref{def: MBLO-based sampling supple}, with an implicitly fixed squeezing level $r$.
For any constant $\Lambda >0$, there exists a squeezing level threshold $r_{\rm th} = \Omega\left(M^{\Lambda}\right)$ such that $\rm{P}^{\#P} \subseteq \rm{BPP}^{NP^{\mathcal{M}}}$ holds for any $r \geq r_{\rm th}$.
\end{theorem}	


Consequently, in the high-squeezing regime specified by Theorem~\ref{theorem: strong hardness}, the existence of an efficient classical algorithm for MBLO-based sampling would imply a collapse of the polynomial hierarchy (PH) to $\rm{BPP}^{NP}$ via Toda's theorem $\rm{PH} \subseteq \rm{P}^{\rm{\#P}}$~\cite{toda1991pp}.
Under the standard assumption that the PH does not collapse, we conclude that MBLO-based sampling is classically hard whenever the squeezing exceeds this threshold.

\subsection{Proof sketch}\label{sec: hardness proof sketch}

We begin by outlining the proof of Theorem~\ref{theorem: strong hardness}.
Here, our high-level strategy is to establish Stockmeyer's reduction~\cite{stockmeyer1985approximation}. 
Specifically, given oracle access to the sampler $\mathcal{M}$ that outputs samples according to $p(\bm{n})$ in Definition~\ref{def: MBLO-based sampling supple}, Stockmeyer's algorithm~\cite{stockmeyer1985approximation} enables a multiplicative estimation of $p(\bm{n})$ to within a desired accuracy in complexity class ${\rm BPP^{NP}}$. 
We then show that such an estimate suffices to solve a \#P-hard problem when the squeezing level exceeds a certain threshold, which yields $\rm{P}^{\#P} \subseteq \rm{BPP}^{NP^{\mathcal{M}}}$ as stated in Theorem~\ref{theorem: strong hardness}.

More concretely, we encode a quantity (specified in the proof) that is \#P-hard to multiplicatively estimate into an ideal output probability, $q(\bm{n})$, of MBLO-based sampling in Definition~\ref{def: MBLO-based sampling supple} in the infinite-squeezing limit $r \rightarrow \infty$.
We then determine a squeezing level threshold above which $p(\bm{n})$ is multiplicatively close to $q(\bm{n})$ to the required accuracy.
Under this condition, a multiplicative estimate of $p(\bm{n})$ obtained via Stockmeyer's algorithm (given access to $\mathcal{M}$) yields a multiplicative estimate of the underlying \#P-hard quantity.

\subsection{A detailed proof of Theorem~\ref{theorem: strong hardness}}

First, consider the \textit{ideal} output probability of GBS~\cite{hamilton2017gaussian, kruse2019detailed, deshpande2022quantum, go2025sufficient}, corresponding to a Gaussian state that has zero mean and covariance matrix $V_{\rm id} = \textbf{G}_{U}V_{\rm in}\textbf{G}_{U}^T$, where $V_{\rm in}$ is the covariance matrix of input state $\rho_{\rm in}$ for $\mathcal{M}$ in Definition~\ref{def: MBLO-based sampling supple} and $G_{U}$ is the symplectic (orthogonal) matrix corresponding to the input unitary matrix $U$ as in Eq.~\eqref{eq: matrices for our setup};
note that this state corresponds to $\rho_{\rm out}$ in Definition~\ref{def: MBLO-based sampling supple} but in the infinite-squeezing limit $r \rightarrow \infty$.

Let $\bm{n} = (n_1,\dots,n_{M})$ be a total $N$ boson outcome over first $N \leq M$ modes, such that $n_{i} = 1$ for $i \in [N]$ and $n_{i} = 0$ for $i \in [M]\setminus[N]$.
The ideal output probability $q(\bm{n})$ for the outcome $\bm{n}$ can then be expressed as~\cite{hamilton2017gaussian, kruse2019detailed, deshpande2022quantum, go2025sufficient}
\begin{align}\label{eq: ideal output probability}
    q(\bm{n}) = \frac{1}{\sqrt{|\Sigma_Q|}}\Haf(A_{\bm{n}\oplus \bm{n}}),
\end{align}
where $\Sigma_Q = \Sigma + \frac{1}{2}\mathbb{I}_{2M}$, $\Sigma = SV_{\rm id}S^{\dag}$ for basis transformation matrix 
\begin{align}
    S = \frac{1}{\sqrt{2}}\begin{pmatrix}
        \mathbb{I}_{M} & i\mathbb{I}_{M} \\ \mathbb{I}_{M} & -i\mathbb{I}_{M}
    \end{pmatrix},
\end{align}
and $A$ matrix is given by
\begin{align}\label{eq: def of A}
    A = \mathbb{X}_{2M}(\mathbb{I}_{2M} - \Sigma_Q^{-1}) \quad\text{where}\quad \mathbb{X}_{2M} = \begin{pmatrix}
        0 & \mathbb{I}_{M} \\ 
        \mathbb{I}_{M} & 0
    \end{pmatrix},
\end{align}
and $A_{\bm{n}\oplus \bm{n}}$ is a $2N\times 2N$ matrix obtained by taking $n_i$ copies of $i$-th and $(i + M)$-th rows and columns of $A$ for all $i \in [M]$.

Similarly, using the above conventions, for the same input state $\rho_{\rm in}$ and the same input unitary matrix $U$, the output probability $p(\bm{n} )$ of $\mathcal{M}$ can be expressed as 
\begin{align}\label{eq: p(n) supple}
    p(\bm{n} ) =  \frac{1}{\sqrt{|\Sigma'_Q|}}\Haf(A'_{\bm{n}\oplus\bm{n}}) ,
\end{align}
where now $\Sigma'_Q = SV_{\rm id}S^{\dag} + \frac{e^{-2r}}{2}S\textbf{N}_{U}\textbf{N}_{U}^TS^{\dag} + \frac{1}{2}\mathbb{I}_{2M}$ and $A' = \mathbb{X}_{2M}(\mathbb{I}_{2M} - {\Sigma'_Q}^{-1})$. 
Here, we can rewrite $\Sigma'_{Q}$ in terms of $\Sigma_{Q}$ as 
\begin{align}\label{eq: definition of noise matrix}
    \Sigma'_{Q} = \Sigma_{Q} + \delta\Sigma \quad \text{where}\quad \delta\Sigma = \frac{e^{-2r}}{2}S\textbf{N}_{U}\textbf{N}_{U}^TS^{\dag}.
\end{align}
One can readily check that, in the infinite squeezing limit $r \rightarrow \infty$, we have $\Sigma_{Q}' \rightarrow \Sigma_{Q}$ and accordingly $p(\bm{n}) \rightarrow q(\bm{n})$.

As noted, our high-level strategy is twofold: (i) encode a quantity (specified in Lemma~\ref{theorem: sharp P hardness of computing permanent} below) that is \#P-hard to estimate multiplicatively into an ideal output probability $q(\bm{n})$, and (ii) determine a squeezing level threshold above which $p(\bm{n})$ is multiplicatively close to $q(\bm{n})$ to the required accuracy.
Taken together, these two ingredients yield the desired statement $\rm{P}^{\#P} \subseteq \rm{BPP}^{NP^{\mathcal{M}}}$ via Stockmeyer's reduction~\cite{stockmeyer1985approximation}.

\subsubsection{Encoding hard-to-estimate quantity to the ideal probability}

We first identify the hard-to-estimate quantity that will be encoded into the ideal output probability $q(\bm{n})$ in Eq.~\eqref{eq: ideal output probability}. 



\begin{lemma}[Corollary from~\cite{aaronson2011linear}]\label{theorem: sharp P hardness of computing permanent}
    For a matrix $W' \in \{-1,0,1\}^{N_0 \times N_0}$, it is $\rm{\#P}$-hard to compute an multiplicative estimate $\bar{W}$ of ${\rm{Per}}(W')^2$ satisfying 
    \begin{align}
    (1-g){\rm{Per}}(W')^2 \leq \bar{W} \leq (1+g){\rm{Per}}(W')^2, 
    \end{align}
    for a sufficiently small multiplicative factor $g \ll 1$.
\end{lemma}

We now specify the inputs $U$ and $\rho_{\rm in}$ to $\mathcal{M}$ such that $q(\bm{n})$ encodes $\Per(W')^2$ by $q(\bm{n}) \propto \Per(W')^2$.
Importantly, as recently proposed in~\cite{bouland2024average}, we embed a small $W' \in \{-1,0,1\}^{N_0 \times N_0}$ into a relatively large unitary matrix $U \in {\rm U}(M)$ (i.e., $M \gg N_{0}$).
Indeed, as it turns out, this approach leads to a reduced requirement on the squeezing level for our conclusion $\rm{P}^{\#P} \subseteq \rm{BPP}^{NP^{\mathcal{M}}}$.


To formalize this encoding, for a given $N_0\times N_0$ worst-case matrix $W' \in \{-1,0,1\}^{N_0 \times N_0}$ and a fixed constant $\Lambda >0$, we define $N$ such that it satisfies $N =  \left\lceil N_0^{1/\lambda} \right\rceil$ for any constant $\lambda < \Lambda$ that makes even $N$.
We then construct an $N\times N$ matrix $W$ given by
\begin{align}\label{eq: definition for W}
    W = \begin{pmatrix}
        0_{\frac{ N}{2}} & W'\oplus I_{\frac{N}{2} - N_0} \\ (W')^{T}\oplus I_{\frac{N}{2} - N_0} & 0_{\frac{ N}{2}}
    \end{pmatrix} .
\end{align}
Observe that, by the structure of $W$, we have 
\begin{align}\label{eq: haf-W = per-W'}
    \Haf(W) = \Per(W'\oplus I_{\frac{N}{2} - N_0}) = \Per(W').
\end{align}
Moreover, as $W$ is a real symmetric matrix, it allows a matrix decomposition of the form $W = YY^T$ for an $N\times N$ real matrix $Y$.
Here, as stated in the following lemma, this $Y$ matrix can be embedded as a submatrix into an $M\times M$ unitary matrix (under proper normalization), provided that $M \geq 2N$.

\begin{lemma}[\cite{aaronson2011computational}]\label{lemma: encoding worst-case matrix in unitary matrix}
    Let $Y \in \mathbb{C}^{N \times N}$. Then, for all $M \geq 2N$ and $\varepsilon \leq \frac{1}{\Vert Y \Vert_2}$ for the spectral norm $\Vert \cdot \Vert_2$ ($p = 2$ norm), there exists an $M\times M$ unitary matrix $U$ that contains $\varepsilon Y$ as a top-leftmost $N\times N$ submatrix, where the corresponding $U$ can be computed in polynomial time given $Y$.
\end{lemma}

Hence, by Lemma~\ref{lemma: encoding worst-case matrix in unitary matrix}, we can construct an $M\times M$ unitary matrix $U$ that contains $\frac{1}{\Vert Y \Vert_2} Y$ as a top-leftmost $N\times N$ submatrix. 
Now, consider an input state $\rho_{\rm in} = \ket{r_0}\bra{r_0}^{N}\otimes\ket{0}\bra{0}^{M-N}$ for a single-mode squeezed vacuum state $\ket{r_0}$ with its squeezing level $r_0$, where $M \geq 2N$.
Then, for the input state $\rho_{\rm in}$ and the circuit unitary matrix $U$, we have~\cite{hamilton2017gaussian, kruse2019detailed, deshpande2022quantum}
\begin{align}
    &\Sigma_{Q} = \begin{pmatrix}
        U(\cosh^2 r_0\mathbb{I}_{N}\oplus \mathbb{I}_{M-N})U^{\dag} & U(\cosh r_0\sinh r_0\mathbb{I}_{N}\oplus 0_{M-N})U^{T} \\
        U^*(\cosh r_0\sinh r_0\mathbb{I}_{N}\oplus 0_{M-N})U^{\dag} &   U^{*}(\cosh^2 r_0\mathbb{I}_{N}\oplus \mathbb{I}_{M-N})U^{T}
    \end{pmatrix} , \label{eq: Sigma Q for hardness 2}   
\end{align}
where $0_{M-N}$ is an $M-N$-dimensional all-zero matrix.
This results in
\begin{align}
    &\Sigma_{Q}^{-1} = \begin{pmatrix}
        \mathbb{I}_{M} & -\tanh r_0 \cdot U(\mathbb{I}_{N}\oplus 0_{M-N})U^{T} \\
        -\tanh r_0 \cdot U^*(\mathbb{I}_{N}\oplus 0_{M-N})U^{\dag} &  \mathbb{I}_{M}
    \end{pmatrix} .  \label{eq: H42}
\end{align}
Then, by Eq.~\eqref{eq: def of A}, we obtain
\begin{align}
      A = \tanh r_0 \begin{pmatrix}
         U^*(\mathbb{I}_{N}\oplus 0_{M-N})U^{\dag}  & 0 \\
        0 & U(\mathbb{I}_{N}\oplus 0_{M-N})U^{T}  
    \end{pmatrix}   .
\end{align}
For the constructed $U$ that contains $\frac{1}{\Vert Y \Vert_2} Y$ as a top-leftmost $N\times N$ submatrix (i.e., by Lemma~\ref{lemma: encoding worst-case matrix in unitary matrix}), the top-leftmost $N\times N$ submatrix of $U(\mathbb{I}_{N}\oplus 0_{M-N})U^{T}$ is given by $\frac{1}{\Vert Y \Vert_2^2}YY^T = \frac{1}{\Vert W \Vert_2} W$ for $W$ in Eq.~\eqref{eq: definition for W}, which is the same for $U^*(\mathbb{I}_{N}\oplus 0_{M-N})U^{\dag}$ since $W$ is a real matrix.
Accordingly, for $\bm{n}$ corresponding to the collision-free outcome with one boson in each of the first $N$ modes among the $M$ output modes, we have
\begin{align}
     A_{\bm{n}\oplus\bm{n}} = \frac{\tanh r_0}{\Vert W \Vert_2}\begin{pmatrix}
        W  & 0 \\
        0 & W 
    \end{pmatrix} , 
\end{align}
which, together with Eq.~\eqref{eq: ideal output probability} and Eq.~\eqref{eq: haf-W = per-W'}, yields the desired relation $q(\bm{n}) \propto \ \Per(W')^2$.

\subsubsection{Bounding the probability ratio}

We next characterize the probability ratio in our settings:
\begin{align}\label{eq: ratio of p(n) and q(n)}
    \frac{p(\bm{n} )}{q(\bm{n})} = \frac{\sqrt{|\Sigma_Q|}}{\sqrt{|\Sigma'_Q|}}\frac{\Haf(A'_{\bm{n}\oplus\bm{n}})}{\Haf(A_{\bm{n}\oplus \bm{n}})}   .
\end{align}
Namely, we aim to identify a squeezing level $r$ that makes 
\begin{align}\label{eq: our goal}
    \left| 1 - \frac{p(\bm{n} )}{q(\bm{n})} \right| = \left\vert 1- \frac{\sqrt{|\Sigma_Q|}}{\sqrt{|\Sigma'_Q|}}\frac{\Haf(A'_{\bm{n}\oplus\bm{n}})}{\Haf(A_{\bm{n}\oplus \bm{n}})}\right\vert \leq \epsilon ,
\end{align}
for a desired relative error $\epsilon$.
This will ensure that, given access to the sampler $\mathcal{M}$, one can obtain a multiplicative estimate of $q(\bm{n})$ via Stockmeyer's reduction.

Our first observation is that the factor $\sqrt{|\Sigma'_Q|}$ is bounded multiplicatively to $\sqrt{|\Sigma_Q|}$ in terms of the Frobenius norm of $\delta\Sigma$.

\begin{lemma}\label{lemma: multiplicative coefficient is bounded}
Suppose that $\left\vert\Tr(\delta\Sigma\cdot\Sigma_{Q}^{-1})\right\vert \leq \frac{1}{3}$. 
Then, the ratio $\frac{\sqrt{|\Sigma_Q|}}{\sqrt{|\Sigma'_Q|}}$ satisfies
\begin{align}\label{eq: lemma: multiplicative coefficient bound}
     \left\vert 1 - \frac{\sqrt{|\Sigma_Q|}}{\sqrt{|\Sigma'_Q|}} \right\vert \leq \| \delta\Sigma\|_{F} \| \Sigma_{Q}^{-1}\|_{F} ,
\end{align}
where $\|\cdot\|_{F}$ denotes the Frobenius norm.
\end{lemma}
\begin{proof}
    See Sec.~\ref{appendix: section: lemma for bounding multiplicative coefficient}.
\end{proof}




Also, rewrite $A'$ in terms of $A$ by
\begin{align}
    A' &= \mathbb{X}_{2M}(\mathbb{I}_{2M} - (\Sigma_Q + \delta\Sigma)^{-1}) \\
    &= \mathbb{X}_{2M}\left(\mathbb{I}_{2M} - \Sigma_{Q}^{-1}\left(\mathbb{I}_{2M} + \delta\Sigma\cdot\Sigma_{Q}^{-1}\right)^{-1}\right) \\
    &= \mathbb{X}_{2M}\left(\mathbb{I}_{2M} - \Sigma_{Q}^{-1}\left( \mathbb{I}_{2M} - \delta\Sigma\cdot\Sigma_{Q}^{-1}\left(\mathbb{I}_{2M} + \delta\Sigma\cdot\Sigma_{Q}^{-1}\right)^{-1}\right) \right) \\
    &= A + \mathbb{X}_{2M}\cdot\Sigma_{Q}^{-1}\cdot\delta\Sigma\cdot\Sigma_{Q}^{-1}\left(\mathbb{I}_{2M} + \delta\Sigma\cdot\Sigma_{Q}^{-1}\right)^{-1}, \label{eq: H12}
\end{align}
provided that $\mathbb{I}_{2M} + \delta\Sigma\cdot\Sigma_{Q}^{-1}$ is invertible (by Observation~\ref{observation: invertible} in Sec.~\ref{appendix: section: lemma for bounding multiplicative coefficient}). 

In what follows, we show that the hafnian ${\Haf(A'_{\bm{n}\oplus\bm{n}})}$ can also be multiplicatively bounded to ${\Haf(A_{\bm{n}\oplus \bm{n}})}$ with respect to the matrix norm $\| \delta\Sigma\|_{F}$, such that the right-hand side of Eq.~\eqref{eq: our goal} can also be bounded in terms of $\| \delta\Sigma\|_{F}$.
Since $ \left\Vert \delta\Sigma\right\Vert_{F} = \frac{e^{-2r}}{2} \Vert \textbf{N}_{U}\textbf{N}_{U}^T \Vert_{F}$ by definition of $\delta\Sigma$ in Eq.~\eqref{eq: definition of noise matrix}, by Lemma~\ref{lemma: norm of noise matrix} that upper bounds $\Vert \textbf{N}_{U}\textbf{N}_{U}^T \Vert_{F}$, we will obtain a required squeezing level $r$ to obtain a desired multiplicative imprecision in Eq.~\eqref{eq: our goal}.




To proceed, recall that by our construction of $U$ and $\rho_{\rm in}$, we have  
\begin{align}
     A_{\bm{n}\oplus\bm{n}} = \frac{\tanh r_0}{\Vert W \Vert_2}\begin{pmatrix}
        W  & 0 \\
        0 & W 
    \end{pmatrix} .
\end{align}
To simplify the form, we newly define the normalized version of $A_{\bm{n}\oplus\bm{n}}$ as $B$ such that 
\begin{align}
    B = \begin{pmatrix}
        W & 0 \\
        0 & W 
    \end{pmatrix} ,
\end{align}
and correspondingly, according to Eq.~\eqref{eq: H12}, the normalized version of $A'_{\bm{n}\oplus\bm{n}}$ can be represented as
\begin{align}
    B' = \begin{pmatrix}
        W & 0 \\
        0 & W
    \end{pmatrix} + \frac{\Vert W \Vert_2}{\tanh r_0}\left(\mathbb{X}_{2M}\cdot\Sigma_{Q}^{-1}\cdot\delta\Sigma\cdot\Sigma_{Q}^{-1}\left(\mathbb{I}_{2M} + \delta\Sigma\cdot\Sigma_{Q}^{-1}\right)^{-1} \right)_{\bm{n}\oplus \bm{n}} .
\end{align}
Now, we newly denote the $2N\times 2N$ noise matrix in the right-hand side as
\begin{align}\label{eq: definition for C 2}
    C \coloneqq \frac{\Vert W \Vert_2}{\tanh r_0}\left(\mathbb{X}_{2M}\cdot\Sigma_{Q}^{-1}\cdot\delta\Sigma\cdot\Sigma_{Q}^{-1}\left(\mathbb{I}_{2M} + \delta\Sigma\cdot\Sigma_{Q}^{-1}\right)^{-1} \right)_{\bm{n}\oplus \bm{n}},
\end{align} 
such that $B' = B + C$.
Then we have
\begin{align}\label{eq: ration of haf A' and haf A 2}
    \frac{\Haf(A'_{\bm{n}\oplus \bm{n}})}{\Haf(A_{\bm{n}\oplus \bm{n}})} = \frac{\Haf(B + C)}{\Haf(B)} .
\end{align}
Note that this formalism is valid only when $\Haf(B) = \Per(W')^2 \neq 0$; thus, we need to separate the case for $\Per(W') \neq 0$ and $\Per(W') = 0$.

We now introduce a key lemma for our hardness proof, which specifies the required bound on $\left\Vert \delta\Sigma\right\Vert_{F}$ needed to ensure that the left-hand side of Eq.~\eqref{eq: our goal} is bounded by a desired relative error.

\begin{lemma}\label{lemma: multiplicative error is bounded in worst case}
Given an $N_0\times N_0$ matrix $W' \in \{-1,0,1\}^{N_0 \times N_0}$, let $B$ be the $2N\times 2N$ matrix with $N > 2N_0$ defined by
\begin{align}
    B = \begin{pmatrix}
        W & 0 \\
        0 & W 
    \end{pmatrix} \quad \text{for} \quad W  = \begin{pmatrix}
        0_{\frac{ N}{2}} & W'\oplus I_{\frac{N}{2} - N_0} \\ (W')^{T}\oplus I_{\frac{N}{2} - N_0} & 0_{\frac{ N}{2}}
    \end{pmatrix}   ,
\end{align}
and let $C$ be the matrix defined in Eq.~\eqref{eq: definition for C 2}.
If $\delta\Sigma =  e^{-2r}S\textbf{N}_{U}\textbf{N}_{U}^TS^{\dag}/2$ satisfies
\begin{align}
    \max_{U}  \left\Vert \delta\Sigma\right\Vert_{F} 
    \leq \frac{\tanh r_0 \varepsilon' }{36(N_0!)^2N^2M \Vert W \Vert_2 },  
\end{align}
then the following bounds hold:
\begin{align}
   \left\vert 1- \frac{\sqrt{|\Sigma_Q|}}{\sqrt{|\Sigma'_Q|}}\frac{{\rm{Haf}}(B+C)}{{\rm{Haf}}(B)}\right\vert \leq \varepsilon' \quad \text{when} \quad {\rm{Haf}}(B) \neq 0, \\
   \left\vert \frac{\sqrt{|\Sigma_Q|}}{\sqrt{|\Sigma'_Q|}}{\rm{Haf}}(B+C) \right\vert \leq \varepsilon' \quad \text{when} \quad {\rm{Haf}}(B) = 0,
\end{align}
where $\Sigma_{Q}$ is given in Eq.~\eqref{eq: Sigma Q for hardness 2} and $\Sigma_{Q}' = \Sigma_{Q} + \delta\Sigma$.
\end{lemma}
\begin{proof}
    See Sec.~\ref{appendix: section: bounding multiplicative error in worst case}
\end{proof}

\subsubsection{Proof of Theorem~\ref{theorem: strong hardness}}

By collecting all the preceding arguments, we are now ready to prove Theorem~\ref{theorem: strong hardness}.

\begin{proof}[Proof of Theorem~\ref{theorem: strong hardness}]

In this proof, we determined the critical squeezing level at which the \#P-hard problem in Lemma~\ref{theorem: sharp P hardness of computing permanent} becomes solvable within $\rm{BPP}^{NP}$ given oracle access to the randomized classical sampler $\mathcal{M}$. 
More specifically, given an $N_0\times N_0$ matrix $W' = \{-1,0,1\}^{N_0 \times N_0}$, let $g \ll 1$ be a desired multiplicative imprecision level in Lemma~\ref{theorem: sharp P hardness of computing permanent}.
The goal is to compute a multiplicative estimate $\bar{W}$ of $\Per(W')^2$ that satisfies
\begin{align}
    (1-g)\Per(W')^2 \leq \bar{W} \leq (1+g)\Per(W')^2, 
\end{align}
within $\rm{BPP}^{\rm{NP}^{\mathcal{M}}}$.

First, for any fixed $\Lambda > 0$, we construct an $N\times N$ matrix $W$ defined in Eq.~\eqref{eq: definition for W}, where $N = \left\lceil N_0^{1/\lambda} \right\rceil$ for any constant $0<\lambda < \Lambda$.
Next, we decompose $W = YY^{T}$ for an $N\times N$ matrix $Y$, and we construct $M\times M$ unitary matrix $U$ for $M \geq 2N$ such that a top-leftmost $N\times N$ submatrix of $U$ is $\frac{1}{\Vert Y \Vert_2}Y$, as given in Lemma~\ref{lemma: encoding worst-case matrix in unitary matrix}. 
Then, the top-leftmost $N\times N$ submatrix of $U(\mathbb{I}_{N}\oplus 0_{M-N})U^{T}$ is $\frac{1}{\Vert Y \Vert_2^2}YY^T = \frac{1}{\Vert W \Vert_2} W$, and the same for $U^*(\mathbb{I}_{N}\oplus 0_{M-N})U^{\dag}$ since $W$ is a real matrix.

Given the unitary matrix $U$, we set the input state $\rho_{\rm in} = \ket{r_0}\bra{r_0}^{N}\otimes\ket{0}\bra{0}^{M-N}$ for a squeezed vacuum state $\ket{r_0}$ with squeezing level $r_0 = O(1)$.
The ideal output probability $q(\bm{n})$ for this input state $\rho_{\rm in}$ and unitary matrix $U$ is then given by 
\begin{align}\label{eq: definition for q(n) in proof 2}
     q(\bm{n}) = \frac{1}{\sqrt{|\Sigma_Q|}}\Haf(A_{\bm{n}\oplus \bm{n}}) = \frac{\tanh^{N} r_0}{\Vert W \Vert_2^{N} \cosh^{N} r_0}\Haf(W)^2  = \frac{\tanh^{N} r_0}{\Vert W \Vert_2^{N} \cosh^{N} r_0}\Per(W')^2,
\end{align}
where we used $|\Sigma_Q| = |\Sigma_{Q}^{-1}|^{-1} = \cosh^{2N}r_0$ from Eq.~\eqref{eq: H42}.

We first consider the case that $\Per(W') \neq 0$, such that $q(\bm{n})$ is non-zero. 
Note that for the constructed $U$, using Eq.~\eqref{eq: ratio of p(n) and q(n)} and Eq.~\eqref{eq: ration of haf A' and haf A 2}, we have the relation
\begin{align}
    \frac{p(\bm{n})}{q(\bm{n})} = \frac{\sqrt{|\Sigma_Q|}}{\sqrt{|\Sigma'_Q|}}\frac{\Haf(A'_{\bm{n}\oplus\bm{n}})}{\Haf(A_{\bm{n}\oplus \bm{n}})} = \frac{\sqrt{|\Sigma_Q|}}{\sqrt{|\Sigma'_Q|}}\frac{\Haf(B+C)}{\Haf(B)},
\end{align}
for $B = W \oplus W$ and $C$ given in Eq.~\eqref{eq: definition for C 2}, and thus, by Lemma~\ref{lemma: multiplicative error is bounded in worst case}, we have 
\begin{align}\label{eq: error by squeezing 2}
    \left\vert 1- \frac{p(\bm{n})}{q(\bm{n})}  \right\vert \leq \frac{g}{2}   ,
\end{align}
as long as $\left\Vert \delta\Sigma\right\Vert_{F}$ satisfies
\begin{align}\label{eq: condition for reduction 2}
    \max_{U} \left( \left\Vert \delta\Sigma\right\Vert_{F} \right) \leq \frac{\tanh r_0 g }{72(N_0!)^2N^2M \Vert W \Vert_2 } .  
\end{align}

Provided oracle access to $\mathcal{M}$, Stockmeyer's algorithm~\cite{stockmeyer1985approximation}allows us to obtain the multiplicative estimate $\bar{P}$ of $p(\bm{n})$ in complexity class $\rm{BPP}^{\rm{NP}^{\mathcal{M}}}$, which satisfies 
\begin{equation}\label{eq: error by Stockmeyer 2}
   \left\vert 1 - \frac{\bar{P}}{p(\bm{n})} \right\vert \leq  \frac{g}{4}  ,
\end{equation}
for any $g \geq 1/\text{poly}(M)$ in $\text{poly}(M)$ time.
Combining these arguments in Eq.~\eqref{eq: error by squeezing 2} and Eq.~\eqref{eq: error by Stockmeyer 2}, given that Eq.~\eqref{eq: condition for reduction 2} holds, 
we have 
\begin{align}
    \left\vert 1 - \frac{\bar{P}}{q(\bm{n})} \right\vert \leq \left( 1 - \frac{g}{4}\right)\left( 1 - \frac{g}{2} \right) - 1 \leq g .
\end{align}
Therefore, for $q(\bm{n})$ given in Eq.~\eqref{eq: definition for q(n) in proof 2}, we can obtain an multiplicative estimate of $\Per(W')^2$ by $\bar{W} = \frac{\Vert W \Vert_2^{N} \cosh^{N} r_0}{\tanh^{N} r_0}\bar{P}$ from the obtained value $\bar{P}$ via Stockmeyer's algorithm, which satisfies 
\begin{align}\label{eq: estimate W for non zero permanent}
   \left\vert 1 -\frac{\bar{W}}{\Per(W')^2}  \right\vert  \leq g \ll 1,
\end{align}
provided that $\Per(W') \neq 0$ and $\left\Vert \delta\Sigma\right\Vert_{F}$ satisfies Eq.~\eqref{eq: condition for reduction 2}.
This gives a desired multiplicative estimate $\bar{W}$ satisfying $(1-g)\Per(W')^2 \leq \bar{W} \leq (1+g)\Per(W')^2$, in the case that $\Per(W') \neq 0$.

When $\Per(W') = 0$ such that $q(\bm{n}) = 0$, by Lemma~\ref{lemma: multiplicative error is bounded in worst case}, we have 
\begin{align}
   \left\vert \frac{\sqrt{|\Sigma_Q|}}{\sqrt{|\Sigma'_Q|}}{\rm{Haf}}(B+C) \right\vert \leq \frac{g}{2} \quad \Rightarrow \quad p(\bm{n}) =  \frac{1}{\sqrt{|\Sigma'_Q|}}\Haf(A'_{\bm{n}\oplus\bm{n}}) \leq \frac{\tanh^{N} r_0}{\Vert W \Vert_2^{N} \cosh^{N} r_0} \frac{g}{2} ,
\end{align}
when $\left\Vert \delta\Sigma\right\Vert_{F}$ satisfies Eq.~\eqref{eq: condition for reduction 2}.
Hence, by Stockmeyer's algorithm, we can similarly obtain an estimate $\bar{P}$ of $p(\bm{n})$ satisfying Eq.~\eqref{eq: error by Stockmeyer 2}, which is bounded by
\begin{align}
    \bar{P} \leq \left( 1 + \frac{g}{4}\right) p(\bm{n}) 
    \leq \frac{\tanh^{N} r_0}{\Vert W \Vert_2^{N} \cosh^{N} r_0} \frac{g}{2}\left( 1 + \frac{g}{4}\right).
\end{align}
This gives an estimate $\bar{W}$ for $\Per(W')^2$ that satisfies 
\begin{align}
    \bar{W} = \frac{\Vert W \Vert_2^{N} \cosh^{N} r_0}{\tanh^{N} r_0}\bar{P} \leq  \frac{g}{2}\left( 1 + \frac{g}{4}\right) \leq g \ll 1.
\end{align}
Meanwhile, because $\Per(W')^2 \geq 1$ whenever $\Per(W') \neq 0$, the obtained estimate $\bar{W}$ always satisfies $\bar{W} \geq 1-g$ when $\Per(W') \neq 0$ by Eq.~\eqref{eq: estimate W for non zero permanent}. 
Consequently, provided $g \ll 1$, we can distinguish whether $\Per(W') = 0$ or $\Per(W') \neq 0$ by checking if the obtained estimate $\bar{W}$ satisfies $\bar{W} \leq g$ or $\bar{W} \geq 1 - g$. 
Accordingly, if the obtained $\bar{W}$ is observed to be a small constant (say, smaller than $1/2$), then we update the obtained $\bar{W}$ to $\bar{W} = 0$, which gives an exact estimate of $\Per(W')^2 = 0$ provided that $g \ll 1$.
Therefore, for any cases $\Per(W') \neq 0$ or $\Per(W') = 0$, we can obtain a desired multiplicative estimate $\bar{W}$ satisfying 
\begin{align}
    (1-g)\Per(W')^2 \leq \bar{W} \leq (1+g)\Per(W')^2,
\end{align}
in $\rm{BPP}^{NP^{\mathcal{M}}}$.

We conclude the proof by identifying the required squeezing level, i.e., the squeezing level that makes $\left\Vert \delta\Sigma\right\Vert_{F}$ satisfy Eq.~\eqref{eq: condition for reduction 2}. 
By using Lemma~\ref{lemma: norm of noise matrix} and using the fact that $\Vert W \Vert_2 = \Vert W' \oplus I_{\frac{N}{2} - N_0} \Vert_2 \leq \max_{W'} \Vert W' \Vert_{F} = N_0$ and $\tanh r_0 = O(1)$, we can revise the condition from Eq.~\eqref{eq: condition for reduction 2} by
\begin{align}
     e^{-2r}\max_{U} \left( \left\Vert \textbf{N}_{U}\textbf{N}_{U}^T\right\Vert_{F}\right) \leq \frac{\tanh r_0 g }{72 (N_0!)^2 N_0 N^2 M } 
     \quad \Rightarrow \quad
     e^{-2r} \leq  O\left( g (N_0!)^{-2} N_0^{-1} N^{-2} M^{-3} \right)  . 
\end{align}
Here, using $N_0 \leq N^{\lambda} \leq M^{\lambda}$, $N_0 ! = e^{\Theta(N_0\ln N_0) } \leq e^{\Theta(M^{\lambda}\ln M)}$, and $g \geq 1/\text{poly}(M)$, we can obtain the simplified condition given by
\begin{align}\label{eq: condition for r (H63)}
     e^{-2r} \leq e^{-\Omega(M^{\lambda}\ln{M})} .
\end{align}
Lastly, since $\Lambda > \lambda$, we finally have the squeezing level bound $r_{\rm th}$ for the reduction such that 
\begin{align}
       r \geq r_{\rm th} = \Omega\left( M^{\Lambda} \right) , 
\end{align}
which satisfies the condition in Eq.~\eqref{eq: condition for r (H63)}.
This concludes the proof.

\end{proof}

\section{Proof of technical lemmas}

In this section, we provide the proofs of the technical lemmas used in the proofs of our main results.

\subsection{Proof of Lemma~\ref{lemma: minimum eigenvalue for universal brickwork graph state}}\label{appendix: section: lemma for minimum eigenvalue}

In the following, we analyze the minimum eigenvalue of the symmetric product ${\bf N}(G_{U}, \bm{\phi}_U) {\bf N}(G_{U}, \bm{\phi}_U)^{T}$ of the noise matrix ${\bf N}(G_{U}, \bm{\phi}_U)$, corresponding to the universal brickwork graph $G_{U}$ and universal phase-shift angles $\bm{\phi}_U$ defined in Definition~\ref{def: brickwork graph and angles}.

For the convenience of the analysis, we recap the matrix transformation rule introduced in Sec.~\ref{appendix: section: conventions for input output relation}.
Let $G_1 = (V_1,E_1)$ and $G_2 = (V_2, E_2)$ be graphs whose begin and end nodes are assigned with designated sequences.
Also, let $\bm{\phi}_1$ be phase-shift angles for $G_1$ assigned for each vertex of $G_1$ except $\text{end}(G_1)$ with designated sequence (such that, $|\bm{\phi}_2| = |V_1| - |\text{end}(G_1)|$), and similar for $\bm{\phi}_{2}$ and $G_2$.
Then, for concatenation $G_1 \circ G_2$ and sum $G_1 \oplus G_2$ of graph $G_1$ and $G_2$ (where $|\text{begin}(G_2)| = |\text{end}(G_1)|$ for concatenation), together with phase-shift angles $\bm{\phi}_1$ and $\bm{\phi}_2$ for each of $G_1$ and $G_2$, we have the following properties
\begin{align}
    {\bf G}(G_{1} \circ G_{2},\bm{\phi}_1 \cdot \bm{\phi}_2) &= {\bf G}(G_2,\bm{\phi}_2){\bf G}(G_1,\bm{\phi}_1), \label{eq: symplectic for graph concatenation} \\
    {\bf G}(G_{1} \oplus G_{2},\bm{\phi}_1 \cdot \bm{\phi}_2) &= P( {\bf G}(G_1,\bm{\phi}_1) \oplus {\bf G}(G_2,\bm{\phi}_2))P^{T}   , \label{eq: symplectic for graph sum}
    \\
    {\bf N}(G_{1} \circ G_{2},\bm{\phi}_1 \cdot \bm{\phi}_2) &= \left[ {\bf G}(G_2,\bm{\phi}_2) {\bf N}(G_1,\bm{\phi}_1) \;|\; {\bf N}(G_2,\bm{\phi}_2)  \right], \label{eq: noise for graph concatenation} \\
    {\bf N}(G_{1} \oplus G_{2},\bm{\phi}_1 \cdot \bm{\phi}_2) 
    &=  P( {\bf N}(G_1,\bm{\phi}_1) \oplus {\bf N}(G_2,\bm{\phi}_2) ),\label{eq: noise for graph sum}
\end{align}
where $\cdot$ denotes the concatenation of sequences, and $[ A \;|\; B]$ for matrices $A$ and $B$ denotes the column-wise concatenation of the matrices. 
Also, $P$ denotes an $2m\times 2m$ permutation matrix for $m = m_1 + m_2$ with $m_1 = |{\rm{begin}}(G_1)| = |{\rm{end}}(G_1)|$ and $m_2 = |{\rm{begin}}(G_2)| = |{\rm{end}}(G_2)|$ such that
\begin{align}\label{eq: permutation for graph sum}
    P = \begin{pmatrix}
        \mathbb{I}_{m_1} & 0_{m_1} & 0_{m_2} & 0_{m_2} \\ 0_{m_1} & 0_{m_1} & \mathbb{I}_{m_2} & 0_{m_2} \\ 0_{m_1} & \mathbb{I}_{m_1} & 0_{m_2} & 0_{m_2} \\ 0_{m_1} & 0_{m_1} & 0_{m_2} & \mathbb{I}_{m_2}
    \end{pmatrix} ,
\end{align}
which rearranges the quadrature operators from $(\hat{\bm{q}}_{1} \; \hat{\bm{p}}_{1} \; \hat{\bm{q}}_{2} \; \hat{\bm{p}}_{2})^{T}$ to $(\hat{\bm{q}}_{1} \;  \hat{\bm{q}}_{2} \; \hat{\bm{p}}_{1} \; \hat{\bm{p}}_{2})^{T}$, for $(\hat{\bm{q}}_{1} \; \hat{\bm{p}}_{1} )^{T}$ and $(\hat{\bm{q}}_{2} \; \hat{\bm{p}}_{2})^{T}$ being the quadrature operators for $m_1$-mode and $m_2$-mode systems, respectively.

We start with examining the symmetric product ${\bf N}(G_{H}^{(l)}, \bm{\phi}){\bf N}(G_{H}^{(l)}, \bm{\phi})^{T}$ of the noise matrix ${\bf N}(G_{H}^{(l)}, \bm{\phi})$, for a horizontal one-dimensional chain graph $G_{H}^{(l)}$ with $l \geq 2$ number of vertices, and arbitrarily chosen angles $\bm{\phi} = \{\phi_1,\dots,\phi_{l-1}\}$ satisfying $\forall\phi_{i} \in \left(-\frac{\pi}{2}, \frac{\pi}{2} \right)$.

\begin{lemma}\label{lemma: noise matrix for horizontal chain}
    Let ${\bf N}(G_{H}^{(l)}, \bm{\phi})$ be a noise matrix for a horizontal one-dimensional chain graph $G_{H}^{(l)}$ with $l \geq 3$ number of vertices, and arbitrarily chosen angles $\bm{\phi} = \{\phi_1,\dots,\phi_{l-1}\}$ for $\forall\phi_{i} \in \left(-\frac{\pi}{2}, \frac{\pi}{2} \right)$. 
    Then, ${\bf N}(G_{H}^{(l)}, \bm{\phi}){\bf N}(G_{H}^{(l)}, \bm{\phi})^{T}$ can be represented as 
    \begin{align}
        {\bf N}(G_{H}^{(l)}, \bm{\phi}){\bf N}(G_{H}^{(l)}, \bm{\phi})^{T} = \mathbb{I}_{2} + D,
    \end{align}
    for a $2\times 2$ positive semi-definite matrix $D$ determined by $l$ and $\bm{\phi}$.
\end{lemma}

\begin{proof}
We start with the fact that $G_{H}^{(l)}$ is a concatenation of $l-1$ number of two-mode graph $G_{H}^{(2)}$, such that $G_{H}^{(l)} = G_{H}^{(2)} \circ \cdots \circ G_{H}^{(2)}$.
Let us define $\textbf{N}_{i} = {\bf N}(G_{H}^{(2)}, \phi_{i})$ for $i \in [l-1]$, and define $\textbf{G}_{i} = {\bf G}(G_{H}^{(2)}, \phi_{i})$.
Also, we define $\textbf{G} = \textbf{G}_{l-1}\textbf{G}_{l-2}\cdots \textbf{G}_{1}$ as a total product of $\textbf{G}_{i}\;\forall i\in[l-1]$, and define $\bar{\textbf{G}}_{k} = \textbf{G}\textbf{G}_{1}^{-1}\textbf{G}_{2}^{-1}\cdots \textbf{G}_{k}^{-1}$ (for example, $ \bar{\textbf{G}}_{l-1} = \mathbb{I}_{2}$, and $ \bar{\textbf{G}}_{l-2} =  \textbf{G}_{l-1}$).
Then, by Eq.~\eqref{eq: symplectic for graph concatenation} and Eq.~\eqref{eq: noise for graph concatenation}, ${\bf{N}} = {\bf N}(G_{H}^{(l)}, \bm{\phi})$ can be expressed as a concatenation of noise matrices such that
\begin{align}
    {\bf N}(G_{H}^{(l)}, \bm{\phi}) = \left[ \bar{\textbf{G}}_{1}\textbf{N}_{1} \;|\; \bar{\textbf{G}}_{2}\textbf{N}_{2} \;|\; \cdots \;|\;  \bar{\textbf{G}}_{l-2}\textbf{N}_{l-2}  \;|\;  \bar{\textbf{G}}_{l-1}\textbf{N}_{l-1}  \right], 
\end{align}
where $[ A \;|\; B]$ denotes the column-wise concatenation of the matrices $A$ and $B$ as noted.
Hence, we obtain
\begin{align}\label{eq: expansion for NN^T}
    {\bf N}(G_{H}^{(l)}, \bm{\phi}){\bf N}(G_{H}^{(l)}, \bm{\phi})^{T} = \bar{\textbf{G}}_{1}\textbf{N}_{1}\textbf{N}_{1}^{T}\bar{\textbf{G}}_{l}^{T} +  \bar{\textbf{G}}_{2}\textbf{N}_{2}\textbf{N}_{2}^{T}\bar{\textbf{G}}_{2}^{T} + \cdots + \bar{\textbf{G}}_{l-2}\textbf{N}_{l-2}\textbf{N}_{l-2}^{T}\bar{\textbf{G}}_{l-2}^{T} +\bar{\textbf{G}}_{l-1}\textbf{N}_{l-1}\textbf{N}_{l-1}^{T}\bar{\textbf{G}}_{l-1}^{T}  .
\end{align}
Here, let us concentrate on the last two terms in the right-hand side of Eq.~\eqref{eq: expansion for NN^T}. 
First, as we have found previously in Eq.~\eqref{eq: input output relation for 2 mode cluster}, $\textbf{N}_{i} = \begin{pmatrix}
    0 \\ 1
\end{pmatrix}$ for all $i \in [l-1]$, and accordingly, $\textbf{N}_{i}\textbf{N}_{i}^{T} = \begin{pmatrix}
    0 & 0 \\ 0 & 1
\end{pmatrix}$ for all $i \in [l-1]$ (as also described in in~\cite{alexander2014noise, menicucci2014fault}).
Hence, given that $\bar{\textbf{G}}_{l-1} = \mathbb{I}_{2}$, we have 
\begin{align}
    \bar{\textbf{G}}_{l-1}\textbf{N}_{l-1}\textbf{N}_{l-1}^{T}\bar{\textbf{G}}_{l-1}^{T} = \begin{pmatrix}
    0 & 0 \\ 0 & 1
\end{pmatrix}.
\end{align}
Also, as given in Eq.~\eqref{eq: input output relation for 2 mode cluster}, $\textbf{G}_{i} = \begin{pmatrix}
    -\tan\phi_{i} & -1 \\ 1 & 0
\end{pmatrix}$, and thus, $\bar{\textbf{G}}_{l-2} =  \textbf{G}_{l-1} = \begin{pmatrix}
    -\tan\phi_{l-1} & -1 \\ 1 & 0
\end{pmatrix}$.
Accordingly, we also have 
\begin{align}
  \bar{\textbf{G}}_{l-2}\textbf{N}_{l-2}\textbf{N}_{l-2}^{T}\bar{\textbf{G}}_{l-2}^{T} = \begin{pmatrix}
    1 & 0 \\ 0 & 0
\end{pmatrix}  .
\end{align}
Crucially, the summation of the last two terms on the right-hand side of Eq.~\eqref{eq: expansion for NN^T} yields the identity matrix $\mathbb{I}_{2}$.
Hence, we finally have
\begin{align}\label{eq: final expansion for NN^T in 1-d chain}
     {\bf N}(G_{H}^{(l)}, \bm{\phi}){\bf N}(G_{H}^{(l)}, \bm{\phi})^{T} 
     =
     \begin{cases}
        \mathbb{I}_2  \quad \text{for} \quad l = 3  , \\
        \mathbb{I}_2 + \sum_{i = 1}^{l-3} \bar{\textbf{G}}_{i}\textbf{N}_{i}\textbf{N}_{i}^{T}\bar{\textbf{G}}_{i}^{T} \quad \text{for} \quad l \geq 4  . \\
    \end{cases}
\end{align}
Note that for arbitrarily given $\bar{\textbf{G}}_{i}$, the last term in the right-hand side of Eq.~\eqref{eq: final expansion for NN^T in 1-d chain} is a summation of positive semi-definite matrices, which is also a positive semi-definite matrix.
This concludes the proof.


\end{proof}

Next, we examine the noise matrix ${\bf N}(G_{T}, \bm{\phi}_{T})$.
Recall that, $G_{T}$ denotes a brick graph, and $\bm{\phi}_{T}$ represents the corresponding brick phase-shift angles defined in Definition~\ref{def: brickwork cluster state and angle}, such that by cluster state $\ket{G_{T}}$ one can implement ${\bf G}(G_{T},\bm{\phi}_{T} ) = \begin{pmatrix}
        T_{R} & -T_{I} \\ T_{I} & T_{R}
    \end{pmatrix}$ for real part $T_{R}$ and imaginary part $T_{I}$ of a brick beam splitter $T = \begin{pmatrix}
        e^{i\beta}\cos\theta & -\sin\theta \\ e^{i\beta}\sin\theta & \cos\theta
    \end{pmatrix}$.
We find that for an arbitrary $T$, the minimum eigenvalue of ${\bf N}(G_{T}, \bm{\phi}_{T}){\bf N}(G_{T}, \bm{\phi}_{T})^{T}$ is lower bounded by a positive constant $3-2\sqrt{2}$.

\begin{lemma}\label{lemma: noise matrix for brickwork cluster state}
    Let ${\bf N}(G_{T}, \bm{\phi}_{T})$ be a noise matrix for the brick graph $G_{T}$ and the brick phase-shift angles $\bm{\phi}_{T}$ defined in Definition~\ref{def: brickwork cluster state and angle}. 
    Then, the minimum eigenvalue $\lambda_{min}$ of ${\bf N}(G_{T}, \bm{\phi}_{T}){\bf N}(G_{T}, \bm{\phi}_{T})^T$ satisfies 
    \begin{align}
        \lambda_{min} \geq 3-2\sqrt{2} .
    \end{align}
\end{lemma}

\begin{proof}

By definition, $G_{T} = G_{\sqsupset} \circ G_{\sqsupset} \circ G_{\sqsupset}$ and $\bm{\phi}_{T} = \bm{\phi}_{b} \cdot \bm{\phi}_{a} \cdot \bm{\phi}_{a}$.
Using Eq.~\eqref{eq: symplectic for graph concatenation} and Eq.~\eqref{eq: noise for graph concatenation}, we have
\begin{align}
    {\bf N}(G_{T}, \bm{\phi}_{T}) = \left[ {\bf G}(G_{\sqsupset} \circ G_{\sqsupset}, \bm{\phi}_{a} \cdot \bm{\phi}_{a} ) {\bf N}(G_{\sqsupset}, \bm{\phi}_{b})   \;|\;   {\bf G}(G_{\sqsupset} , \bm{\phi}_{a} ){\bf N}(G_{\sqsupset}, \bm{\phi}_{a})  \;|\; {\bf N}(G_{\sqsupset}, \bm{\phi}_{a}) \right],
\end{align}
and thus, we have
\begin{align}\label{eq: I13}
    {\bf N}(G_{T}, \bm{\phi}_{T}){\bf N}(G_{T}, \bm{\phi}_{T})^{T} &= {\bf G}(G_{\sqsupset} \circ G_{\sqsupset}, \bm{\phi}_{a} \cdot \bm{\phi}_{a} ) {\bf N}(G_{\sqsupset}, \bm{\phi}_{b}){\bf N}(G_{\sqsupset}, \bm{\phi}_{b})^{T} {\bf G}(G_{\sqsupset} \circ G_{\sqsupset}, \bm{\phi}_{a} \cdot \bm{\phi}_{a} )^{T} \nonumber \\
    &+ {\bf G}(G_{\sqsupset} , \bm{\phi}_{a} ){\bf N}(G_{\sqsupset}, \bm{\phi}_{a}) {\bf N}(G_{\sqsupset}, \bm{\phi}_{a})^{T} {\bf G}(G_{\sqsupset} , \bm{\phi}_{a} )^{T} 
    + {\bf N}(G_{\sqsupset}, \bm{\phi}_{a}){\bf N}(G_{\sqsupset}, \bm{\phi}_{a})^{T} .
\end{align}
Hence, the minimum eigenvalue of ${\bf N}(G_{T}, \bm{\phi}_{T}){\bf N}(G_{T}, \bm{\phi}_{T})^{T}$ is at least as large as the minimum eigenvalue of ${\bf N}(G_{\sqsupset}, \bm{\phi}_{a}){\bf N}(G_{\sqsupset}, \bm{\phi}_{a})^{T}$, because the first two terms in the right-hand side of Eq.~\eqref{eq: I13} are positive semi-definite.
Also, using the fact that $G_{\sqsupset}= (G_{H}^{(5)} \oplus G_{H}^{(5)})\circ G_{V}^{(3)}$ and using Eq.~\eqref{eq: noise for graph concatenation} and Eq.~\eqref{eq: noise for graph sum}, observe that
\begin{align}
    &{\bf N}(G_{\sqsupset}, \bm{\phi}_{a}){\bf N}(G_{\sqsupset}, \bm{\phi}_{a})^{T}  \nonumber \\
    &=  {\bf G}(G_{V}^{(3)}, \bm{\phi}_{\rm{int}} ) P \left( {\bf N}(G_{H}^{(5)}, \bm{\phi}^{(1)}){\bf N}(G_{H}^{(5)}, \bm{\phi}^{(1)})^{T}  \oplus {\bf N}(G_{H}^{(5)}, \bm{\phi}^{(2)}) {\bf N}(G_{H}^{(5)}, \bm{\phi}^{(2)})^{T} \right) P^{T} {\bf G}(G_{V}^{(3)}, \bm{\phi}_{\rm{int}} )^{T}  \nonumber \\
    &+ {\bf N}(G_{V}^{(3)}, \bm{\phi}_{\rm{int}} ){\bf N}(G_{V}^{(3)}, \bm{\phi}_{\rm{int}} )^{T} \label{eq: I14} \\
    &= {\bf G}(G_{V}^{(3)}, \bm{\phi}_{\rm{int}} ){\bf G}(G_{V}^{(3)}, \bm{\phi}_{\rm{int}} )^{T} + {\bf G}(G_{V}^{(3)}, \bm{\phi}_{\rm{int}} )PDP^{T}{\bf G}(G_{V}^{(3)}, \bm{\phi}_{\rm{int}} )^{T} + {\bf N}(G_{V}^{(3)}, \bm{\phi}_{\rm{int}} ){\bf N}(G_{V}^{(3)}, \bm{\phi}_{\rm{int}} )^{T} , 
\end{align}
where $P$ is a $4\times 4$ permutation matrix (given in Eq.~\eqref{eq: permutation for graph sum}), and we have employed the convention $\bm{\phi}_{a} = \bm{\phi}^{(1)} \cdot \bm{\phi}^{(2)} \cdot \bm{\phi}_{\rm{int}}$ as in Fig.~\ref{fig: graph G_subseteq}, and also used the fact that the sandwiched term in the right-hand side of Eq.~\eqref{eq: I14} can be expressed as
\begin{align}
    {\bf N}(G_{H}^{(5)}, \bm{\phi}^{(1)}){\bf N}(G_{H}^{(5)}, \bm{\phi}^{(1)})^{T}  \oplus {\bf N}(G_{H}^{(5)}, \bm{\phi}^{(2)}) {\bf N}(G_{H}^{(5)}, \bm{\phi}^{(2)})^{T} = \mathbb{I}_{4} + D ,
\end{align}
for a $4\times 4$ positive semi-definite matrix $D$ by Lemma~\ref{lemma: noise matrix for horizontal chain}. 
Combining these arguments, the minimum eigenvalue of ${\bf N}(G_{T}, \bm{\phi}_{T}){\bf N}(G_{T}, \bm{\phi}_{T})^{T}$ is lower bounded by the minimum eigenvalue of ${\bf G}(G_{V}^{(3)}, \phi_{\rm{int}} ) {\bf G}(G_{V}^{(3)}, \phi_{\rm{int}} ) ^{T}$, which is given by $1 + 2\tan\phi_{\rm{int}}\left(\tan\phi_{\rm{int}} - \sqrt{1 + \tan^2\phi_{\rm{int}}}\right)$ from Eq.~\eqref{eq: input output relation for vertical cluster state}.
Because $\tan\phi_{\rm{int}} \in [-1,1]$ from Eq.~\eqref{eq: phase-shift angles a for brickwork}, this minimum eigenvalue is further minimized when $\tan\phi_{\rm{int}} = 1$, which gives $3-2\sqrt{2}$.
Therefore, the minimum eigenvalue of ${\bf N}(G_{T}, \bm{\phi}_{T}){\bf N}(G_{T}, \bm{\phi}_{T})^{T}$ is bounded from below by $3-2\sqrt{2}$.

\end{proof}

Based on the above arguments, we are finally ready to prove Lemma~\ref{lemma: minimum eigenvalue for universal brickwork graph state}.
In the following, we show that the minimum eigenvalue of $ {\bf N}(G_{U}, \bm{\phi}_{U}) {\bf N}(G_{U}, \bm{\phi}_{U})^{T}$ for the universal brickwork graph $G_{U}$ and universal phase-shift angles $\bm{\phi}_U$ (defined in Definition~\ref{def: brickwork graph and angles}) is bounded from below by $(3-2\sqrt{2})M + 1$.

\begin{proof}[Proof of Lemma~\ref{lemma: minimum eigenvalue for universal brickwork graph state}]

Recall that, as defined in Definition~\ref{def: brickwork graph and angles}, the universal brickwork graph $G_{U}$ can be expressed as
\begin{align}\label{eq: redef for universal graph}
    G_{U} = \underbrace{G_{D}^{(M)} \circ G_{D}^{(M)}\circ \cdots \circ G_{D}^{(M)}}_{k\geq\frac{M}{2}+1\;\text{times}} \quad \text{where} \quad G_{D}^{(M)} = ( \underbrace{G_{T} \oplus G_{T} \oplus \cdots \oplus G_{T}}_{\frac{M}{2} \; \text{times}} )\circ (G_{H}^{(13)} \oplus \underbrace{G_{T}\oplus G_{T}\oplus \cdots \oplus G_{T}}_{\frac{M}{2} - 1 \; \text{times}} \oplus \, G_{H}^{(13)}) .
\end{align}
Also, for $j \in [t]$ with $t = k - 1 \geq M/2$, the set of universal phase-shift angles $\bm{\phi}_{U}$ is given by 
\begin{align}\label{eq: redef for universal phase}
    \bm{\phi}_{U} =  \bm{\phi}_{T_{1}}\cdot \cdots\cdot\bm{\phi}_{T_{t}}\bm{\phi}_{\rm ps} \quad \text{where} \quad \bm{\phi}_{T_{j}} = \bm{\phi}_{T_{1,j}}\cdot\bm{\phi}_{T_{2,j}}\cdot \cdots \cdot \bm{\phi}_{T_{\frac{M}{2},j}} \cdot \bm{\phi}_0 \cdot \bm{\phi}_{T_{\frac{M}{2}+1, j}} \cdot \cdots \cdot \bm{\phi}_{T_{M-1, j}}\cdot \bm{\phi}_0 .
\end{align}
Here, each $\bm{\phi}_{T_{i,j}}$ denotes the set of brick phase-shift angles that gives ${\bf G}(G_{T}, \bm{\phi}_{T_{i,j}}) = \begin{pmatrix}
     (T_{i,j})_{R} & -(T_{i,j})_{I} \\ (T_{i,j})_{I} & (T_{i,j})_{R}
\end{pmatrix}$ for real $(T_{i,j})_{R}$ and imaginary $(T_{i,j})_{R}$ part of the brick beam splitter $T_{i,j} = \begin{pmatrix}
        e^{i\beta_{i,j}}\cos\theta_{i,j} & -\sin\theta_{i,j} \\ e^{i\beta_{i,j}}\sin\theta_{i,j} & \cos\theta_{i,j}
\end{pmatrix}$ (for $i \in [M-1] $ and $j \in [k]$),
and $\bm{\phi}_0$ is an all-zero phase-shift angles $\bm{\phi}_0 = \{0,\dots,0\}$ that gives an identity matrix ${\bf G}(G_{H}^{(13)}, \bm{\phi}_0) = \mathbb{I}_2$.
Also, $\bm{\phi}_{\rm ps}$ denotes the set of phase-shift angles that, when combined with the graph $G_{D}^{(M)}$, implements the final phase shift ${\bf G}(G_{D}^{(M)}, \bm{\phi}_{\rm ps}) = \begin{pmatrix}
        (U_{\rm ps})_{R} & -(U_{\rm ps})_{I} \\ (U_{\rm ps})_{I} & (U_{\rm ps})_{R} 
\end{pmatrix} $ as appeared in Eq.~\eqref{eq: decomposing unitary circuit} and discussed in Sec.~\ref{section: cluster state for universal LO circuit}.

Now, for the graph $G_{D}^{(M)}$ and the set of phase-shift angles $\bm{\phi}_{T_{j}}$ defined in Eq.~\eqref{eq: redef for universal phase}, by using matrix transformation rules from Eq.~\eqref{eq: symplectic for graph concatenation} to Eq.~\eqref{eq: noise for graph sum}, observe that
\begin{align}
    &{\bf N}(G_{D}^{(M)}, \bm{\phi}_{T_{j}}){\bf N}(G_{D}^{(M)}, \bm{\phi}_{T_{j}})^{T} \nonumber \\
    &= P_1 \left( {\bf N}(G_{H}^{(13)}, \bm{\phi}_{0}) {\bf N}(G_{H}^{(13)}, \bm{\phi}_{0})^{T} \oplus \left( \bigoplus_{i = \frac{M}{2}+1}^{M-1}  {\bf N}(G_{T}, \bm{\phi}_{T_{i,j}}){\bf N}(G_{T}, \bm{\phi}_{T_{i,j}} )^{T} \right) \oplus {\bf N}(G_{H}^{(13)}, \bm{\phi}_{0}) {\bf N}(G_{H}^{(13)}, \bm{\phi}_{0})^{T} \right) P_1^{T} \nonumber \\
    &+ P_2\left( \mathbb{I}_{2} \oplus  \left( \bigoplus_{i = \frac{M}{2}+1}^{M-1}  {\bf G}(G_{T}, \bm{\phi}_{T_{i,j}}) \right)\oplus \mathbb{I}_{2} \right) \left( \bigoplus_{i = 1}^{\frac{M}{2}}  {\bf N}(G_{T}, \bm{\phi}_{T_{i,j}}){\bf N}(G_{T}, \bm{\phi}_{T_{i,j}} )^{T} \right) \left( \mathbb{I}_{2} \oplus  \left( \bigoplus_{i = \frac{M}{2}+1}^{M-1}  {\bf G}(G_{T}, \bm{\phi}_{T_{i,j}}) \right)\oplus \mathbb{I}_{2} \right)^{T} P_2^{T}, \label{eq: I20}
\end{align}
for $2M\times 2M$ permutation matrices $P_1$ and $P_2$. 
Note that by definition, ${\bf G}(G_{T}, \bm{\phi}_{T_{i,j}})$ is given by an orthogonal matrix, such that the second term in the right-hand side of Eq.~\eqref{eq: I20} is the orthogonal transformation of the snadwiched term, which preserves its eigenvalues.
Also, by Lemma~\ref{lemma: noise matrix for brickwork cluster state}, the minimum eigenvalue of ${\bf N}(G_{T}, \bm{\phi}_{T_{i,j}}){\bf N}(G_{T}, \bm{\phi}_{T_{i,j}})^{T}$ is lower bounded by $3 - 2\sqrt{2}$. 
Lastly, as given in Lemma~\ref{lemma: noise matrix for horizontal chain}, the minimum eigenvalue of $ {\bf N}(G_{H}^{(13)}, \bm{\phi}_{0}){\bf N}(G_{H}^{(13)}, \bm{\phi}_{0})^{T}$ is lower bounded by $1$.
Combining all these arguments, by Weyl's inequality~\cite{weyl1912asymptotische}, the minimum eigenvalue of ${\bf N}(G_{D}^{(M)}, \bm{\phi}_{T_{j}}){\bf N}(G_{D}^{(M)}, \bm{\phi}_{T_{j}})^{T}$ is at least $6 - 4\sqrt{2}$.

Therefore, from Eq.~\eqref{eq: redef for universal graph} and Eq.~\eqref{eq: redef for universal phase}, using the matrix transformation rules, for $t = k - 1 \geq M/2$, we have
\begin{align}\label{eq: noise matrix for universal brickwork}
     & {\bf N}(G_{D}^{(M,t)}, \bm{\phi}_{T_{1}}\cdot \cdots\cdot\bm{\phi}_{T_{t}} ) {\bf N}(G_{D}^{(M,t)}, \bm{\phi}_{T_{1}}\cdot \cdots\cdot\bm{\phi}_{T_{t}} )^{T} \nonumber \\
     &= {\bf G}(G_{D}^{(M,t-1)},\bm{\phi}_{T_{2}}\cdot \cdots\cdot\bm{\phi}_{T_{t}}  ) {\bf N}(G_{D}^{(M)}, \bm{\phi}_{T_{1}}){\bf N}(G_{D}^{(M)}, \bm{\phi}_{T_{1}})^{T}{\bf G}(G_{D}^{(M,t-1)},\bm{\phi}_{T_{2}}\cdot \cdots\cdot\bm{\phi}_{T_{t}}  )^{T} \nonumber \\
     &+  {\bf G}(G_{D}^{(M,t-2)},\bm{\phi}_{T_{3}}\cdot \cdots\cdot\bm{\phi}_{T_{t}}  ) {\bf N}(G_{D}^{(M)}, \bm{\phi}_{T_{2}}){\bf N}(G_{D}^{(M)}, \bm{\phi}_{T_{2}})^{T}{\bf G}(G_{D}^{(M,t-2)},\bm{\phi}_{T_{3}}\cdot \cdots\cdot\bm{\phi}_{T_{t}}  )^{T} \nonumber \\
     &+ \cdots \nonumber \\
     &+ {\bf G}(G_{D}^{(M)},\bm{\phi}_{T_{t}}  ) {\bf N}(G_{D}^{(M)}, \bm{\phi}_{T_{t-1}}){\bf N}(G_{D}^{(M)}, \bm{\phi}_{T_{t-1}})^{T} {\bf G}(G_{D}^{(M)},\bm{\phi}_{T_{t}}  )^{T} \nonumber \\
     &+ {\bf N}(G_{D}^{(M)}, \bm{\phi}_{T_{t}}){\bf N}(G_{D}^{(M)}, \bm{\phi}_{T_{t}})^{T}. 
\end{align}
Because orthogonal transformation of a matrix preserves its eigenvalues, by Weyl's inequality~\cite{weyl1912asymptotische}, the minimum eigenvalue of ${\bf N}(G_{D}^{(M,t)}, \bm{\phi}_{T_{1}}\cdot \cdots\cdot\bm{\phi}_{T_{t}} ) {\bf N}(G_{D}^{(M,t)}, \bm{\phi}_{T_{1}}\cdot \cdots\cdot\bm{\phi}_{T_{t}} )^{T}$ is lower bounded by the sum of minimum eigenvalues of ${\bf N}(G_{D}^{(M)}, \bm{\phi}_{T_{j}}){\bf N}(G_{D}^{(M)}, \bm{\phi}_{T_{j}})^{T}$ for all $j \in [t]$, which is $(6 - 4\sqrt{2})t$.

Lastly, considering the final phase shifts at the end, we have
\begin{align}
    {\bf N}(G_{U}, \bm{\phi}_{U}) {\bf N}(G_{U}, \bm{\phi}_{U})^{T}
    &= {\bf G}(G_{D}^{(M)}, \bm{\phi}_{\rm ps}) {\bf N}(G_{D}^{(M,t)}, \bm{\phi}_{T_{1}}\cdot \cdots\cdot\bm{\phi}_{T_{t}} ) {\bf N}(G_{D}^{(M,t)}, \bm{\phi}_{T_{1}}\cdot \cdots\cdot\bm{\phi}_{T_{t}} )^{T} {\bf G}(G_{D}^{(M)}, \bm{\phi}_{\rm ps})^{T} \nonumber \\
    &+ {\bf N}(G_{D}^{(M)}, \bm{\phi}_{\rm ps}){\bf N}(G_{D}^{(M)}, \bm{\phi}_{\rm ps})^{T} .
\end{align}
Here, when operating on $G_{D}^{(M)}$ to implement the single-mode phase-shift operation ${\bf G}(G_{D}^{(M)}, \bm{\phi}_{\rm ps})$, the implementation effectively reduces to that of a one-dimensional chain-structured cluster state.
This is because all phase-shift angles are set to zero for modes not used to implement the single-mode phase shift, leaving only the relevant modes (i.e., the horizontal chain) active during the computation.
Hence, by Lemma~\ref{lemma: noise matrix for horizontal chain}, the minimum eigenvalue of ${\bf N}(G_{D}^{(M)}, \bm{\phi}_{\rm ps}){\bf N}(G_{D}^{(M)}, \bm{\phi}_{\rm ps})^{T}$ is lower bounded by $1$. 
Then, as argued previously, the minimum eigenvalue of ${\bf N}(G_{U}, \bm{\phi}_{U}) {\bf N}(G_{U}, \bm{\phi}_{U})^{T}$ can be lower bounded by $(6 - 4\sqrt{2})t + 1$.
Since $t \geq \frac{M}{2}$, this yields the desired bound.

\end{proof}

\subsection{Proof of Lemma~\ref{lemma: norm of noise matrix}}\label{appendix: section: norm of noise matrix}

Building upon the arguments developed in Sec.~\ref{appendix: section: easiness} and Sec.~\ref{appendix: section: lemma for minimum eigenvalue}, we expand the Frobenius norm as follows: 
\begin{align}
    \left\Vert \textbf{N}_{U}\textbf{N}_{U}^T \right\Vert_{F} 
    &\leq
    \left\Vert \mathbb{I}_{2M} \right\Vert_{F}  +  \left\Vert {\bf N}(G_{U}, \bm{\phi}_U) {\bf N}(G_{U}, \bm{\phi}_U)^{T} \right\Vert_{F}   \label{a12} \\
    &\leq
    \left\Vert \mathbb{I}_{2M} \right\Vert_{F}  +  \sum_{j = 1}^{k - 1} \left\Vert {\bf N}(G_{D}^{(M)}, \bm{\phi}_{T_{j}}){\bf N}(G_{D}^{(M)}, \bm{\phi}_{T_{j}})^{T} \right\Vert_{F} +   \left\Vert {\bf N}(G_{D}^{(M)}, \bm{\phi}_{\rm ps}){\bf N}(G_{D}^{(M)}, \bm{\phi}_{\rm ps})^{T} \right\Vert_{F}  \label{a13} \\
    &\leq 
    \left\Vert \mathbb{I}_{2M} \right\Vert_{F}  +  \sum_{j = 1}^{k - 1} \left( \sum_{i = 1}^{M-1} \left\Vert {\bf N}(G_{T}, \bm{\phi}_{T_{i,j}}){\bf N}(G_{T}, \bm{\phi}_{T_{i,j}} )^{T}\right\Vert_{F}  +  2\left\Vert {\bf N}(G_{H}^{(13)}, \bm{\phi}_{0}) {\bf N}(G_{H}^{(13)}, \bm{\phi}_{0})^{T} \right\Vert_{F}  \right)  \nonumber \\
    &+  \left\Vert {\bf N}(G_{D}^{(M)}, \bm{\phi}_{\rm ps}){\bf N}(G_{D}^{(M)}, \bm{\phi}_{\rm ps})^{T} \right\Vert_{F}  , \label{a14}
\end{align}
where we sequentially applied the triangle inequality for the Frobenius norm, and specifically, used Eq.~\eqref{eq: matrices for our setup} in Eq.~\eqref{a12}, used Eq.~\eqref{eq: noise matrix for universal brickwork} in Eq.~\eqref{a13}, and used Eq.~\eqref{eq: I20} in Eq.~\eqref{a14}. 
First, from Lemma~\ref{lemma: noise matrix for horizontal chain}, we obtain
\begin{align}
    \left\Vert {\bf N}(G_{H}^{(13)}, \bm{\phi}_{0}) {\bf N}(G_{H}^{(13)}, \bm{\phi}_{0})^{T} \right\Vert_{F}  \leq  \sum_{i=1}^{12}  \left\Vert {\bf N}(G_{H}^{(2)}, 0 ) {\bf N}(G_{H}^{(2)}, 0 )^{T} \right\Vert_{F} = 12.
\end{align}
Moreover, by Definition~\ref{def: brickwork cluster state and angle} we have
\begin{align}
    \left\Vert {\bf N}(G_{T}, \bm{\phi}_{T_{i,j}}){\bf N}(G_{T}, \bm{\phi}_{T_{i,j}} )^{T} \right\Vert_{F}  
    \leq 
    \left\Vert {\bf N}(G_{\sqsupset}, \bm{\phi}_{b}) {\bf N}(G_{\sqsupset}, \bm{\phi}_{b})^{T} \right\Vert_{F}  + 2 \left\Vert {\bf N}(G_{\sqsupset}, \bm{\phi}_{a}) {\bf N}(G_{\sqsupset}, \bm{\phi}_{a})^{T} \right\Vert_{F}   .
\end{align}
By further decomposing the noise matrices as in Eq.~\eqref{eq: I14}, for the phase-shift angles $\bm{\phi}_{a}$ and $\bm{\phi}_{b}$ in Eq.~\eqref{eq: phase-shift angles a for brickwork} and Eq.~\eqref{eq: phase-shift angles b for brickwork} that implement an arbitrary brick beam splitter $T_{i,j}$, the right-hand side of the above expression is bounded by a constant.
This follows from the fact that the dimension and all elements of the noise matrices are bounded by constants.

Similarly, we have
\begin{align}
    \left\Vert {\bf N}(G_{D}^{(M)}, \bm{\phi}_{\rm ps}){\bf N}(G_{D}^{(M)}, \bm{\phi}_{\rm ps})^{T} \right\Vert_{F} 
    \leq 
    \sum_{i = 1}^{M-1} \left\Vert {\bf N}(G_{T}, \bm{\phi}_{i}'){\bf N}(G_{T}, \bm{\phi}_{i}'  )^{T}\right\Vert_{F}  +  2\left\Vert {\bf N}(G_{H}^{(13)}, \bm{\phi}_{0}) {\bf N}(G_{H}^{(13)}, \bm{\phi}_{0})^{T} \right\Vert_{F}  ,
\end{align}
where $\bm{\phi}_{i}'$ denotes a set of phase-shift angles that, when combined with the graph $G_{T}$, implement the target phase-shift operations. 
As in the previous case, each term on the right-hand side is bounded by a constant.

Taken together, from Eq.~\eqref{a14}, we finally have 
\begin{align}
    \left\Vert \textbf{N}_{U}\textbf{N}_{U}^T \right\Vert_{F} \leq O(Mk)   ,
\end{align}
where $k \geq M/2 + 1$ by Definition~\ref{def: brickwork graph and angles}, thus concluding the proof.

\subsection{Proof of Lemma~\ref{lemma: multiplicative coefficient is bounded}}\label{appendix: section: lemma for bounding multiplicative coefficient}

For convenience, we introduce three observations that will be frequently used throughout the proof. 

\begin{observation}
The noise matrix $\delta\Sigma = \frac{e^{-2r}}{2}S{\bf N}_{U}{\bf N}_{U}^TS^{\dag}$ in Eq.~\eqref{eq: definition of noise matrix} is positive semi-definite. 
\end{observation}

\begin{observation}
The covariance matrix $\Sigma_{Q}$ and its inverse $\Sigma_{Q}^{-1}$ are positive definite. 
\end{observation}

\begin{observation}\label{observation: invertible}
$\mathbb{I}_{2M} + \delta\Sigma\cdot\Sigma_Q^{-1}$ has only positive eigenvalues and thus invertible. 
\end{observation}

Based on these observations, we prove Lemma~\ref{lemma: multiplicative coefficient is bounded}.

\begin{proof}[Proof of Lemma~\ref{lemma: multiplicative coefficient is bounded}]

Observe that
\begin{align}
    \text{det}(\Sigma'_{Q}) = \det(\Sigma_Q + \delta\Sigma) = \det(\Sigma_{Q})\det(\mathbb{I}_{2M}+\delta\Sigma\cdot\Sigma_Q^{-1}) ,
\end{align}
and thus we have a simplified form
\begin{align}\label{eq: determinant ratio}
    \frac{\sqrt{|\Sigma_Q|}}{\sqrt{|\Sigma'_Q|}} = \frac{1}{\sqrt{\det(\mathbb{I}_{2M}+\delta\Sigma\cdot\Sigma_Q^{-1})}} .
\end{align}
Throughout the proof, we investigate the lower and upper bounds of Eq.~\eqref{eq: determinant ratio}, under the assumption that $\left| \Tr(\delta\Sigma\cdot\Sigma_Q^{-1})\right|  \leq 1/3$.

We first examine the lower bound.
Given that $\mathbb{I}_{2M} + \delta\Sigma\cdot\Sigma_Q^{-1}$ has only positive eigenvalues, let us denote $\{\lambda_i\}_{i=1}^{2M}$ as a set of eigenvalues of $\mathbb{I}_{2M} + \delta\Sigma\cdot\Sigma_Q^{-1}$ which are all positive. 
Then, we have 
\begin{align}
    \log\det(\mathbb{I}_{2M} + \delta\Sigma\cdot\Sigma_Q^{-1}) = \sum_{i=1}^{2M}\log(\lambda_{i}) \leq \sum_{i=1}^{2M}(\lambda_{i} - 1) 
     = \Tr(\delta\Sigma\cdot\Sigma_Q^{-1}).
\end{align}
Moreover, observe that 
\begin{align}
    e^{\frac{1}{2}\Tr(\delta\Sigma\cdot\Sigma_Q^{-1})} \leq \frac{1}{1-\frac{1}{2}\Tr(\delta\Sigma\cdot\Sigma_Q^{-1})}.
\end{align}
provided that $\Tr(\delta\Sigma\cdot\Sigma_Q^{-1}) < 2$.
Therefore, we have the lower bound
\begin{align}
    \frac{1}{\sqrt{\det(\mathbb{I}_{2M}+\delta\Sigma\cdot\Sigma_Q^{-1})}} &\geq e^{-\frac{1}{2}\Tr(\delta\Sigma\cdot\Sigma_Q^{-1})} 
    \geq 1-\frac{1}{2}\Tr(\delta\Sigma\cdot\Sigma_Q^{-1}) 
    \geq 1 - \frac{1}{2}\| \delta\Sigma\|_{F} \| \Sigma_{Q}^{-1}\|_{F}, \label{eq: lower bound on inv det: cauchy schwarz}
\end{align}
where the inequality in Eq.~\eqref{eq: lower bound on inv det: cauchy schwarz} comes from the the Cauchy–Schwarz inequality.

We next examine the upper bound of Eq.~\eqref{eq: determinant ratio}.
By similarly denoting $\{\lambda_i\}_{i=1}^{2M}$ as a set of eigenvalues of $\mathbb{I}_{2M} + \delta\Sigma\cdot\Sigma_Q^{-1}$, we have 
\begin{align}\label{eq: lower bound on log det}
     \log\det(\mathbb{I}_{2M} + \delta\Sigma\cdot\Sigma_Q^{-1})= \sum_{i=1}^{2M}\log(\lambda_{i}) \geq \sum_{i=1}^{2M} \left( 1 - \frac{1}{\lambda_{i}}\right) = \Tr(\mathbb{I}_{2M} - (\mathbb{I}_{2M} + \delta\Sigma\cdot\Sigma_Q^{-1})^{-1}) ,
\end{align}
provided that $\mathbb{I}_{2M} + \delta\Sigma\cdot\Sigma_{Q}^{-1}$ is invertible. 
To obtain the lower bound of the right-hand side in Eq.~\eqref{eq: lower bound on log det}, consider a series expansion
\begin{align}
    \mathbb{I}_{2M} - (\mathbb{I}_{2M} + \delta\Sigma\cdot\Sigma_Q^{-1})^{-1} = \delta\Sigma\cdot\Sigma_Q^{-1} - (\delta\Sigma\cdot\Sigma_Q^{-1})^2 + (\delta\Sigma\cdot\Sigma_Q^{-1})^3 - \cdots.
\end{align}
Then, given that $\left| \Tr(\delta\Sigma\cdot\Sigma_Q^{-1})\right| < 1$, we have
\begin{align}
    \Tr(\mathbb{I}_{2M} - (\mathbb{I}_{2M} + \delta\Sigma\cdot\Sigma_Q^{-1})^{-1})  
    &= \Tr(\delta\Sigma\cdot\Sigma_Q^{-1}) -\Tr((\delta\Sigma\cdot\Sigma_Q^{-1})^2) + \Tr((\delta\Sigma\cdot\Sigma_Q^{-1})^3) - \cdots \\
    &\geq -\left|\Tr(\delta\Sigma\cdot\Sigma_Q^{-1})\right| - \left| \Tr((\delta\Sigma\cdot\Sigma_Q^{-1})^2)\right| - \left|\Tr((\delta\Sigma\cdot\Sigma_Q^{-1})^3)\right| - \cdots \\
    &\geq -\left|\Tr(\delta\Sigma\cdot\Sigma_Q^{-1})\right| - \left| \Tr(\delta\Sigma\cdot\Sigma_Q^{-1})\right|^2 - \left|\Tr(\delta\Sigma\cdot\Sigma_Q^{-1})\right|^3 - \cdots \label{eq: trace of power vs power of trace}\\
    &= -\left|\Tr(\delta\Sigma\cdot\Sigma_Q^{-1})\right|\left( 1 + \left|\Tr(\delta\Sigma\cdot\Sigma_Q^{-1})\right| + \left| \Tr(\delta\Sigma\cdot\Sigma_Q^{-1})\right|^2 + \cdots \right) \\
    &= -\frac{\left| \Tr(\delta\Sigma\cdot\Sigma_Q^{-1})\right|}{1 - \left| \Tr(\delta\Sigma\cdot\Sigma_Q^{-1})\right|} ,
\end{align}
where the inequality in Eq.~\eqref{eq: trace of power vs power of trace} holds by Lemma~\ref{lemma: trace of power vs power of trace}, given that $\delta\Sigma$ and $\Sigma_{Q}$ are both positive (semi) definite.

\begin{lemma}\label{lemma: trace of power vs power of trace}

Let $A$ and $B$ be the positive semi-definite matrices. Then, for any positive integer $k \in \mathbb{Z}^{+}$, the following inequality holds. 
\begin{align}
    \Tr((AB)^k) \leq \Tr(AB)^k .
\end{align}
\end{lemma}
\begin{proof}
    See Sec.~\ref{appendix: section: proof for trace power}.
\end{proof}
To proceed, observe that 
\begin{align}
    e^{\frac{1}{2} \Tr(\mathbb{I}_{2M} - (\mathbb{I}_{2M} + \delta\Sigma\cdot\Sigma_Q^{-1})^{-1})} &\geq 1 + \frac{1}{2}\Tr(\mathbb{I}_{2M} - (\mathbb{I}_{2M} + \delta\Sigma\cdot\Sigma_Q^{-1})^{-1}) \\
    &\geq 1 -\frac{1}{2}\frac{\left| \Tr(\delta\Sigma\cdot\Sigma_Q^{-1})\right|}{1 - \left| \Tr(\delta\Sigma\cdot\Sigma_Q^{-1})\right|} \\
    &= \frac{1 - \frac{3}{2}\left| \Tr(\delta\Sigma\cdot\Sigma_Q^{-1})\right|}{1 - \left| \Tr(\delta\Sigma\cdot\Sigma_Q^{-1})\right|}.
\end{align}
Hence, we obtain the upper bound
\begin{align}
    \frac{1}{\sqrt{\det(\mathbb{I}_{2M}+\delta\Sigma\cdot\Sigma_Q^{-1})}} 
    &\leq e^{-\frac{1}{2} \Tr(\mathbb{I}_{2M} - (\mathbb{I}_{2M} + \delta\Sigma\cdot\Sigma_Q^{-1})^{-1})} \\
    &\leq \frac{1 - \left| \Tr(\delta\Sigma\cdot\Sigma_Q^{-1})\right|}{1 - \frac{3}{2} \left| \Tr(\delta\Sigma\cdot\Sigma_Q^{-1})\right|}  \\
    &\leq 1 +  \left| \Tr(\delta\Sigma\cdot\Sigma_Q^{-1})\right| \\
    &\leq 1 + \| \delta\Sigma\|_{F} \| \Sigma_{Q}^{-1}\|_{F}, \label{eq: J22}
\end{align}
where the third inequality holds provided that $\left| \Tr(\delta\Sigma\cdot\Sigma_Q^{-1})\right|  \leq 1/3$, and the last inequality comes from Cauchy-Schwarz inequality.
Therefore, combining Eq.~\eqref{eq: determinant ratio}, Eq.~\eqref{eq: lower bound on inv det: cauchy schwarz}, and Eq.~\eqref{eq: J22}, we finally obtain the bound
\begin{align}
     \left\vert 1 - \frac{\sqrt{|\Sigma_Q|}}{\sqrt{|\Sigma'_Q|}} \right\vert \leq \| \delta\Sigma\|_{F} \| \Sigma_{Q}^{-1}\|_{F},
\end{align}
concluding the proof.

\end{proof}

\subsection{Proof of Lemma~\ref{lemma: multiplicative error is bounded in worst case}}\label{appendix: section: bounding multiplicative error in worst case}

By utilizing the representation of the hafnian of the sum of matrices~\cite{efimov2021hafnian}, $\Haf(B+C)$ can be represented as 
\begin{align}
    \Haf(B+C) = \sum_{k=0}^{N}\sum_{\substack{J \subseteq [2N]\\|J|= 2N -2k}}\Haf(B_{J})\Haf(C_{\bar{J}}) = \Haf(B) + \sum_{k=1}^{N}\sum_{\substack{J \subseteq [2N]\\|J|= 2N -2k}}\Haf(B_{J})\Haf(C_{\bar{J}})   ,
\end{align}
where $B_{J}$ denotes the submatrix of $B$, with rows and columns taken according to the elements of $J \subseteq [2N]$ (similarly for $C_{\bar{J}}$, where $\bar{J}$ is a coset $[2N]\setminus J$).
Then we have
\begin{align}
    \frac{\Haf(B + C)}{\Haf(B)} =  1 + \sum_{k=1}^{N}\sum_{\substack{J \subseteq [2N]\\|J|= 2N -2k}}\frac{\Haf(B_{J})}{\Haf(B)}\Haf(C_{\bar{J}})    .  \label{eq: H19}
\end{align}
Recall that, $B$ and $C$ are given by 
\begin{align}
    B = \begin{pmatrix}
        W & 0 \\
        0 & W 
    \end{pmatrix} ,
    \quad C \coloneqq \frac{\Vert W \Vert_2}{\tanh r_0}\left(\mathbb{X}_{2M}\cdot\Sigma_{Q}^{-1}\cdot\delta\Sigma\cdot\Sigma_{Q}^{-1}\left(\mathbb{I}_{2M} + \delta\Sigma\cdot\Sigma_{Q}^{-1}\right)^{-1} \right)_{\bm{n}\oplus \bm{n}}   ,
\end{align}
where
\begin{align}
    W  = \begin{pmatrix}
        0_{\frac{ N}{2}} & W'\oplus I_{\frac{N}{2} - N_0} \\ (W')^{T}\oplus I_{\frac{N}{2} - N_0} & 0_{\frac{ N}{2}}
    \end{pmatrix}   ,
\end{align}
for a given $N_0\times N_0$ matrix $W' = \{-1,0,1\}^{N_0 \times N_0}$, and $N > 2N_0$.

To quantify the size of $\Haf(B+C)$, for convenience of analysis, we define the maximum absolute value along the elements of $C$ over all possible unitary matrices
\begin{align}
     C_{max} \coloneqq \max_{i,j,U}|C_{i,j}| .
\end{align}
Note that both of the matrices $\delta\Sigma$ and $\Sigma_{Q}$ composing the matrix $C$ implicitly depend on the unitary matrix $U$ we implement using the cluster state.

\begin{proof}[Proof of Lemma~\ref{lemma: multiplicative error is bounded in worst case}]

We first argue for the case that $\Per(W') \neq 0$, and later extend to the case $\Per(W') = 0$.
By definition, $\Haf(B) = \Per(W')^2 \geq 1$ always hold for $W' = \{-1,0,1\}^{N_0 \times N_0}$ unless $\Per(W') = 0$.
Also, from the fact that $B$ is a block matrix, we have 
\begin{align}\label{eq: H136}
   \sum_{k=1}^{N}\sum_{\substack{J \subseteq [2N]\\|J|= 2N -2k}}\Haf(B_{J})\Haf(C_{\bar{J}}) 
  = \sum_{k=1}^{N}\sum_{\substack{I, J \subseteq [N]\\ |I|,|J| \in 2\mathbb{Z}^{0+} \\ |I| + |J| = 2N -2k}}\Haf(W_{I})\Haf(W_{J})\Haf(C_{\bar{I}\oplus\bar{J}}) .
\end{align}
Here, as long as $|I| \geq N_0$ for $I \subseteq [N]$, the maximum value of $|\Haf(W_I)|$ is bounded by $N_0!$. 
Also, using the convention $C_{max} \coloneqq \max_{i,j,U}|C_{i,j}|$, we have $|\Haf(C_{\bar{I}\oplus\bar{J}})| \leq (2k-1)!!(C_{max})^{k}$ from the fact that $C_{\bar{I}\oplus\bar{J}}$ is $2k\times 2k$ matrix.
Lastly, for any $k \in [1, N]$, we have
\begin{align}\label{eq: number of possible config}
    \sum_{\substack{I, J \subseteq [N]\\ |I|,|J| \in 2\mathbb{Z}^{0+} \\ |I| + |J| = 2N -2k}} = \sum_{l = 0}^{\min\{k, N-k\}}\binom{N}{2l}\binom{N}{\min\{2k,2N-2k\} - 2l} = \frac{1}{2}\left[ \binom{2N}{2k} + (-1)^{k}\binom{N}{k}\right] \leq \binom{2N}{2k}   .
\end{align}
Combining all these arguments, we have 
\begin{align}
    \left\vert\sum_{k=1}^{N}\sum_{\substack{J \subseteq [2N]\\|J|= 2N -2k}}\frac{\Haf(B_{J})}{\Haf(B)}\Haf(C_{\bar{J}}) \right\vert 
    &\leq \sum_{k=1}^{N}\binom{2N}{2k}(N_0!)^2(2k-1)!!(C_{max})^k \\
    &\leq (N_0!)^2 \sum_{k=1}^{\infty}\binom{2N}{2k}(2k-1)!!(C_{max})^k \\
    &\leq (N_0!)^2 \sum_{k=1}^{\infty}\frac{(2N)^{2k}}{2^k k!}(C_{max})^k  \\
    &= (N_0!)^2 \left( e^{2N^2C_{max}}-1 \right) \\
    &\leq 3(N_0!)^2N^2C_{max} ,
\end{align}
where we assumed that $N^2C_{max} < 1/2$ in the last inequality. 
Here, the maximum value of $C_{max}$ can be bounded as 
\begin{align}
    \max_{i,j}|C_{i,j}| &\leq \frac{\Vert W \Vert_2}{\tanh r_0}\left\Vert \Sigma_{Q}^{-1}\cdot\delta\Sigma\cdot\Sigma_{Q}^{-1}\left(\mathbb{I}_{2M} + \delta\Sigma\cdot\Sigma_{Q}^{-1}\right)^{-1} \right\Vert_{F} \\
    &= \frac{\Vert W \Vert_2}{\tanh r_0}\left\Vert\Sigma_{Q}^{-1}\cdot\left( \delta\Sigma\cdot\Sigma_{Q}^{-1} -  \left(\delta\Sigma\cdot\Sigma_{Q}^{-1}\right)^2 + \left(\delta\Sigma\cdot\Sigma_{Q}^{-1}\right)^3 - \cdots\right)\right\Vert_{F} \\
     &\leq \frac{\Vert W \Vert_2}{\tanh r_0}\left\Vert \Sigma_{Q}^{-1}\right\Vert_{F} \left( \left\Vert  \delta\Sigma\cdot\Sigma_{Q}^{-1} \right\Vert_{F} + \left\Vert  \left(\delta\Sigma\cdot\Sigma_{Q}^{-1} \right)^2 \right\Vert_{F} +  \left\Vert  \left(\delta\Sigma\cdot\Sigma_{Q}^{-1} \right)^3 \right\Vert_{F} + \cdots \right) \\
     &\leq \frac{\Vert W \Vert_2}{\tanh r_0}\left\Vert \Sigma_{Q}^{-1}\right\Vert_{F} \left( \left\Vert  \delta\Sigma\right\Vert_{F} \left\Vert\Sigma_{Q}^{-1} \right\Vert_{F} + \left\Vert  \delta\Sigma\right\Vert_{F}^2 \left\Vert \Sigma_{Q}^{-1} \right\Vert_{F}^2 +  \left\Vert  \delta\Sigma\right\Vert_{F}^3 \left\Vert\Sigma_{Q}^{-1}  \right\Vert_{F}^3 + \cdots \right)  \\
     &\leq \frac{2\Vert W \Vert_2}{\tanh r_0}\left\Vert \Sigma_{Q}^{-1}\right\Vert_{F}^2 \left\Vert \delta\Sigma\right\Vert_{F} ,
\end{align}
where we invoked the sub-multiplicativity and sub-additivity of the Frobenius norm, further assuming $\left\Vert  \delta\Sigma\right\Vert_{F}\left\Vert\Sigma_{Q}^{-1} \right\Vert_{F} \leq 1/2$ for the last inequality.
Therefore, we can bound $C_{max}$ as
\begin{align}
     C_{max} \leq \frac{2\Vert W \Vert_2}{\tanh r_0}\max_{U} \left( \left\Vert \Sigma_{Q}^{-1}\right\Vert_{F}^2 \left\Vert \delta\Sigma\right\Vert_{F} \right) 
     \leq 
     \frac{2\Vert W \Vert_2}{\tanh r_0}  \max_{U} \left(\left\Vert \Sigma_{Q}^{-1}\right\Vert_{F}^2 \right) \max_{U} \left( \left\Vert \delta\Sigma\right\Vert_{F} \right)
     =
     \frac{2\Vert W \Vert_2}{\tanh r_0}  \left\Vert \Sigma_{Q}^{-1}\right\Vert_{F}^2 \max_{U} \left( \left\Vert \delta\Sigma\right\Vert_{F} \right), 
\end{align}
where the last equality holds by the fact that from Eq.~\eqref{eq: H42}, $ \left\Vert \Sigma_{Q}^{-1}\right\Vert_{F}$ is always given as 
\begin{align}
    \left\Vert \Sigma_{Q}^{-1}\right\Vert_{F} = \sqrt{\sum_{i,j = 1}^{2M} |( \Sigma_{Q}^{-1})_{i,j}|^2} = \sqrt{2M + 2\tanh^2 r_0\sum_{i,j = 1}^{M} |( U(\mathbb{I}_{N}\oplus 0_{M-N})U^{T})_{i,j}|^2} = \sqrt{2M+2N\tanh^2 r_0} ,
\end{align}
which is independent on $U$.

Now, we temporarily denote $\max_{U} \left( \left\Vert \delta\Sigma\right\Vert_{F} \right)$ as $\Vert \delta\Sigma \Vert_{M}$. 
Then, by Lemma~\ref{lemma: multiplicative coefficient is bounded}, given that $\left| \Tr(\delta\Sigma\cdot\Sigma_Q^{-1})\right| \leq \left\Vert  \delta\Sigma\right\Vert_{F}\left\Vert\Sigma_{Q}^{-1} \right\Vert_{F} \leq\Vert \delta\Sigma \Vert_{M} \Vert \Sigma_{Q}^{-1}\Vert_{F}  \leq 1/3$, we have
\begin{align}\label{eq: H60}
     \left\vert 1 - \frac{\sqrt{|\Sigma_Q|}}{\sqrt{|\Sigma'_Q|}} \right\vert \leq \| \Sigma_{Q}^{-1}\|_{F} \| \delta\Sigma\|_{F}  \leq  \Vert \Sigma_{Q}^{-1}\Vert_{F} \Vert \delta\Sigma \Vert_{M}  .
\end{align}
Also, we have
\begin{align}\label{eq: H61}
    \left\vert 1 - \frac{\Haf(B+C)}{\Haf(B)} \right\vert 
    = \left\vert\sum_{k=1}^{N}\sum_{\substack{J \subseteq [2N]\\|J|= 2N -2k}}\frac{\Haf(B_{J})}{\Haf(B)}\Haf(C_{\bar{J}}) \right\vert
    \leq 3(N_0!)^2N^2C_{max}
    \leq \frac{6(N_0!)^2N^2\Vert W \Vert_2}{\tanh r_0}  \left\Vert \Sigma_{Q}^{-1}\right\Vert_{F}^2  \Vert \delta\Sigma \Vert_{M} .
\end{align}
Combining Eq.~\eqref{eq: H60} and Eq.~\eqref{eq: H61}, we have 
\begin{align}
    \left\vert 1 - \frac{\sqrt{|\Sigma_Q|}}{\sqrt{|\Sigma'_Q|}} \frac{\Haf(B+C)}{\Haf(B)} \right\vert 
    &\leq 
    \Vert \Sigma_{Q}^{-1}\Vert_{F} \Vert \delta\Sigma \Vert_{M}   + \frac{6(N_0!)^2N^2\Vert W \Vert_2}{\tanh r_0}  \left\Vert \Sigma_{Q}^{-1}\right\Vert_{F}^2  \Vert \delta\Sigma \Vert_{M}  + \frac{6(N_0!)^2N^2\Vert W \Vert_2}{\tanh r_0}  \left\Vert \Sigma_{Q}^{-1}\right\Vert_{F}^3 \Vert \delta\Sigma \Vert_{M}^2 \\
    &\leq \frac{9(N_0!)^2N^2\Vert W \Vert_2}{\tanh r_0}  \left\Vert \Sigma_{Q}^{-1}\right\Vert_{F}^2 \Vert \delta\Sigma \Vert_{M},    
\end{align}
given that $\Vert \delta\Sigma \Vert_{M} \Vert \Sigma_{Q}^{-1}\Vert_{F}  \leq 1/3$.
Hence, to obtain the desired imprecision $\varepsilon' < 1$, we set the condition for $ \Vert \delta\Sigma \Vert_{M}$ that
\begin{align}\label{eq: condition 2}
    \Vert \delta\Sigma \Vert_{M} \leq
    \frac{\tanh r_0 \varepsilon' }{36(N_0!)^2N^2M \Vert W \Vert_2 }  \leq 
    \frac{\tanh r_0 \varepsilon' }{9(N_0!)^2N^2\Vert W \Vert_2 \left\Vert \Sigma_{Q}^{-1}\right\Vert_{F}^2  }  ,
\end{align}
where the last inequality holds by $ \left\Vert \Sigma_{Q}^{-1}\right\Vert_{F}^2 \leq 2M + 2N\tanh^2 r_0 \leq 4M$.
By this condition combined with $\varepsilon' < 1$, our assumptions $N^2C_{max} < 1/2 $ and $\Vert \delta\Sigma \Vert_{M} \Vert \Sigma_{Q}^{-1}\Vert_{F}  \leq 1/3$ are automatically satisfied, thereby obtaining the first bound for the case $\Per(W') \neq 0$.

We now argue the case that $\Haf(B) = \Per(W')^2 = 0$. 
By following the previous arguments, we have 
\begin{align}
    \left\vert  \Haf(B + C) \right\vert 
    &= \left\vert \sum_{k=1}^{N}\sum_{\substack{J \subseteq [2N]\\|J|= 2N -2k}}\Haf(B_{J})\Haf(C_{\bar{J}}) \right\vert   
    \leq  \sum_{k=1}^{N}\binom{2N}{2k}(N_0!)^2(2k-1)!!(C_{max})^k \\
    &\leq 3(N_0!)^2N^2C_{max}  
    \leq \frac{6(N_0!)^2N^2\Vert W \Vert_2}{\tanh r_0}  \left\Vert \Sigma_{Q}^{-1}\right\Vert_{F}^2  \Vert \delta\Sigma \Vert_{M} 
    \leq \frac{2}{3}\varepsilon', 
\end{align}
where in the last inequality we assumed the condition given in Eq.~\eqref{eq: condition 2}.
Then, we have 
\begin{align}
    \left\vert \frac{\sqrt{|\Sigma_Q|}}{\sqrt{|\Sigma'_Q|}}  \Haf(B + C) \right\vert 
    \leq \frac{2}{3}\frac{\sqrt{|\Sigma_Q|}}{\sqrt{|\Sigma'_Q|}}\varepsilon' 
    \leq \frac{2}{3}(1 + \Vert \delta\Sigma \Vert_{M} \Vert \Sigma_{Q}^{-1}\Vert_{F})\varepsilon' 
    \leq \varepsilon' ,
\end{align}
where we used Lemma~\ref{lemma: multiplicative coefficient is bounded} in the third inequality and $\Vert \delta\Sigma \Vert_{M} \Vert \Sigma_{Q}^{-1}\Vert_{F}  \leq 1/3$ in the last inequality. 
This gives the second bound for the case that $\Haf(B) = \Per(W')^2 = 0$, and thus, we conclude the proof.

\end{proof}

\subsection{Proof of Lemma~\ref{lemma: trace of power vs power of trace}}\label{appendix: section: proof for trace power}

\begin{proof}[Proof of Lemma~\ref{lemma: trace of power vs power of trace}]
Given that both $A$ and $B$ are positive semi-definite and symmetric matrices, we can decompose $A = ODO^T$ and $B = V\Lambda V^T$ for orthogonal matrices $O$ and $V$, and diagonal matrices $D$ and $\Lambda$ whose elements are all positive. 
Using these conventions, we can also define the square root of the matrices $A^{\frac{1}{2}} = OD^{\frac{1}{2}}O^T $ and $B^{\frac{1}{2}} = V\Lambda^{\frac{1}{2}} V^T$
Then, note that a matrix $A^{\frac{1}{2}}BA^{\frac{1}{2}}$ is a positive semi-definite and symmetric matrix, because 
\begin{align}
    A^{\frac{1}{2}}BA^{\frac{1}{2}} = OD^{\frac{1}{2}}O^TV\Lambda V^TOD^{\frac{1}{2}}O^T = \left(OD^{\frac{1}{2}}O^TV\Lambda^{\frac{1}{2}}\right)\left(OD^{\frac{1}{2}}O^TV\Lambda^{\frac{1}{2}}\right)^{T}.
\end{align}

Now, let us denote $\{\lambda_i\}$ as a set of all eigenvalues of $A^{\frac{1}{2}}BA^{\frac{1}{2}}$, which are all real and non-negative due to the positive semi-definiteness of $A^{\frac{1}{2}}BA^{\frac{1}{2}}$. 
Then, for an arbitrary positive integer $k$, we have 
\begin{align}
     \Tr((AB)^k) &= \Tr( \left(A^{\frac{1}{2}}BA^{\frac{1}{2}} \right)^k) 
     = \sum_{i}\lambda_i^k
     \leq \left(  \sum_{i} \lambda_i \right)^k \\
     &= \Tr(A^{\frac{1}{2}}BA^{\frac{1}{2}})^k = \Tr(AB)^{k}   , 
\end{align}
concluding the proof. 

\end{proof}

\section{Achieving universal brickwork cluster state from regular cluster states}

We now detail the physical procedure for compiling the universal brickwork cluster state $\ket{G_{U}}$ (Definition~\ref{def: brickwork graph and angles}) from standard regular cluster-state lattices.
Fig.~\ref{fig: grid to universal} illustrates the overall procedure by which the universal brickwork cluster state $\ket{G_{U}}$ can be obtained from grid and hexagonal lattices via the graph-transformation operations known as \textit{vertex deletion} and \textit{wire shortening}~\cite{gu2009quantum, dahlberg2018transforming, ghosh2023complexity}.
Specifically, performing a homodyne measurement in the $\hat{q}$-basis on the $i$th mode of a cluster state (red circle in Fig.~\ref{fig: grid to universal}), corresponding to a vertex $v_{i}$ in the underlying graph $G$, removes that vertex along with all edges incident to $v_{i}$, resulting in a cluster state whose underlying graph is $G \setminus \{v_i\}$. 
Also, for a one-dimensional chain graph, performing a homodyne measurement in the $\hat{p}$-basis (blue circle in Fig.~\ref{fig: grid to universal}) on an interior vertex shortens the chain by removing that vertex and directly connecting its two neighboring vertices.

\begin{figure}[t]
\centering
\includegraphics[width=\linewidth]{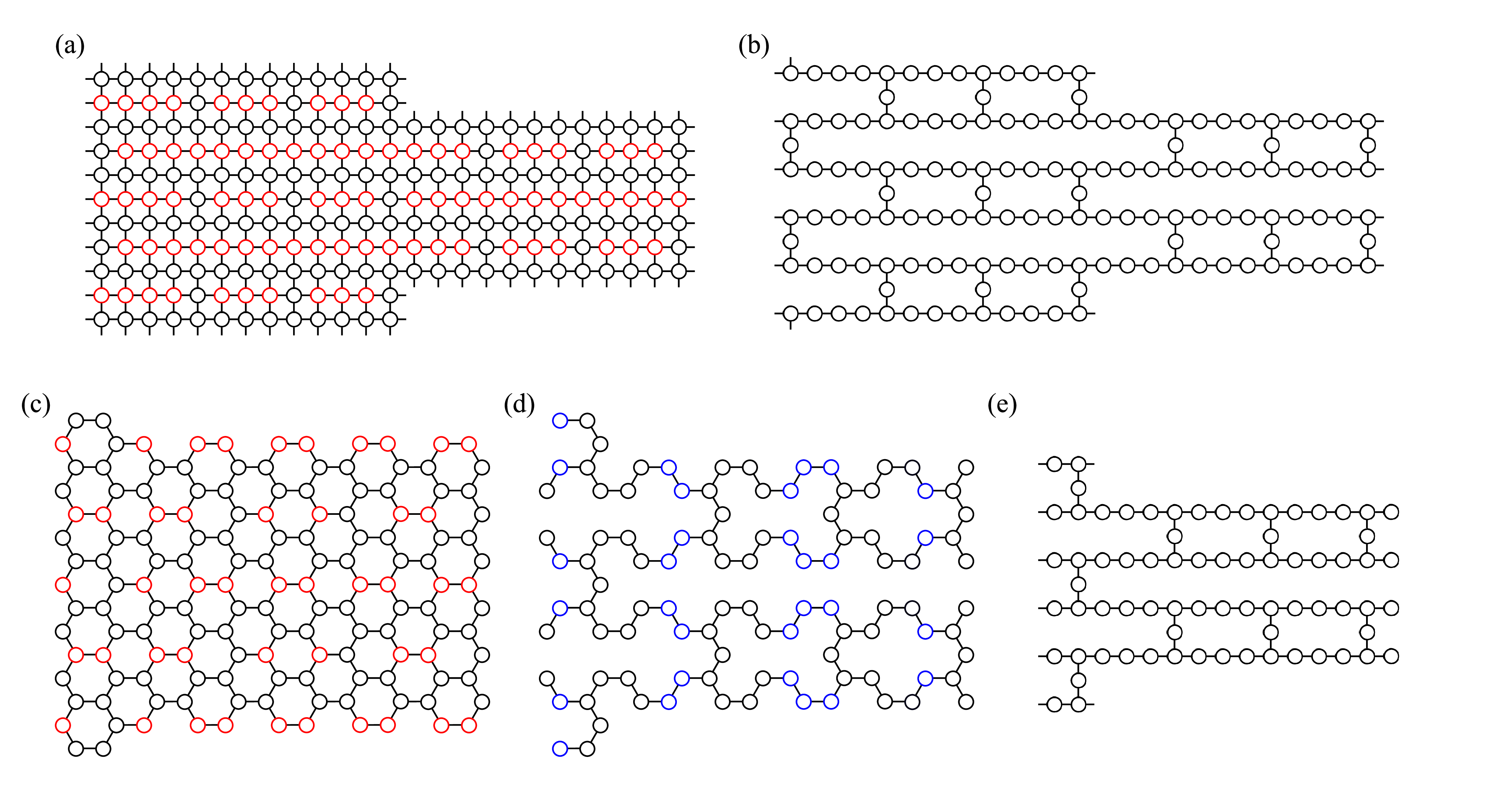}
\caption{Schematic of constructing the universal brickwork cluster state $\ket{G_{U}}$ from regular lattice cluster states via graph-transformation operations.
(a) Grid lattice. 
By applying vertex deletion to the red-highlighted nodes, which removes the vertex along with all incident edges, the universal brickwork graph $G_{U}$ shown in (b) is obtained.
(c) Hexagonal lattice.
Applying vertex deletion to the red-highlighted nodes yields the intermediate graph shown in (d).
Subsequently, applying wire shortening to the blue-highlighted nodes in (d) finally produces the universal brickwork graph $G_{U}$ illustrated in (e). 
}
\label{fig: grid to universal}
\end{figure}

\bibliographystyle{unsrt}
\bibliography{Reference.bib}